\documentclass[10pt,a4paper,twocolumn]{revtex4}
\usepackage[latin1]{inputenc}
\usepackage{amsmath}
\usepackage{amsfonts}
\usepackage{amssymb}
\usepackage{amsthm}
\usepackage{graphicx}

\newcommand{\Real}{{\mathrm{Re}}}

\newcommand{\Tr}{{\mathrm{Tr}}} 
\newcommand{\Sp}{{\mathrm{Sp}}}

\newtheorem{Proposition}{Proposition}
\newtheorem{Lemma}{Lemma}
\newtheorem{Corollary}{Corollary}

\newtheorem{Definition}{Definition}

\begin{document}
\title{Catalytic Coherence}
\author{Johan {\AA}berg}
\email{johan.aberg@physik.uni-freiburg.de}
\affiliation{Institute for Physics, University of Freiburg, Hermann-Herder-Strasse 3, D-79104 Freiburg, Germany}

\begin{abstract}
Due to conservation of energy we cannot directly turn a quantum system with a definite energy into a superposition of different energies. However, if we have access to an additional resource in terms of a system with a high degree of coherence, as for standard models of laser light, we can overcome this limitation.  
The question is to what extent coherence gets degraded when utilized. 
Here it is shown that coherence can be turned into a catalyst, meaning that we can use it repeatedly without ever diminishing its power to enable coherent operations. This finding stands in contrast to the degradation of other quantum resources, and has direct consequences for quantum thermodynamics, as it shows that latent energy that may be locked into superpositions of energy eigenstates can be released catalytically.
\end{abstract}

\maketitle

{\it Introduction.--} Coherence is a resource that enables us to implement coherent operations on quantum systems; a canonical example being the use of lasers to put atoms in superposition between two energy levels \cite{MandelWolf}, or analogously for  radio-fields and nuclear spins \cite{Vandershypen04}. When we excite an atom we need another system, an `energy reservoir', where the energy is taken. What kind of energy reservoir would we need in order to put an atom in superposition between two energies? With a bit of thought one can realize that this is impossible if the atom and the reservoir initially have definite energies. (See Appendix \ref{EnergyConserveCoh} for details. This observation can be understood in the wider context of `reference frames' and symmetry preserving operations \cite{Bartlett07,Gour08,Marvian11}.) One way to resolve the apparent contradiction with the above claim, that such superpositions indeed can be generated, is to realize that this usually is achieved via, e.g., lasers or radio-fields. These are often modeled as  coherent states, typically described as superpositions of the energy eigenstates of the electromagnetic field \cite{Glauber63a,Glauber63b,MandelWolf} (although this can be debated, see e.g.~\cite{Molmer97,Banacloche98,Molmer98,Rudolph01,Enk02,Bartlett05}). In other words, the coherence of the laser is a resource that enables us to put the atom in superposition, or more generally to perform operations that coherently mix energies. 

The main result of this investigation is that coherence can be turned into a catalyst, in the sense that it enables otherwise impossible tasks, without itself being consumed. Although maybe reminiscent of entanglement catalysis \cite{Jonathan99}, this stands in contrast to other quantum resources, e.g., reference frames for measurements, which appear to degrade upon use \cite{Bartlett06,Bartlett07b,Poulin07,Bartlett07}.
We establish catalytic coherence within two models. The first model (the doubly-infinite energy ladder) is convenient to analyze, but is somewhat unphysical in that its Hamiltonian has no ground state. The second model (the half-infinite ladder) amends this problem. We furthermore numerically investigate remnants of catalytic coherence in the Jaynes-Cummings model \cite{JC63,Shore93}. Finally, we apply catalytic coherence to `work extraction' in the context of quantum thermodynamics.

 \begin{figure}[h!]
 \includegraphics[width= 8cm]{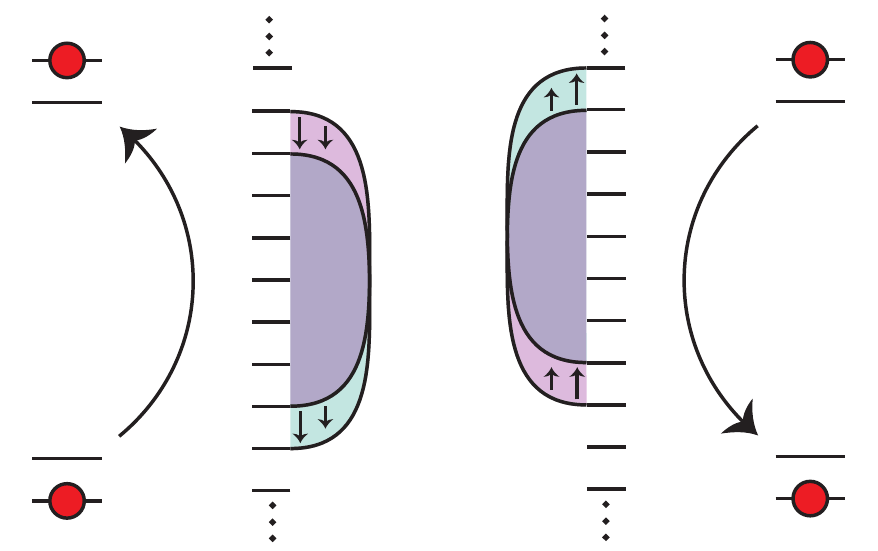} 
   \caption{\label{FigureCoherent} {\bf Rigid translation property.}
     A two-level system interacts with a reservoir whose energy levels can be described as a ladder. By design, the interaction shifts the whole state of the reservoir `rigidly' down or up along this ladder, depending on whether the two-level system absorbs (left part) or donates (right part) energy. For example, by removal of one quantum of energy, the uniform superposition  $|\eta_{L,l_0}\rangle = \sum_{l=0}^{L-1}| l_0 + l \rangle/\sqrt{L}$ is shifted to $|\eta_{L,l_0-1}\rangle$. If $L$ is large, the difference between  $|\eta_{L,l_0-1}\rangle$ and $|\eta_{L,l_0}\rangle$ is small. Hence, the  state of the reservoir does not change much by the loss or gain of a single quantum, which enables coherent operations on the two-level system. In the limit of large $L$, these implementations can be made perfect. This is analogous to how coherent states on a bosonic mode can implement coherent operations on an atom via the Jaynes-Cummings model.
   }
\end{figure}

{\it The doubly-infinite energy ladder.--}
Catalytic coherence is achieved via a specific design of the interaction between the system and the energy reservoir. While this construction can be used to generate general coherent operations on $N$-level systems (see Appendix \ref{DoublyInfinite}) we here focus on two-level systems. 
Let $S$ be a  system for which the Hamiltonian $H_S$ has eigenvalues $h_0=0$ and $h_1 = s>0$,  corresponding to the eigenstates $|\psi_0\rangle$ and $|\psi_1\rangle$.  The Hamiltonian of the reservoir $E$ is $H_E = s\sum_{j\in\mathbb{Z}}j|j\rangle\langle j|$, where $s$ is the energy spacing in the ladder, $\{|j\rangle\}_{j\in\mathbb{Z}}$ is an orthonormal basis, and $\mathbb{Z}$ denotes the set of  integers.  Regarded as a Hamiltonian, $H_E$ is slightly odd in that it does not have any ground state. We shall remedy this problem shortly, but due to several convenient properties we shall use this model for the initial analysis. As a first step we define the `shift-operator' $\Delta = \sum_{j\in\mathbb{Z}}|j+1\rangle\langle j|$. As one can see, this unitary operator translates every state `rigidly' along the energy ladder (see Fig.~\ref{FigureCoherent}). With the aid of $\Delta$ we define the following family of unitary operators on $\mathcal{H}_S\otimes\mathcal{H}_E$
\begin{eqnarray}
\nonumber V(U) & = & \sum_{n,n'=0,1} |\psi_n\rangle\langle\psi_n|U|\psi_{n'}\rangle\langle\psi_{n'}|\otimes \Delta^{n'-n}\\
\label{nkfjgbn} & = & \sum_{j\in\mathbb{Z}}V_j(U),\\
\nonumber V_j(U) & = &  \sum_{n,n'=0,1}|\psi_{n}\rangle\langle\psi_{n}|Q|\psi_{n'}\rangle\langle\psi_{n'}|\otimes|j-n\rangle\langle j-n'|,
\end{eqnarray}
where $U$ is an arbitrary unitary operator on $\mathcal{H}_S$. 
(This type of interaction has previously been considered in quantum thermodynamics \cite{Skrzypczyk13}.) By construction, all $V(U)$ commute with $H_S + H_E$, i.e., they are energy conserving. Furthermore they commute with all $\Delta^{a}$, and thus act uniformly over the energy ladder (see \cite{Skrzypczyk13} for discussions on this). The family of all $V(U)$ serves as the set of `allowed operations' in our model. In the following we shall investigate what kind of transformations we can implement on $S$, and how this depends on the coherence in the reservoir $E$.

{\it Coherence and coherent operations.--}
Given a state $\sigma$ on the reservoir we can implement channels on $S$ via
\begin{equation}
\label{Phidef}
\Phi_{\sigma,U}(\rho)  = \Tr_E [V(U)\rho\otimes\sigma V(U)^{\dagger}].
\end{equation}
(This generally requires time-control, see Appendix \ref{Sec:inducedch}.)
We let $\mathcal{C}(\sigma)$ denote the set of channels $\Phi_{\sigma,U}$ that can be obtained for arbitrary unitary $U$, given $\sigma$. Let us now see what aspects of $\sigma$ it is that determine $\Phi_{\sigma,U}$. If one inserts Eq.~(\ref{nkfjgbn}) into (\ref{Phidef}), it turns out that $\Phi_{\sigma,U}$, and thus $\mathcal{C}(\sigma)$, depends on $\sigma$ only via expectation values of the form $\Tr(\Delta^{a}\sigma)$ for $a\in\mathbb{Z}$ (which do not determine $\sigma$ uniquely).

Next, suppose we wish to perform a unitary operation that mixes different energy levels. By Eqs.~(\ref{nkfjgbn}) and (\ref{Phidef}) one can see that if $\Tr(\Delta^{a}\sigma)\approx 1$ for $a =-2,\ldots, 2$, then $\Phi_{\sigma,U}(\rho)\approx U\rho U^{\dagger}$. (For general $N$-level systems, the necessary range of $a$ is determined by the amount of energy that the reservoir would need to donate or absorb.) 
Hence, when we speak of a `high degree of coherence', it means that the state $\sigma$ of the reservoir is such that $\Tr(\Delta^{a}\sigma)\approx 1$ for a broad range of $a$, with the rationale that this allows us to perform coherent operations to a good approximation. 

A concrete example is the family  of states $\sigma = |\eta_{L,l_0}\rangle\langle\eta_{L,l_0}|$ of  the form $|\eta_{L,l_0}\rangle = \sum_{l=0}^{L-1}|l_0 + l\rangle/\sqrt{L}$, i.e., uniform superpositions over a collection of consecutive energy eigenstates. 
In the limit of large $L$, the channel $\Phi_{|\eta_{L,l_0}\rangle\langle\eta_{L,l_0}|,U}$ converges to the unitary operation $\rho\mapsto U\rho U^{\dagger}$ (as one could expect from the Wigner-Arkai-Yanase theorem \cite{Wigner52,Arkai60,Yanase61,Ghirardi81a,Ghirardi81b,Ozawa02a,Busch11,Loveridge11,Marvian12,Ahmadi13,Ozawa02b,Karasawa07,Karasawa09,Navascues14,Barnes99,Enk01,GeaBanacloche02,GeaBanacloche05}). Hence, this model is powerful enough to implement all unitary operations on $S$, given a sufficient degree  of coherence in the reservoir. Next we show that the degree of coherence does not change over repeated applications.

{\it Catalytic coherence in the doubly-infinite ladder.--}
To investigate the catalytic properties of the coherence, we need to determine how the state of the energy reservoir changes when we use it. For that purpose we define the corresponding channel on $E$
\begin{equation}
\label{Lambdadef}
\Lambda_{\rho,U}(\sigma) = \Tr_S [V(U)\rho\otimes\sigma V(U)^{\dagger}].
\end{equation}
By using Eq.~(\ref{nkfjgbn}) one can confirm that
\begin{equation}
\Tr[\Delta^{a}\Lambda_{\rho,U}(\sigma)] = \Tr(\Delta^{a}\sigma),
\end{equation}
for all $\sigma$, $\rho$, $U$, and $a$. In other words, the expectation values $\langle \Delta^a\rangle =  \Tr(\Delta^{a}\sigma)$ are invariants under the action of these operations. 

Earlier we noted that it is precisely the expectation values $\langle \Delta^a\rangle$ that determine which channels that can be implemented. Hence, if we use the reservoir a second time, we can implement the very same channels that we implemented the first time, i.e., $\Phi_{\Lambda(\sigma),U} = \Phi_{\sigma,U}$, and hence $\mathcal{C}\boldsymbol{(}\Lambda_{\rho,U}(\sigma)\boldsymbol{)} = \mathcal{C}(\sigma)$. In other words, we do not degrade the coherence resource in the reservoir by using it. In this sense, coherence is catalytic in this model. One can also prove a stronger type of catalytic property, which does not assume that the reservoir and the systems initially are uncorrelated, see Appendix \ref{DoublyInfinite}. (See also Appendix \ref{TranslationPhaseRef} for a reformulation of catalytic coherence in terms of correlations with a reference system.)  

Note that the  catalytic property holds for all states $\sigma$ on the reservoir, and is \emph{not} limited to states with a high degree of coherence. It should also be emphasized that although $\langle \Delta^{a}\rangle$ are invariants, the underlying \emph{state} does change (see Appendix \ref{DoublyInfinite} for an example), which is in contrast to entanglement catalysis \cite{Jonathan99}. In other words, analogously to how measurements induce back-action on reference frames  \cite{Aharonov84,Aharonov98,Casher00}, there is indeed a back-action on the energy reservoir. However, as opposed to how certain reference frames appear to degrade due to back action \cite{Bartlett06,Bartlett07b,Poulin07,Bartlett07}, this change of state does not affect the usefulness of the coherence in the reservoir.

{\it The half-infinite ladder.--}
One might worry that the catalytic property is an anomaly related to the lack of ground state, as a broadening distribution otherwise would hit the bottom. In the following we shall show that the capacity to induce channels can be maintained indefinitely, also in a model that has a proper ground state. To this end, we cut away the lower half of the doubly-infinite ladder and thus obtain (the spectrum of) the harmonic oscillator $H^{+}_E = s\sum_{j=0}^{+\infty}j|j\rangle\langle j|$. 
We define a new class of unitary operations on $S$ and $E$ as
\begin{equation}
V_{+}(U) =  |\psi_{0}\rangle\langle\psi_0|\otimes|0\rangle\langle 0|+ \sum_{l=1}^{+\infty}V_{l}(U),
\end{equation}
with $V_l$ as in equation (\ref{nkfjgbn}).
By comparison one can see that $V(U)$ and $V_+(U)$ act identically on all states with at least one quantum, i.e, $\langle l|\sigma|l\rangle = 0$ for $l= 0$. (For a general $S$ this `border zone' would be larger, see Appendix \ref{Harmonic}.) 

 \begin{figure}[h!]
 \includegraphics[width= 8cm]{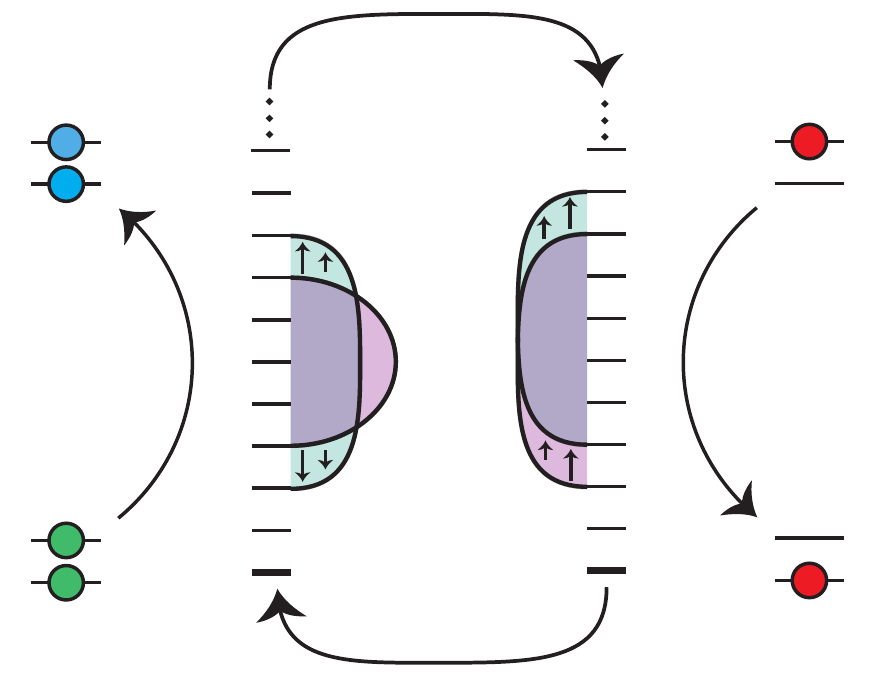} 
   \caption{\label{FigureCycle} {\bf Regenerative catalytic cycles.} 
  When the energy reservoir interacts with the two-level system (left part) the projection of its state onto the number states can increase with at most one level up and down in energy. 
  As long as the projection onto the ground state of the reservoir remains zero, the set of channels that the reservoir can induce on the system stays intact from one interaction to the next. 
    By using another two-level system in a pure excited state (right part) to inject energy into the reservoir, the state is translated `rigidly' up along the energy ladder by one step. By alternating every use of the reservoir with such a pumping, the state of the reservoir can be kept away from the ground state, thus maintaining the coherence properties  indefinitely.
  Note that the state of the energy reservoir does change from one cycle to the next, e.g., the range of number states onto which it projects may become broader for each step. However, the \emph{relevant aspects} of the state, which determine its capacity to induce channels, remain constant.  
   }
\end{figure}

{\it Protocol for a catalytic half-infinite ladder.--}
A simple protocol can maintain the coherence properties of the reservoir indefinitely (see Fig.~\ref{FigureCycle}). We assume an initial state $\sigma_{\mathrm{in}}$ such that $\langle l|\sigma_{\mathrm{in}}|l\rangle = 0$ for $l= 0,1$, i.e., it contains at least two quanta. We let $E$  and $S$ interact via some arbitrary choice of unitary $V_{+}(U)$. As we know from the above reasoning, the effect is identical to $V(U)$. The new state $\overline{\sigma}$ on the energy reservoir is such that $\langle l|\overline{\sigma}|l\rangle = 0$ for $l = 0$. Now, consider an ancillary two-level system $A$, with the two eigenstates $|a_0\rangle$ and $|a_1\rangle$  corresponding to the energies $0$ and $s$, respectively. We assume that $A$ initially is in the excited state $|a_1\rangle$. By applying the operation $V_{+}(U_A)$ for $U_A = |a_0\rangle\langle a_1| + |a_1\rangle\langle a_0|$, the state $\overline{\sigma}$ is translated one rung up along the ladder, to a new state $\sigma_{\mathrm{out}}$. Hence, the reservoir is again in a state with at least two quanta. Since these operations all have been performed on states safely away from the ground state, it follows that $\mathcal{C}(\sigma_{\mathrm{out}}) = \mathcal{C}(\sigma_{\mathrm{in}})$. By iterating this procedure we can conclude that the set of channels that this reservoir can induce is kept intact indefinitely.

{\it Decay of coherence in the Jaynes-Cummings model.--}
For many theoretical purposes it is enough to know that coherence in principle can be made catalytic. It is nevertheless relevant to ask to what extent these phenomena exist in more general types of systems, especially if one would consider experimental investigations. Interactions between a two-level atom and a single mode of the electromagnetic field are often modeled via the Jaynes-Cummings (JC)  Hamiltonian $H_{\mathrm{JC}} =  g\sigma_{+}\otimes a + g\sigma_{-}\otimes a^{\dagger}$ \cite{JC63,Shore93}. Here $a,a^{\dagger}$ are the standard bosonic annihilation and creation operators $[a,a^{\dagger}] = \hat{1}_E$, and $\sigma_{+} = |\psi_1\rangle\langle \psi_0|$, $\sigma_{-} = |\psi_0\rangle\langle \psi_1|$. Similar to our designed interactions, the JC model also moves a quantum of energy between the atom and the reservoir (much as in Fig.~\ref{FigureCoherent}), but it does not act uniformly over the energy ladder. One can nevertheless find a `shadow' of catalytic coherence in the JC-model. The graphs in Fig.~\ref{fig:quasicatalytic} suggest that the capacity to repeatedly induce coherent operations decays slower for higher initial average energies, even if the `width' of the initial superposition is fixed (see Appendix \ref{MoreGeneral} for details).\newline

 \begin{figure}[h]
 \includegraphics[width= 8cm]{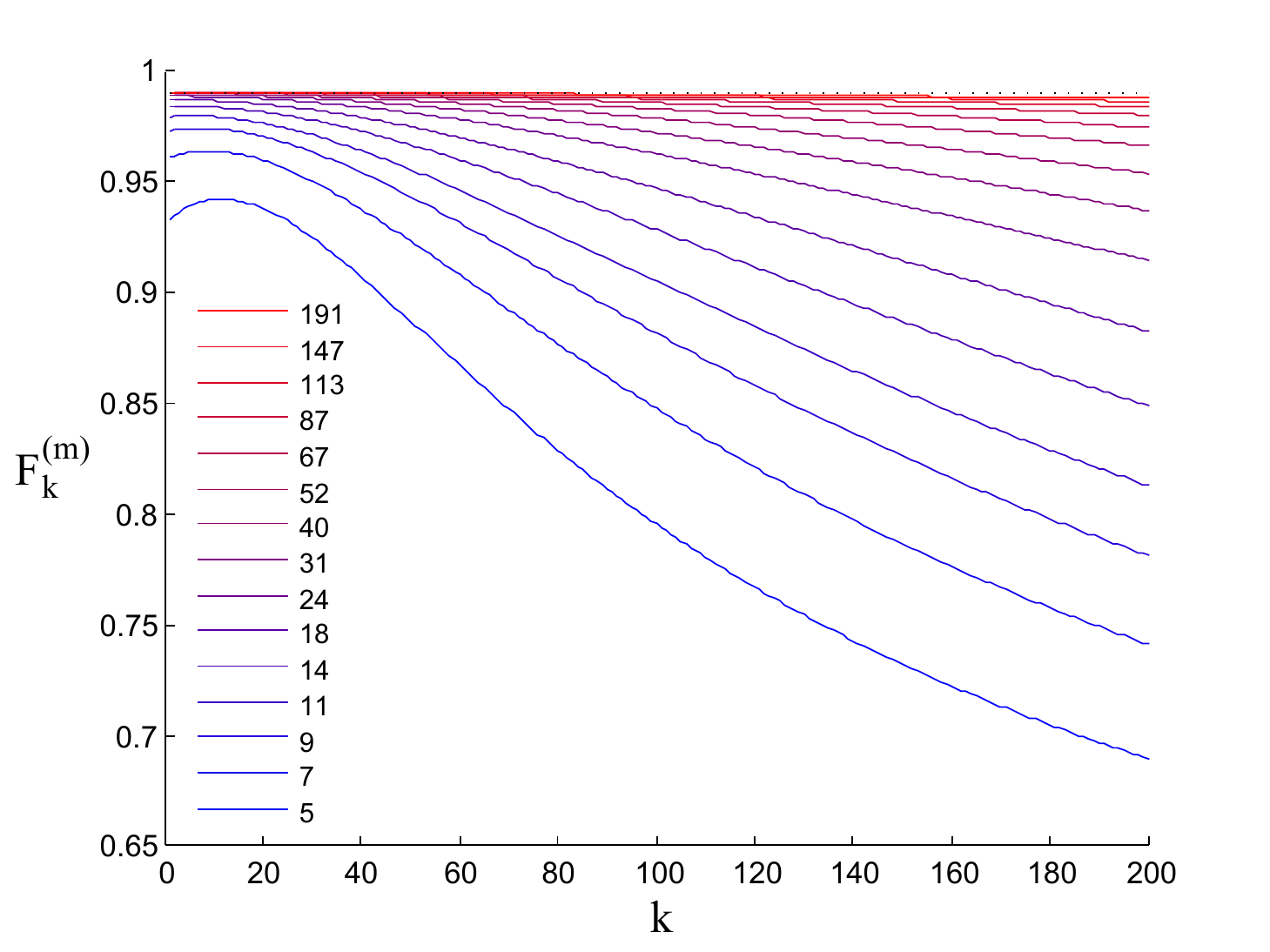} 
   \caption{\label{fig:quasicatalytic} {\bf  Decay of coherence in the Jaynes-Cummings model.} The JC-model exhibits a decay of coherence over repeated use. However, there is a remnant of the catalytic property, in the sense that the `lifetime' of the coherence (counted in number of iterations) appears to increase indefinitely merely by an increase of the average energy of the initial state.  These graphs depict the decay of the fidelity by which the superposition $|\phi\rangle = (|\psi_0\rangle -i|\psi_1\rangle)/\sqrt{2}$ can be created from the ground state $|\psi_0\rangle$, in a collection of two-level systems that sequentially interact with one single reservoir, according to the JC model for a fixed time-step. 
   Each curve corresponds to a different initial state of the reservoir, and shows $F^{(m)}_k = \langle\phi|\Phi(\sigma^{(k)}_m)|\phi\rangle$, against the number of times $k$ the reservoir has been used. Each curve corresponds to an initial state  $\sigma^{(0)}_{m} = |\eta_{L,l_0(m)}\rangle\langle \eta_{L,l_0(m)}|$ for $|\eta_{L,l_0(m)}\rangle$ with $l_{0}(m) = (4m+1)^{2} -24$ and $L= 50$, for $m = 5,7,9,11,14,18,24,31,40,52,67,87,113,147,191$. Note that the `width' of the superposition does not change with $m$. The dotted line is the value of $(1+|\langle \eta_{L,l_0(m)}|\Delta|\eta_{L,l_0(m)}\rangle|)/2  = 0.99$, which would be the fidelity reached in the doubly-infinite ladder-model for these initial states.
   }
\end{figure}

{\it Application: Coherence in expected work extraction.--} 
Work extraction and the closely related concepts of information erasure and Maxwell's demon have a long history (see e.g., \cite{Szilard,Landauer61,Procaccia,Benett03,LeffRexI,LeffRexII,Takara,Esposito06,Maroney09,Reeb13b}) with recently revived interests, e.g., in the contexts of resource theories \cite{Janzing00,Brandao11,Horodecki12,Gour13} and single-shot statistical mechanics \cite{Dahlsten,delRio,TrulyWorkLike,Horodecki11,Egloff,EgloffThesis,Faist,Brandao13}.  
The task is  to extract as much useful energy as possible by equilibrating a system with Hamiltonian $H_S$ and state $\rho$ with respect to a heat bath of temperature $T$.  Here we consider the question of how much work that can be extracted in an average sense.   
Standard results \cite{Procaccia,Lindblad1983,Takara,Esposito06} suggest that the expected work content of a system is characterized by
\begin{equation}
\label{nskdfljv}
\mathcal{A}_{\textrm{`standard'}}(\rho,H_S) = kTD\boldsymbol{(}\rho\Vert G(H_S)\boldsymbol{)},
\end{equation}
where $D(\rho\Vert \eta) = \Tr(\rho\ln\rho) -\Tr(\rho\ln\eta)$ is the relative von Neumann entropy, $k$ is Boltzmann's constant, $G(H_S) = \exp(-\beta H_S)/Z(H_S)$, $Z(H_S) = \Tr\exp(-\beta H_S)$, and $\beta = 1/(kT)$. 
(Equation (\ref{nskdfljv}) can be confirmed in a model without an explicit energy reservoir, see Appendix \ref{WithoutReservoir}.)

However, in light of the above considerations, such approaches appear to  implicitly assume access to ideal coherence resources. 
Recently it has been shown \cite{Skrzypczyk13} that without any coherence, the optimal expected work content is not characterized by (\ref{nskdfljv}), but rather by  
\begin{equation}
\label{kgdfngnkmg}
\mathcal{A}_{\textrm{diagonal}}(\rho,H_S) = kTD\boldsymbol{(}[\rho]_{H_S}\Vert G(H_S)\boldsymbol{)},
\end{equation}
where $[\rho]_{H_S} = \sum_{l}P_{l}\rho P_{l}$, and $P_{l}$ are the projectors onto the eigenspaces of $H_S$. 
However, one can show (see Appendix \ref{ExplicitEnergyReservoirs}) that access to coherence increases the amount of work that can be extracted. Furthermore, in the limit of a large degree of coherence (e.g. large $L$ for $|\eta_{L,l_0}\rangle$), one regains equation (\ref{nskdfljv}).  We can conclude that coherence sails up as an important resource alongside the expected work content. The question is how they relate. How much of one resource can be gained by spending the other? The fact that we here perform the work-extraction analysis entirely within the doubly-infinite ladder model implies that we only use the coherence catalytically and do not `spend' it at all. 

In relation to this thermodynamic application one may note  \cite{Skrzypczyk13b}, where optimal extraction is obtained irrespective of the state of the reservoir, via a larger class of unitary operations  that only conserve energy on average, and has  the power to create and destroy superpositions of energy eigenstates (see Appendix \ref{OnAverage}). Note also \cite{Gelbwaser13}, which uses a coherent extraction device as a negentropy source, to demonstrate a transient efficiency that exceeds the standard Carnot bound for work extraction against a hot and a cold heat bath.

{\it Conclusions and outlook.--}
We have shown that coherence can be turned into a catalytic resource by a specific design of interactions, and used this to analyze work extraction. As observed in the single-shot setting \cite{Dahlsten,delRio,TrulyWorkLike,Horodecki11,Egloff,EgloffThesis,Faist,Brandao13}, the expected work content may not always correspond to ordered `work-like' energy (see discussions in \cite{TrulyWorkLike}). It is an open question if catalytic coherence is associated to a cost of ordered energy (see Appendix \ref{Discussion}). Another question is to quantify expected work extraction with limited access to coherence.

The existence of catalytic coherence raises the question whether other types of resources \cite{Bartlett06,Bartlett07b,Poulin07,Bartlett07} in some sense also can be turned catalytic. In view of the analogy between embezzling states \cite{vanDam03} and coherent states, one can speculate whether there exists some counterpart to catalytic coherence in that setting (see Appendix \ref{Discussion}).  One can also ask whether catalytic coherence and entanglement catalysis \cite{Jonathan99} are special cases of a more general class of catalytic phenomena (see Appendix \ref{Discussion}).  On a more general level, the question is under what conditions, and in what sense, a resource can be made catalytic.

%%%%%%%%%%%%%%%%%%%%%%%%%%%%%%%%%%%
%		ACKNOWLEDGEMENT	%
%%%%%%%%%%%%%%%%%%%%%%%%%%%%%%%%%%%

\acknowledgements
The author thanks L\'idia del Rio, Philippe Faist, and Paul Skrzypczyk for useful discussions, and Charles H. Bennett for useful comments on embezzling states. This research was supported by the Excellence Initiative of the German Federal and State Governments (grant ZUK 43).

%%%%%%%%%%%%%%%%%%%%%%%%%%%%%%%%%%%%%%%%%%%%%%%%%%%
%%%%%%%%%%%%%%%%%%%%%%%%%%%%%%%%%%%%%%%%%%%%%%%%%%%
%%%%%%%%%%%%%%%%%%%%%%%%%%%%%%%%%%%%%%%%%%%%%%%%%%%
%%%%					Appendix							%%%%%%%%%%%%%
%%%%%%%%%%%%%%%%%%%%%%%%%%%%%%%%%%%%%%%%%%%%%%%%%%%
%%%%%%%%%%%%%%%%%%%%%%%%%%%%%%%%%%%%%%%%%%%%%%%%%%%
%%%%%%%%%%%%%%%%%%%%%%%%%%%%%%%%%%%%%%%%%%%%%%%%%%%

\begin{widetext}
\end{widetext}

\begin{appendix}

\section{\label{EnergyConserveCoh}Energy conservation and coherence}
In this investigation we regard unitary operators as `energy conserving' if and only if they commute with the total Hamiltonian. The resulting evolution is such that if the system initially is within an energy eigenspace it remains there. One can view this as a quantum counterpart to classical evolution that is restricted to an  energy shell (although the latter usually refers to a narrow range of energies, while we here assume a precise value).

\subsection{\label{Impossible} Creation of superpositions of energy eigenstates}

In this section we justify the claim that one cannot create a superposition between energy eigenstates from an initial state with a well defined energy, if we are limited to energy conserving unitary operations. (One should not confuse this observation with the fact that one in some sense can simulate the effects of superposition under these conditions. See Section \ref{TranslationPhaseRef}.) For more general discussions on reference frames and symmetry preserving operations, see \cite{Bartlett07,Gour08,Marvian11}.

We let $H_S$ and $H_E$ be Hermitian operators on finite-dimensional Hilbert spaces, corresponding to the Hamiltonians on systems $S$ and $E$.
Let $H^{S} = \sum_{k}h^{S}_kP^{S}_l$ be the eigenvalue decomposition where $h^{S}_k \neq h^{S}_{k'}$ for $k\neq k'$. Analogously $H_E = \sum_{l}h^{E}_lP^{E}_l$, with $h^{E}_l \neq h^{E}_{l'}$ for $l\neq l'$.

The joint Hamiltonian  $H := H_{S}\otimes\hat{1}_E + \hat{1}_S\otimes H_{E}$ may have degeneracies due to `energy matchings' of the two separate Hamiltonians. For each energy $x$ we let $\Omega(x) = \{(k,l): h^{S}_k + h^{E}_{l} = x\}$. ($\Omega(x)$ is the empty set if $x$ is not in the spectrum.) We can write the projector onto the eigenspace corresponding to $x$ as 
\begin{equation}
P_{x} := \sum_{(k,l)\in\Omega(x)}P^S_k\otimes P^E_l,
\end{equation}
which is the zero operator if $x$ is not in the spectrum.

Suppose $\eta\in\mathcal{S}(\mathcal{H}_S\otimes\mathcal{H}_E)$ has a definite energy, i.e., $P_x\eta P_x = \eta$ for some $x$. It follows that the reduced states $\eta_S := \Tr_E\eta$ and $\eta_E :=\Tr_S\eta$, on $S$ and $E$, respectively, must be block diagonal with respect to the local energy eigenspaces. To see this, we first note that
\begin{equation}
\begin{split}
\eta_{S}  = & \sum_{(k,l),(k',l')\in\Omega(x)}\Tr_E([P^S_k\otimes P^E_l]\eta[P^S_{k'}\otimes P^E_{l'}])\\
 = & \sum_{(k,l),(k',l)\in\Omega(x)}P^S_k\Tr_E([\hat{1}_S\otimes P^E_l]\eta)P^S_{k'}.
\end{split}
\end{equation}
Next, we observe that $(k,l),(k',l)\in\Omega(x)$ means that $h^{S}_k + h^{E}_{l} = x$ and $h^S_{k'} + h^E_{l} = x$. This implies that $h^S_k = h^S_{k'}$, which by construction implies that $k = k'$. Hence 
\begin{equation}
\eta_{S} =  \sum_{(k,l)\in\Omega(x)}P^S_k\Tr_E([\hat{1}_S\otimes P^E_l]\eta)P^S_{k},
\end{equation}
which means that $\eta_S$ is block-diagonal with respect to the energy eigenspaces of $H_S$. We can make the analogous reasoning for $\eta_E$.

If a unitary operator $V$  is energy conserving, it means that it is block-diagonal with respect to the eigenspaces of $H$, i.e., $V = \sum_x P_x VP_x$, or equivalently $VP_x = P_xV$.

Suppose that the states $\rho_S\in\mathcal{S}(\mathcal{H}_S)$ and $\rho_E \in\mathcal{S}(\mathcal{H}_E)$ each have definite energies, i.e., $P^S_{n}\rho_SP^S_{n} = \rho_S$ for some $n$, and $P^E_{m}\rho_EP^E_{m} = \rho_E$ for some $m$. Hence $\rho_S\otimes\rho_E$ has support on the energy eigenspace corresponding to $x = h^S_n+ h^E_m$. (We do not really need product states. It is enough that joint state has a definite energy.) If $V$ is an energy conserving unitary, it follows that $VP_x = P_xV$. Hence, $\eta :=V\rho_S\otimes\rho_EV^{\dagger}$ is such that $P_x\eta P_x = \eta$. By the above reasonings, we can conclude that the marginal states $\eta_S$ and $\eta_E$ must be block-diagonal.

\subsection{\label{OnAverage} An alternative: Energy preservation on average}
As mentioned above, we here associate energy conservation with unitary operators that commute with the Hamiltonian of
the total system. One can imagine an alternative that only requires energy to be conserved on average, i.e., that the
expectation value of the energy after the operation is the same as before the operation. 
However, if this should hold
for all possible initial states $\rho$, one merely regains a unitary $U$ that commutes with the Hamiltonian $H$, since
\begin{equation}
\begin{split}
 & \Tr(HU\rho U^{\dagger}) = \Tr(H\rho),\quad \forall
\rho\in \mathcal{S}(\mathcal{H}) \\
 & \quad \Leftrightarrow \quad [U,H] =0, 
\end{split}
\end{equation}
where $\mathcal{S}(\mathcal{H})$ denotes the set of density operators on $\mathcal{H}$. In other words, for a unitary operator that does not commute with $H$, there always exists an initial state for which the average energy is not conserved. 

This type of on-average energy conserving unitary operations may also have the power to create and destroy superpositions between distinct energy eigenstates, without access to any additional resources. This can be illustrated by the following example. Suppose that the Hamiltonian $H$ has three eigenvalues $-1,0,1$, with corresponding eigenstates $|\psi_{-1}\rangle, |\psi_0\rangle, |\psi_{1}\rangle$. As one can see, all
states of the form
\begin{equation}
|\nu\rangle = c_1\frac{1}{\sqrt{2}}(|\psi_{-1}\rangle +
e^{i\chi}|\psi_1\rangle) + c_0|\psi_{0}\rangle,
\end{equation}
with $|c_1|^2 + |c_0|^2 = 1$ and $\chi \in\mathbb{R}$, are such that $\langle \nu|H|\nu\rangle =0$. Hence, any unitary operator $U$ such that $U|\psi_0\rangle = (|\psi_{-1}\rangle + |\psi_1\rangle)/\sqrt{2}$ would be energy conserving on average with respect to the initial state $|\psi_0\rangle$. Furthermore, this operation creates the superposition $(|\psi_{-1}\rangle + |\psi_1\rangle)/\sqrt{2}$ of the two distinct energy eigenstates $|\psi_{-1}\rangle, |\psi_1\rangle$, from the initial diagonal state  $|\psi_{0}\rangle$. The conjugate $U^{\dagger}$ similarly destroys the superposition.
Hence, this class of operations has the power to create or destroy superpositions between energy eigenstates without additional resources.

This example illustrates how basic assumptions can affect the nature of coherence as a resource. Although interesting, we will not consider this question further, but rather focus on unitary operators that commute with the Hamiltonian.

\section{\label{DoublyInfinite}Doubly infinite energy ladder}
As a first step to investigate the catalytic properties of coherence we here consider a simple model for energy conserving operations on a combination of a system and an energy reservoir. This type of model has previously been considered in \cite{Skrzypczyk13}.

\subsection{\label{Sec:DoublInfModel}The doubly-infinite ladder model}
Our model consists of two parts: A `system' $S$, modeled by a finite-dimensional complex Hilbert space $\mathcal{H}_S$, and an `energy reservoir' $E$, modeled by a separable Hilbert space $\mathcal{H}_E$. We let $\{|j\rangle\}_{j\in\mathbb{Z}}$ denote a complete orthonormal basis of $\mathcal{H}_E$.
Here $\mathbb{Z}$ denote the set of (negative as well as positive) integers. Given a real number $s>0$ we define the following unbounded Hermitian operator
\begin{equation}
\label{DefHEreservoir}
H_E^{s}:= s\sum_{j\in\mathbb{Z}}j|j\rangle\langle j|.
\end{equation}
We will in the following regard this operator as the Hamiltonian of the energy reservoir $E$, although, needless to say, it is somewhat unphysical, since it has a `bottomless' spectrum. However, we will remedy this problem in Section \ref{Harmonic}. (\cite{Skrzypczyk13} mostly uses a continuum reservoir rather than a discrete model. However, see the remarks in Appendix G of \cite{Skrzypczyk13}.)

Given a finite-dimensional Hilbert space $\mathcal{H}_S$ and a real number $s >0$ we let $H_{s}(\mathcal{H}_S)$ denote the set of Hermitian operators on $\mathcal{H}_S$ such that all its eigenvalues are integer multiples of $s$, i.e., all its eigenvalues can be written as $h_n = s z_n$ where $z_n\in\mathbb{Z}$. The reason for this construction is to `match' all energy differences in $H_S$ with the energy spacings in $H_E^{(s)}$. By this we make sure that for each pair of eigenstates in $S$ there is a transition (in fact infinitely many) in $E$ that can compensate for the change in energy. 
This allows for non-trivial energy-conserving unitary operations. (If $H_S + H_E$ is non-degenerate, we can only implement unitary operators that are diagonal in the joint eigenbasis. On $S$ this leads to mere convex combinations of unitary operations that are diagonal in the unique eigenbasis of $H_S$.)

As a bit of a side-remark one may note that it ultimately may be desirable to develop a generalization which allows for small perturbations of the perfect energy conservation, i.e., allowing evolution within an energy shell, as it often is done in statistical mechanics. However, we will not consider such generalizations in this investigation, but proceed with the theoretically convenient restriction to perfect energy conservation.

As seen by the following lemma, the eigenspaces of the joint Hamiltonian $H_S + H_E^{(s)}$ (we should strictly speaking write $H_S\otimes \hat{1}_E + \hat{1}_S\otimes H_E^{(s)}$) have a very simple structure.
\begin{Lemma}
\label{eigenspaces}
Let $s>0$ be given. Let $\mathcal{H}_S$ be a finite-dimensional Hilbert space with $\dim\mathcal{H}_S=N$,  $H_S\in H_{s}(\mathcal{H}_S)$, with eigenvectors $|\psi_n\rangle$ and corresponding eigenvalues $s z_n$ (where $z_n\in\mathbb{Z}$). Then
\begin{equation}
\begin{split}
H_S + H_E^{(s)}  = & s \sum_{j\in\mathbb{Z}}jP^{(j)},\\
 P^{(j)}  := &  \sum_{n=1}^{N}|\psi_{n}\rangle\langle\psi_{n}|\otimes |j-z_n \rangle\langle j-z_{n}| 
\end{split}
\end{equation}
is an eigenvalue decomposition, where $P^{(j)}$ is the projector onto the eigenspace corresponding to the eigenvalue $js$. 
\end{Lemma}

The whole purpose of the construction with the energy reservoir is to keep track of the energy. Hence, for the dynamics on the joint system $SE$ we consider unitary operations that conserve the total energy, i.e, that commute with the joint Hamiltonian $H_S + H_{E}^{(s)}$. However, we furthermore demand that the dynamics should act `uniformly' over the energy levels in the reservoir. To express this a bit more precisely we introduce a  unitary `shift-operator' $\Delta$ on $\mathcal{H}_E$, which can be thought of as a `rigid translation' along the doubly-infinite energy ladder, 
\begin{equation}
\Delta := \sum_{k\in\mathbb{Z}}|k+1\rangle \langle k|.
\end{equation}
One can see that 
\begin{equation*}
\Delta^a\Delta^b = \Delta^b\Delta^a = \Delta^{a+b},\quad  {\Delta^{a}}^{\dagger} =\Delta^{-a},\quad a,b\in\mathbb{Z},
\end{equation*}
where we let $\Delta^{0} = \hat{1}_E$.
\begin{Definition}
Let $s>0$, and let $\mathcal{H}_S$ be a finite-dimensional Hilbert space and let $H_S\in H_{s}(\mathcal{H}_S)$. We use the following notation:
\begin{itemize}
\item $L(\mathcal{H}_S)$ is the set of linear operators on $\mathcal{H}_S$, and $U(\mathcal{H}_S)$ the set of unitary operators on $\mathcal{H}_S$.
\item $\mathbb{LT}(H_S + H_E^{(s)})$ is the set of bounded linear operators $X$ on $\mathcal{H}_S\otimes\mathcal{H}_{E}$ such that 
$[X,H_S + H_E^{(s)}] = 0$,  and $[X,\hat{1}_S\otimes\Delta^a] = 0$ for all $a\in \mathbb{Z}$.
\item $\mathbb{UT}(H_S + H_E^{(s)})$ is the set of unitary operators $U$ on $\mathcal{H}_S\otimes\mathcal{H}_{E}$ such that 
$[U,H_S + H_E^{(s)}] = 0$,  and $[U,\hat{1}_S\otimes\Delta^a] = 0$ for all $a\in \mathbb{Z}$.
\item  Given a Hilbert space $\mathcal{H}$ we let $\mathcal{S}(\mathcal{H})$ denote the set of positive semi-definite trace class operators with trace $1$ on $\mathcal{H}$, which we refer to as density operators or `states'. We let $\mathcal{S}_{+}(\mathcal{H})$ denote the set of density operators with full support.
\end{itemize}
\end{Definition}

The set $\mathbb{UT}(H_S + H_E^{(s)})$ is precisely the set of evolution operators we allow in our model; they conserve energy, and commute with all energy translations $\Delta^{a}$. The following lemma shows that we can describe the set $\mathbb{UT}(H_S + H_E^{(s)})$ in a very convenient manner.
\begin{Lemma}
\label{nvmxcnv}
Let $\mathcal{H}_S$ be finite-dimensional with $N:=\dim\mathcal{H}_S$, and let $H_S\in H_{s}(\mathcal{H}_S)$. Let $|\psi_{n}\rangle$ be orthonormal eigenvectors of $H_S$ with corresponding eigenvalues $s z_n$. Define the mapping
\begin{equation}
\label{bijection}
V(Q) := \sum_{n,n'=1}^{N}|\psi_{n}\rangle\langle\psi_{n}|Q|\psi_{n'}\rangle\langle\psi_{n'}|\otimes \Delta^{z_{n'} -z_{n}},
 \end{equation}
 for all $Q$ in  $L(H_S)$.
The mapping $V$ is a bijection between $L(\mathcal{H}_S)$ and $\mathbb{LT}(H_S + H_E^{(s)})$, which preserves the   structures of these operator spaces, in the sense that
\begin{equation}
\label{preservation}
\begin{split}
 & V(\alpha A + \beta B) = \alpha V(A) + \beta V(B),\\
 & V(AB) = V(A)V(B),\\
 & V(A^{\dagger}) = V(A)^{\dagger},\\
 & V(\hat{1}_S) = \hat{1}_{S}\otimes\hat{1}_E,
\end{split}
 \end{equation}
 for all $A,B\in L(H_S)$ and all $\alpha,\beta\in\mathbb{C}$.
  As a consequence, $V$ is a bijection between $U(\mathcal{H}_S)$ and $\mathbb{UT}(H_S + H_E^{(s)})$.
\end{Lemma}
Equation (\ref{bijection}) provides us with a very convenient handle on $\mathbb{UT}(H_S + H_E^{(s)})$, since we can reach each element of $\mathbb{UT}(H_S + H_E^{(s)})$ uniquely via $V(U)$ for some $U\in U(\mathcal{H}_S)$. Due to this we will in this investigation hardly ever refer explicitly to the set $\mathbb{UT}(H_S + H_E^{(s)})$, but rather describe it indirectly via $V$.

It is worth noting that one can rewrite the mapping $V$ as
\begin{equation}
\label{nfkjnb}
\begin{split}
V(Q)  = & \sum_{j\in\mathbb{Z}}V_{j}(Q),\\
 V_{j}(Q)  := & \sum_{n,n'=1}^{N}|\psi_{n}\rangle\langle\psi_{n}|Q|\psi_{n'}\rangle\langle\psi_{n'}|\otimes|j-z_n\rangle\langle j-z_{n'}|,
\end{split}
\end{equation}
where $V_{j}(Q) = P^{(j)}V_{j}(Q) = V_{j}(Q)P^{(j)}$.
Hence, what $V$ essentially does is to embed an infinite number of `copies' of $Q$ into the space of operators on $\mathcal{H}_S\otimes\mathcal{H}_E$. Moreover, if $U\in U(\mathcal{H}_S)$ then 
$V_{j}(U)V_j(U)^{\dagger} = V_j(U)^{\dagger} V_{j}(U) = P^{(j)}$. Hence, each $V_{j}(U)$ is unitary on the subspace onto which $P^{(j)}$ projects.

\begin{proof}[Proof of Lemma \ref{nvmxcnv}.]
Since $V(Q)$ merely is a finite sum of bounded operators, it follows that $V$ maps $L(\mathcal{H}_S)$ into the set of bounded linear operators on $\mathcal{H}_S\otimes\mathcal{H}_E$. Since all $\Delta^{a}$ commutes, it follows that $[V(Q),\hat{1}_{S}\otimes\Delta^{a}]=0$ for all $a\in\mathbb{Z}$. One can also  verify (e.g. via showing $[H^{(s)}_E,\Delta^a] = sa\Delta^{a}$) that $[V(Q),H_S + H_E^{(s)}]=0$.  Hence $V(L(\mathcal{H}_S))\subseteq \mathbb{LT}(H_S + H_E^{(s)})$. It  furthermore only requires a verifications to see that the conditions in (\ref{preservation}) hold.  
It thus only remains to show that $V$ is surjective [i.e., $V(L(\mathcal{H}_S)) = \mathbb{LT}(H_S + H_E^{(s)})$] and injective (i.e., two different elements in $L(\mathcal{H}_S)$ cannot get mapped to the same element).
The injectivity can be obtained if we apply (\ref{bijection})  to $V(Q_1)= V(Q_2)$, use the orthonormality of $|\psi_n\rangle$, and note that $\Delta^{a}$ is never the zero operator. For the surjectivity we assume $Y \in \mathbb{LT}(H_S + H_E^{(s)})$, and note that $[Y,H_S + H_E^{(s)}] = 0$ implies that $Y = \sum_{j\in \mathbb{Z}}P^{(j)}YP^{(j)}$, with $P^{(j)}$ as in Lemma \ref{eigenspaces}.
Hence,
\begin{equation}
\begin{split}
Y  = & \sum_{j}\sum_{n,n'}   (\langle\psi_{n}|\langle j-z_{n}|)Y (|\psi_{n'}\rangle|j-z_{n'} \rangle)\\
& \times |\psi_n\rangle\langle\psi_{n'}|\otimes|j-z_n\rangle\langle j-z_{n'}|.
\end{split}
\end{equation}
Using $|j-z_{n'} \rangle = \Delta^{j}|-z_{n'}\rangle$, and $[\hat{1}_S\otimes{\Delta^{\dagger}}^{j}]Y[\hat{1}_S\otimes \Delta^{j}] = Y$ (since $[Y,\hat{1}_S\otimes \Delta^{j}]=0$) we can conclude that 
\begin{equation}
X := \sum_{n,n'} (\langle\psi_{n}|\langle -z_{n}|)Y (|\psi_{n'}\rangle|-z_{n'} \rangle)|\psi_n\rangle\langle\psi_{n'}|,
\end{equation}
is such that $Y = V(X)$, which proves the Lemma. Combining the bijectivity of $V$ with (\ref{preservation}) yields that $V$ also is a bijection between $U(\mathcal{H}_S)$ and $\mathbb{UT}(H_S + H_E^{(s)})$.
\end{proof}

\subsection{\label{Hamiltonians} Concerning Hamiltonians}

As seen from the previous sections we let the Hamiltonians $H_S$ and $H_E$  characterize the energy in the systems. At first sight it may appear as if we have ignored the dynamics that these Hamiltonians induce. However, this can be understood in the sense of the interaction picture.

By inserting a Hermitian operator $H$ into (\ref{bijection}) one obtains a Hermitian operator $V(H)$ that, by construction,  commutes with $H_S + H_E$ (and all $\Delta^{a}$). Regarding $V(H)$ as an interaction  term we can construct the total Hamiltonian
\begin{equation}
H_{\mathrm{tot}} = H_S + H_E + V(H),
\end{equation}
 which on the combined system $SE$ induces the unitary evolution 
\begin{equation} 
U_{\mathrm{tot}}= e^{-itH_{\mathrm{tot}}} = e^{-it V(H)}e^{-it(H_S + H_E)},
\end{equation}
where the last equality is due to $[V(H),H_S + H_E]=0$. By shifting to the interaction picture, the evolution is thus described solely by the unitary operator $e^{-it V(H)}$, which is precisely a unitary operator on the form $V(U)$. Conversely, all $V(U)$ can be reached in this way by a suitable choice of $H$ and $t$.
In other words, we can view all the unitary operations $V(U)$ that appear in this investigation as describing an induced Hamiltonian evolution represented in the interaction picture.

\subsection{\label{Sec:inducedch}Induced channels}
Given the unitary operators $V(U)$ on the joint system $SE$ we can induce channels on $S$ and $E$. 

For a fixed state $\sigma$ on the energy reservoir each choice of unitary $U$ on $S$ induces a channel $\Phi^{S,H}_{\sigma,U}$ on system $S$ via
\begin{equation}
\label{channelOnS}
\begin{split}
\Phi^{S,H}_{\sigma,U}(\rho)  := &  \Tr_E V(U)\rho\otimes\sigma V(U)^{\dagger}\\
 = & \sum_{n,n',m,m'=1}^{N}\langle\psi_{n}|U|\psi_{n'}\rangle\langle\psi_{n'}|\rho |\psi_{m'}\rangle\langle\psi_{m'}|U^{\dagger}|\psi_{m}\rangle \\
 & \quad\quad\times |\psi_{n}\rangle\langle\psi_{m}|\Tr (\Delta^{z_{n'} -z_{n}-z_{m'} +z_{m}}\sigma).
\end{split} 
\end{equation}
We denote the set of channels that can be reached on $S$, by using $\sigma$ as a resource, by
\begin{equation}
\mathcal{C}^{S,H}(\sigma) := \Big\{\Phi^{S,H}_{\sigma,U}: U\in U(\mathcal{H}_S)\Big\}.
\end{equation}
As seen, both $\Phi^{S,H}_{\sigma,U}$ and $\mathcal{C}^{S,H}$  depend on the choice of Hamiltonian $H_S$, both via the eigenstates $|\psi_{n}\rangle$, and the eigenvalues $sz_n$. However, we will in the following  often drop one or both of the superscripts when they are not needed, i.e., $\Phi^{S,H}_{\sigma,U}$, $\Phi^{S}_{\sigma,U}$, $\Phi^{H}_{\sigma,U}$, and $\Phi_{\sigma,U}$ denote the same object. Similarly for $\mathcal{C}^{S,H}$, $\mathcal{C}^{S}$, and $\mathcal{C}$.

We can also express the effect of $V(U)$ on the energy reservoir via the channel 
\begin{equation}
\label{channelOnE}
\begin{split}
\Lambda_{\rho,U}(\sigma)  := & \Tr_{S}V(U)\rho\otimes \sigma V(U)^{\dagger}\\
  = & \sum_{n,n',m'=1}^{N}\langle\psi_{n}|U|\psi_{n'}\rangle\langle\psi_{n'}|\rho |\psi_{m'}\rangle\langle\psi_{m'}|U^{\dagger}|\psi_{n}\rangle\\
& \quad\quad\times\Delta^{z_{n'} -z_{n}}\sigma {\Delta^{z_{m'} -z_{n}}}^{\dagger}.
\end{split}
\end{equation} 

{\it The induced channels are unital.--}
One may note that $\Phi^{S,H}_{\sigma,U}(\hat{1}) = \hat{1}$ and $\Lambda_{\rho,U}(\hat{1}) = \hat{1}$, i.e., both $\Phi_{\sigma,U}$ and $\Lambda_{\rho,U}$ are unital (or mixing enhancing). It follows that the outputs of these channels cannot be less mixed than the inputs. (For unital channels and mixing, see e.g. \cite{Bapat,Chefles}.)
 At first sight this may seem to suggest that the set of operations that we can implement on $S$ would be rather limited, e.g., since we cannot make the state of the system more pure. However, one should keep in mind that one can append ancillary systems and implement channels $\Phi_{\sigma,U}$ on this extended system. (We employ this extensively in Sec.~\ref{ExplicitEnergyReservoirs}.) This increases the range of operations that can be implemented on $S$. For an example,  see  the end of Sec.~\ref{Sec:approxEnergymix}, where it is shown that all channels can be implemented on $S$ with the help of  ancillary systems and a high degree of coherence.

{\it The need for time-control.--} Needless to say, the implementation of operations $\Phi$ on $S$ requires time-control, i.e., we need access to a clock. 
Although we will not consider this issue here, a clock may itself be viewed as an additional resource \cite{Peres80,Buzek99,Buzek00,Janzing02,Janzing03a,Janzing03b} that should be accounted for.

In our case we would use a clock both for the purpose of implementing $V(U) = e^{-it V(H)}$ (with the Hamiltonian picture as described in Sec.~\ref{Hamiltonians}) as well as for controlling the effects of the intrinsic time-evolution caused by $H_S$ and $H_E$. (This observation is not particular for the setup we consider here, but would apply to many cases of non-trivial operations on systems with intrinsic dynamics.)

As an example, suppose the reservoir $E$ evolves for a time $\Delta t$ before we attach it to $S$ and perform the operation $\Phi$. The state of the reservoir would thus be characterized by $e^{-i\Delta t H_E}\sigma e^{i\Delta t H_E}$ rather than $\sigma$. Due to the fact that $V(U)$ commutes with $H_S +H_E$, it follows that 
\begin{equation*}
\begin{split}
& \Phi_{e^{-i\Delta t H_E}\sigma  e^{i\Delta t H_E},U}(\rho) \\
& =   e^{-i\Delta t H_S } \Phi_{\sigma,U}\bigg(e^{i\Delta t H_S} \rho e^{-i\Delta t H_S} \bigg)  e^{i\Delta t H_S}\\
& =    \Phi_{\sigma, e^{-i\Delta t H_S }Ue^{i\Delta t H_S} }(\rho).  
\end{split}
\end{equation*}
The effect of the shifted state on $E$ is thus to conjugate the original mapping on $S$ by the intrinsic evolution of $H_S$, which does not per se `destroy' the coherence properties of the implemented map. 
However, if $\Delta t$ would take a random value each time that the operation is implemented (cf. the `super-$\mathcal{G}$-twirling' in \cite{Gour08}), this would lead to  dephasing effects (where the details depend on the spectrum of $H_S$).

\subsection{\label{Sec:approxEnergymix}Approximation of energy-mixing unitary operations}

At first sight it may not be clear what kind of operations one can reach on $S$ in the doubly-infinite ladder-model. Here we show that in the limit of a high degree of coherence, all unitary operators can be implemented on $S$.  

In the following we shall see that the difference between the induced channel $\Phi_{\sigma,U}$ and the unitary channel 
\begin{equation}
\mathcal{U}_U(\rho) := U\rho U^{\dagger},
\end{equation}
becomes very small when the reservoir is sufficiently coherent.

For this purpose we do, on the set of super-operators $\Gamma:L(\mathcal{H}_S)\rightarrow L(\mathcal{H}_S)$,  consider the operator norm induced by the trace norm $\Vert \Gamma\Vert_{\mathrm{tr}}:= \sup_{\Vert X\Vert_1 =1}\Vert \Gamma(X)\Vert_1$.  By trivially extending the super-operator to a sufficiently large ancillary Hilbert space $\mathcal{H}_A$ (it is sufficient with $\dim\mathcal{H}_A = \dim\mathcal{H}_S$ \cite{Aharonov97}) one can define the diamond norm \cite{Kitaev97, Aharonov97} $\Vert \Gamma\Vert_{\diamond} := \Vert \Gamma\otimes I_A \Vert_{\mathrm{tr}}$, where $I_A$ is the identity map on $L(\mathcal{H}_A)$.

In the following, we let $\mathbb{C}^{n\times n}$ denote the set of $n\times n$ matrices with complex entries.
Furthermore, for $A,B\in \mathbb{C}^{n\times n}$ we let $A*B$ denote the Hadamard product, i.e. the component-wise product $(A*B)_{jk} := A_{jk}B_{jk}$. We furthermore let $\Vert \cdot\Vert$ denote the spectral norm $\Vert Q\Vert := \sup_{\Vert \psi\Vert = 1}\Vert Q\psi\Vert$.
For an arbitrary $Q\in \mathbb{C}^{N\times N}$ we let in the following $\boldsymbol{(}s_{l}(Q)\boldsymbol{)}_{l=1}^{N}$ denote the singular values of $Q$ (including zero values) ordered non-increasingly $s_1(Q)\geq s_2(Q)\geq \cdots \geq s_N(Q)$.  
\begin{Lemma}[Theorem 5.5.4 on p.~334 in \cite{Topics}]
Let $A,B\in\mathbb{C}^{N\times N}$ then 
\begin{equation}
\sum_{l=1}^{k}s_{l}(A*B) \leq \sum_{l=1}^{k}s_l(A)s_{l}(B),
\end{equation}
for $k= 1,\ldots, N$.
\end{Lemma}
Since $\Vert Q\Vert_1 = \sum_{l=1}^{N}s_{l}(Q)$ and $\Vert Q \Vert = s_1(Q)$ it follows directly from the above lemma that
\begin{equation}
\label{zuiozoui}
\Vert A* B\Vert_1 \leq \Vert A\Vert_1 \Vert B\Vert.
\end{equation}

\begin{Lemma}[Eq.~(3.7.2) on p.~223 in \cite{Topics}]
\label{abvkbja}
Let $C\in \mathbb{C}^{n\times n}$, then 
\begin{equation}
\Vert C \Vert \leq \sqrt{\max_{j}\sum_{k}|C_{jk}|} \sqrt{\max_{k}\sum_{j}|C_{jk}|}. 
\end{equation}
\end{Lemma}

\begin{Proposition}
\label{aksdjfbvakv}
Let $s>0$, $\mathcal{H}_S$ with $\dim\mathcal{H}_S= N$, and $H_S\in H_s(\mathcal{H}_S)$, with eigenvalues $sz_n$ and corresponding orthonormal eigenvectors $|\psi_n\rangle$. Let $\Phi_{\sigma,U}$ be as defined in equation (\ref{channelOnS}), and let $z_{\mathrm{max}} := \max_{n}z_n$ and $z_{\mathrm{min}} := \min_{n}z_n$. If $\sigma$ is such that 
\begin{equation}
\label{ljbnvaknk}
|1-\Tr(\Delta^{a}\sigma)| \leq \epsilon, \quad \forall a: |a|\leq 2(z_{\mathrm{max}}-z_{\mathrm{min}}), 
\end{equation}
then for very $U\in U(\mathcal{H}_S)$, it is the case that 
\begin{equation}
\Vert \mathcal{U}_U -\Phi_{\sigma,U}\Vert_{\mathrm{tr}} \leq  N^2\epsilon
\end{equation}
and
\begin{equation}
\Vert \mathcal{U}_U -\Phi_{\sigma,U}\Vert_{\diamond} \leq N^3\epsilon.
\end{equation}
\end{Proposition}
Equation (\ref{ljbnvaknk}) provides one possible formalization of what a `high degree of coherence'  should mean in this model, i.e., a wide range of values of $a$ for which $\Tr(\Delta^{a}\sigma)$ is close to $1$

\begin{proof}
Let $\mathcal{H}_A$ be a Hilbert space with orthonormal basis $\{|a_n\rangle\}_{n=1}^{M}$. (The choice of $M$ will in the end give us the induced norm ($M=1$) or the diamond norm ($M=N$).)
Let $X\in L(\mathcal{H}_S\otimes\mathcal{H}_A)$ be such that $\Vert X\Vert_{1}=1$.
For the sake of compactness we will use the notation $|\psi_k,a_n\rangle := |\psi_k\rangle|a_n\rangle$. 

By the triangle inequality for the trace norm
\begin{equation}
\label{fdbjknkdfb0}
\begin{split}
 & \bigg\Vert [U\otimes \hat{1}_A]X[U^{\dagger}\otimes \hat{1}_A]  - [\Phi_{\sigma,U}\otimes I_A](X) \bigg\Vert_1\\
 & \leq   \sum_{n,n'}\sum_{k,l}\Vert A^{(k,l,n,n')}*B^{(k,l)}\Vert_1,
\end{split}
\end{equation}
where `$*$' denotes the Hadamard product, and 
where the matrices $A^{(k,l,n,n')} := [A^{(k,l,n,n')}_{k',l'}]_{k',l'} $ and $B^{(k,l)}:= [B^{(k,l)}_{k',l'}]_{k',l'}$ are defined as
\begin{equation}
\begin{split}
A^{(k,l,n,n')}_{k',l'}  := & \langle\psi_{k'}|U|\psi_{k}\rangle\langle\psi_{k},a_n|X|\psi_{l},a_{n'}\rangle\langle\psi_l|U^{\dagger}|\psi_{l'}\rangle\\
 B^{(k,l)}_{k',l'}  := & 1-\Tr (\Delta^{z_{k} -z_{k'}-z_{l} +z_{l'}}\sigma).
\end{split}
\end{equation}
By Lemma \ref{abvkbja} it follows that  $\Vert B^{(k,l)}\Vert \leq  N\epsilon$.
One can confirm that $\Vert A^{(k,l,n,n')}\Vert_1=  |\langle\psi_{k},a_n|X|\psi_{l},a_{n'}\rangle|$.
By combining these observations with equations (\ref{zuiozoui}) and (\ref{fdbjknkdfb0}) we can conclude that 
\begin{equation}
\begin{split}
 & \bigg\Vert [U\otimes \hat{1}_A]X[U^{\dagger}\otimes \hat{1}_A]  - [\Phi_{\sigma,U}\otimes I_A](X) \bigg\Vert_1\\
 & \leq   N\epsilon\sum_{n,n'}\sum_{k,l} |\langle\psi_{k},a_n|X|\psi_{l},a_{n'}\rangle|\\
 & \leq  N^2M\epsilon,
 \end{split}
\end{equation}
where the last inequality follows by $\sum_{k,l,n,n'} |\langle\psi_{k},a_n|X|\psi_{l},a_{n'}\rangle|
\leq NM\Vert X\Vert_{1}$.
If we choose the ancillary Hilbert space to be trivial, i.e., $\dim\mathcal{H}_A = M = 1$, then we obtain the bound on the norm $\Vert \cdot\Vert_{\mathrm{tr}}$.  It is sufficient to choose the ancillary Hilbert space to have the same dimension as $\mathcal{H}_S$ to obtain the diamond norm \cite{Aharonov97}. In this case $M=N$ and we obtain the corresponding bound.
\end{proof}

For the sake of convenience we will often make use of a specific family of states for which the degree of coherence can be made arbitrarily large, and thus $\Phi_{\sigma,U}$ can be made arbitrarily close to $\mathcal{U}_{U}$.
Define the family of states
\begin{equation}
\label{etaLdef}
|\eta_{L,l_0}\rangle := \frac{1}{\sqrt{L}}\sum_{l=0}^{L-1}|l+l_0\rangle,
\end{equation}
i.e., this is a uniform superposition of $L\geq 1$ sequential energy eigenstates, starting at $l_0$.
One can verify that  
\begin{equation}
\label{etaShiftOverlap}
 \langle\eta_{L,l_0} | \Delta^{a}|\eta_{L,l_0}\rangle = \max\Big(0,1 -\frac{|a|}{L}\Big).
\end{equation}
In other words, the shifted $\Delta^a|\eta_{L,l_0}\rangle$ and the unshifted state $|\eta_{L,l_0}\rangle$ are almost indistinguishable if $L\gg a$.   This family of states can be viewed as a simplified version of coherent states, and the approximate implementation of unitary operations are analogous to the transitions in atoms induced by coherent electromagnetic fields in the Jaynes-Cummings model (see e.g. Sec.~\ref{JCmodel}, or Chapter 7 in \cite{NielsenChuang}).  

The following proposition shows that in the limit of large $L$ we can use the states $|\eta_{L,l_0}\rangle$ to implement arbitrary unitary operations on $S$.
\begin{Proposition}
Let $s>0$, $\mathcal{H}_S$ with $\dim\mathcal{H}_S= N$, and $H_S\in H_s(\mathcal{H}_S)$, with eigenvalues $sz_n$ and corresponding orthonormal eigenvectors $|\psi_n\rangle$. Let $\Phi_{\sigma,U}$ be as defined in equation (\ref{channelOnS}), and let $|\eta_{L,l_0}\rangle$ be as in equation (\ref{etaLdef}), for $l_0\in\mathbb{Z}$ and $L\in\mathbb{N}$. Let $z_{\mathrm{max}} := \max_{n}z_n$ and $z_{\mathrm{min}} := \min_{n}z_n$. If 
\begin{equation}
L\geq 2(z_{\mathrm{max}}-z_{\mathrm{min}}),
\end{equation}
then for very $U\in U(\mathcal{H}_S)$,
\begin{equation}
\Vert \mathcal{U}_U -\Phi_{|\eta_{L,l_0}\rangle\langle\eta_{L,l_0}|,U}\Vert_{\mathrm{tr}} \leq 2N^2\frac{z_{\mathrm{max}}-z_{\mathrm{min}}}{L}
\end{equation}
and
\begin{equation}
\Vert \mathcal{U}_U -\Phi_{|\eta_{L,l_0}\rangle\langle\eta_{L,l_0}|,U}\Vert_{\diamond} \leq 2N^3\frac{z_{\mathrm{max}}-z_{\mathrm{min}}}{L}.
\end{equation}
\end{Proposition}
The proof is almost identical to the proof of Proposition \ref{aksdjfbvakv}, with some minor differences. 
Due to $L \geq 2(z_{\mathrm{max}}-z_{\mathrm{min}})$ it follows that $|z_{k} -z_{k'}-z_{l}+z_{l'}| \leq L$.
Combined with equations (\ref{channelOnS}) and (\ref{etaShiftOverlap}) this yields the new matrix 
\begin{equation}
\begin{split}
 B^{(k,l)}_{k',l'}  := & \frac{|z_{k} -z_{k'}-z_{l}+z_{l'}|}{L}.
\end{split}
\end{equation}
and thus, by Lemma \ref{abvkbja}, $\Vert B^{(k,l)}\Vert \leq  2N(z_{\mathrm{max}}-z_{\mathrm{min}})/L$. The rest of the proof is essentially the same.

{\it Implementing arbitrary channels on system $S$--}
By the above results we can conclude that $\Phi_{\rho,U}$ converges to the unitary operation $\mathcal{U}_{U}$ in the limit of a high degree of coherence in the reservoir. This implies that all channels can be implemented on $S$, if we are allowed to use suitable ancillary systems. To see this, recall that every  channel $\Psi$ on $S$ can be written as $\Psi(\rho) = \Tr_A(U\rho\otimes |a\rangle\langle a|U^{\dagger})$ for an ancillary system $A$ with initial state $|a\rangle$.
We regard $SA$ as the new `system' and use the reservoir $E$ to implement the channel $\Phi^{SA,H_{SA}}_{\sigma,U}$ on the initial state $\rho\otimes|a\rangle\langle a|$. [The only restriction on the joint Hamiltonian is $H_{SA}\in H_{s}(\mathcal{H}_S\otimes\mathcal{H}_A)$.] In the limit of a high degree of coherence it is thus the case that $\Phi^{SA,H_{SA}}_{\sigma,U}(\rho\otimes|a\rangle\langle a|)\rightarrow U\rho\otimes |a\rangle\langle a|U^{\dagger}$. By tracing out $A$ we obtain $\Psi(\rho)$ on $S$.

Note that in the case $H_{SA} = H_{S}\otimes\hat{1}_A + \hat{1}_{S}\otimes H_{A}$, the degree of coherence required to achieve a given level of approximation of a unitary operator on the  combined system $SA$ may potentially be higher than what is needed for the corresponding accuracy on system $S$ alone. The reason is that transformations on the combined system may require the reservoir to donate or absorb larger quantities of energy.  

The same observation applies to the case of a sequential application of the reservoir to a collection of systems $S_1,\ldots, S_L$ (with Hamiltonian $H_{S_1\cdots S_L} = H_{S_1}+ \cdots + H_{S_L}$). A `high quality' approximation of $\mathcal{U}_{U_1\otimes \cdots \otimes U_L}$ on the joint system may require a higher degree of coherence in the reservoir than good quality approximations of the individual operations $\mathcal{U}_{U_1},\ldots,\mathcal{U}_{U_L}$.

\subsection{\label{ACatalyticProperty}A catalytic property}

The previous section shows that the set of channels $\mathcal{C}(\sigma)$ which we can reach on system $S$ depends on the coherence in the state $\sigma$ of the energy reservoir $E$. In other words, the coherence in the state of $E$ is a resource for the implementation of channels on $S$. Here we show that this resource is catalytic, i.e., if we use $E$ to implement a channel on one system $S_1$, then this does not diminish the capacity of $E$ to subsequently implement channels on another system $S_2$. It is worth noting that this catalytic property holds for all states $\sigma$; it is not limited to states with a high degree of coherence.

Consider the expression for the channel $\Phi^{S}_{\sigma,U}$ in equation (\ref{channelOnS}). As one can see, the only way in which $\Phi^{S}_{\sigma,U}$ depends on $\sigma$, is via the expectation value of powers of $\Delta$, i.e., on the sequence of numbers $\boldsymbol{(}\Tr(\Delta^{a}\sigma)\boldsymbol{)}_{a\in\mathbb{Z}}$. (Strictly speaking, it only depends on a finite subset of them.) One should note that this set of numbers does not uniquely define $\sigma$. If one think of $\sigma$ represented as a matrix in the energy eigenbasis, i.e., in terms of the matrix $[\langle j|\sigma|j'\rangle]_{j,j'}$, one can see that $\Tr(\Delta^{0}\sigma) = \sum_{j}\langle j|\sigma|j\rangle$ corresponds to the trace, $\Tr(\Delta^{1}\sigma) = \sum_{j} \langle j|\sigma|j+1\rangle$ to the sum of the elements in the first upper diagonal, $\Tr(\Delta^{2}\sigma) = \sum_{j} \langle j|\sigma|j+2\rangle$ to the second, etc. The following lemma shows that these expectation values are invariant under the action of the channel $\Lambda_{\rho,U}$ on the energy reservoir.

\begin{Lemma}
\label{Invariance}
Let $s>0$, $\mathcal{H}_S$ finite-dimensional, and $H_S\in H_s(\mathcal{H}_S)$. Then
\begin{equation}
\Tr\big(\Delta^{a}\Lambda_{\rho,U}(\sigma)\big) =   \Tr(\Delta^{a}\sigma),
\end{equation}
for all $\sigma \in \mathcal{S}(\mathcal{H}_E)$, $U\in U(\mathcal{H}_S)$, $\rho\in\mathcal{S}(\mathcal{H}_S)$, and $a\in\mathbb{Z}$, where $\Lambda_{\rho,U}$ is as defined in equation (\ref{channelOnE}).
\end{Lemma}
\begin{proof}
Let $sz_k$ and $|\psi_k\rangle$ be the eigenvalues and corresponding eigenvectors of $H_S$, and put  $U_{k',k}:= \langle\psi_{k'}|U|\psi_k\rangle$. Then  
\begin{equation}
\begin{split}
& \Tr\big(\Delta^{a}\Lambda_{\rho,U}(\sigma)\big)  \\
 & =    \sum_{k,l,k'}U_{k',k}U_{k',l}^{*}\langle\psi_{k}|\rho|\psi_{l}\rangle \Tr\big(\Delta^{a} \Delta^{z_{k} -z_{k'}}\sigma {\Delta^{z_{l}-z_{k'}} }^{\dagger} \big)\\
 & =   \sum_{k,l,k'}U_{k',k}U_{k',l}^{*}\langle\psi_{k}|\rho|\psi_{l}\rangle \Tr\big(\Delta^{-z_{l}+z_{k'}}\Delta^{a} \Delta^{z_{k} -z_{k'}}\sigma \big)\\
& =   \sum_{k,l,k'}U_{k',k}U_{k',l}^{*}\langle\psi_{k}|\rho|\psi_{l}\rangle \Tr\big(\Delta^{a+ z_{k} -z_{l}}\sigma\big)\\
 & =   \sum_{k,l}\delta_{k,l}\langle\psi_{k}|\rho|\psi_{l}\rangle \Tr\big(\Delta^{a+ z_{k} -z_{l}}\sigma\big)\\
 & =   \sum_{k}\langle\psi_{k}|\rho|\psi_{k}\rangle \Tr(\Delta^{a}\sigma)\\
 & =   \Tr(\Delta^{a}\sigma).
 \end{split}
\end{equation}
\end{proof} 

The following proposition tells us that if we use the reservoir twice, then the set of operations we can reach in the second application is independent of what we did in the first application. In other words, it is as if the first application had not happened. Note that the only assumptions we impose on the Hamiltonians of the subsequent subsystems are that they are non-interacting, and satisfy the `energy matching condition'. We do not assume that they are identical (or even operating on Hilbert spaces of the same dimension). It is straightforward to generalize this proposition to more than two subsystems.
\begin{Proposition}[Catalytic coherence]
\label{WeakCatalytic}
Let $s>0$, and let $\mathcal{H}_{S_1}$ and $\mathcal{H}_{S_2}$ be finite-dimensional Hilbert spaces. Let $H_{S_1}\in H_s(\mathcal{H}_{S_1})$ and $H_{S_2}\in H_s(\mathcal{H}_{S_2})$. Let $U_1\in U(\mathcal{H}_{S_1})$ and $\rho_1\in\mathcal{S}(\mathcal{H}_{S_{1}})$, $U_2\in U(\mathcal{H}_{S_2})$ and $\rho_2\in\mathcal{S}(\mathcal{H}_{S_{2}})$. Let $\Lambda_{\rho_1,U_1}$ be the channel on $E$ as defined by equation (\ref{channelOnE}) for interaction with system $S_1$, and let $\Phi^{S_2}$ be the channel as defined by equation (\ref{channelOnS}), for system $S_2$. Then 
\begin{equation}
\label{dztjdnf}
\Phi^{S_2,H_{S_2}}_{\Lambda_{\rho_1,U_1}(\sigma),U_2}(\rho_2) = \Phi^{S_2,H_{S_2}}_{\sigma,U_2}(\rho_2),
\end{equation}
and consequently
\begin{equation}
\label{nfkbdjsdfs}
\mathcal{C}^{S_2,H_{S_2}}\boldsymbol{(}\Lambda_{\rho_1,U_1}(\sigma)\boldsymbol{)} = \mathcal{C}^{S_2,H_{S_2}}(\sigma). 
\end{equation}
\end{Proposition}
\begin{proof}
Let $|\psi_{n}^{(1)}\rangle$ and $sz^{(1)}_{n}$ be orthonormal eigenvectors and corresponding eigenvalues of $H_{S_1}$, and let $|\psi_{m}^{(2)}\rangle$ and $sz^{(2)}_{m}$ be orthonormal eigenvectors and corresponding eigenvalues of $H_{S_2}$. By equation (\ref{channelOnS}) we can write
\begin{equation}
\label{adsjhfbva}
\begin{split}
 & \Phi^{S_2,H_{S_2}}_{\Lambda_{\rho_1,U_1}(\sigma),U_2}(\rho_2)   \\
 & =    \sum_{k,k',l,l'}\langle\psi^{(2)}_{k'}|U_2|\psi^{(2)}_{k}\rangle\langle\psi^{(2)}_{k}|\rho_2|\psi^{(2)}_{l}\rangle\langle\psi^{(2)}_l|U_2^{\dagger}|\psi^{(2)}_{l'}\rangle\\
 & \quad \times|\psi^{(2)}_{k'}\rangle\langle\psi^{(2)}_{l'}|  \Tr[\Delta^{z^{(2)}_{k} -z^{(2)}_{k'}-z^{(2)}_{l}+z^{(2)}_{l'}}\Lambda_{\rho_1,U_1}(\sigma)].
 \end{split}
\end{equation}
By Lemma \ref{Invariance} we know that 
\begin{equation}
\label{yndfkjvn}
\begin{split}
& \Tr[\Delta^{z^{(2)}_{k} -z^{(2)}_{k'}-z^{(2)}_{l}+z^{(2)}_{l'}}\Lambda_{\rho_1,U_1}(\sigma)] \\
&  =  \Tr[\Delta^{z^{(2)}_{k} -z^{(2)}_{k'}-z^{(2)}_{l}+z^{(2)}_{l'}}\sigma],
\end{split}
\end{equation}
which proves (\ref{dztjdnf}) and  (\ref{nfkbdjsdfs}). 
\end{proof}

In the above proposition we made no assumptions on the relation between the first and the second application of the reservoir. For the sake of clarity we here also treat the special case that we do use the same Hamiltonian on both the first and second system. This is the case that is described in the main text.
\begin{Corollary}
\label{EqualHamiltonians}
Given the assumptions of Proposition \ref{WeakCatalytic}, but with the additional assumptions 
 $\mathcal{H}_{S_2}\eqsim \mathcal{H}_{S_1}$, $H := H_{S_2} = H_{S_1}$, it follows that 
\begin{equation}
\label{dfjkbvnd}
\Phi^{S_2,H}_{\Lambda_{\rho_1,U_1}(\sigma),U_2}(\rho) = \Phi^{S_1,H}_{\sigma,U_2}(\rho)
\end{equation}
and thus
\begin{equation}
\label{bdjhfbvdfhj}
\mathcal{C}^{S_2,H}\boldsymbol{(}\Lambda_{\rho_1,U_1}(\sigma)\boldsymbol{)} = \mathcal{C}^{S_1,H}(\sigma).
\end{equation}
\end{Corollary}
Strictly speaking, the correct formulation of (\ref{dfjkbvnd}) and (\ref{bdjhfbvdfhj}) would be that the channels are isomorphic, as they act on different systems.
\begin{proof}
Assume that $\mathcal{H}_{S_2}$ and $\mathcal{H}_{S_1}$ are isomorphic, with the isomorphism $X :=\sum_{n=1}^{N}|\psi^{(2)}_n\rangle\langle\psi^{(1)}_n|$, where $H_{S_1}|\psi^{(1)}_n\rangle = sz^{(1)}_{n}|\psi^{(1)}_{n}\rangle$, and $H_{S_2}|\psi^{(2)}_n\rangle = sz^{(2)}_{n}|\psi^{(2)}_{n}\rangle$, with $z_n:= z^{(2)}_n= z^{(1)}_n$ (and thus `$H_{S_2} = H_{S_1}$'). 
By equations (\ref{adsjhfbva}) and (\ref{yndfkjvn}) we see that $X^{\dagger}\Phi^{S_2,H_{S_2}}_{\Lambda_{\rho_1,U_1}(\sigma),U_2}(\rho_2) X = \Phi^{S_1,H_{S_1}}_{\sigma,X^{\dagger}U_2 X}(X^{\dagger}\rho_2X)$. Hence, we can implement `the same' channels on $S_2$ as we can on $S_1$, and thus, in this sense, equations (\ref{dfjkbvnd}) and (\ref{bdjhfbvdfhj}) hold.
\end{proof}

\subsection{\label{StongerCatalyticProperty}A stronger catalytic property}

In the previous section we assumed that the system and the energy reservoir initially are uncorrelated. This is necessary if we wish to describe the evolution on the system in terms of channels. However, here we show that even if all systems, including the energy reservoir, initially are in an arbitrary joint state, it is still possible to prove a form of catalytic property (see Fig.~\ref{FigureStrong2}). Due to pre-correlations, we cannot in general phrase this in terms of channels. However, we can express it in terms of the set of states that can be reached. 
As we shall see, one can regain the channel version of Sec.~\ref{ACatalyticProperty} as a special case. 

The following lemma provides the key observation, namely that operations $V^{H_{S_1}}(Q_1)$  and $V^{H_{S_2}}(Q_2)$, corresponding to different subsystems, commute with each other. (Keep in mind that both operations act on the energy reservoir $E$.)
\begin{Lemma}
\label{ndfkjbn}
Let $s>0$, and let $\mathcal{H}_{S_1}$ and $\mathcal{H}_{S_2}$ be finite-dimensional Hilbert spaces. Let $H_{S_1}\in H_s(\mathcal{H}_{S_1})$ and $H_{S_2}\in H_s(\mathcal{H}_{S_2})$, and define $\mathcal{H}_S := \mathcal{H}_{S_1}\otimes\mathcal{H}_{S_2}$ and $H_{S} := \hat{1}_{2}\otimes H_{S_1} + H_{S_2}\otimes \hat{1}_{S_1}$.
The mapping $V$, as defined in Lemma \ref{nvmxcnv}, can be defined on the spaces $\mathcal{H}_{S_1}\otimes\mathcal{H}_{E}^{s}$, $\mathcal{H}_{S_2}\otimes\mathcal{H}_{E}^{s}$, and $\mathcal{H}_{S}\otimes\mathcal{H}_{E}^{s}$, with respect to the Hamiltonians $H_{S_1}$, $H_{S_2}$, and $H_{S}$, respectively. For the sake of clarity we denote these with superscripts. It is the case that 
\begin{equation}
\begin{split}
 & V^{H_S}(\hat{1}_{S_2}\otimes Q_1) = \hat{1}_{S_2}\otimes V^{H_{S_1}}(Q_1),\quad \forall Q_1\in L(\mathcal{H}_{S_{1}}),\\
 & V^{H_S}(Q_2\otimes \hat{1}_{S_1}) =  V^{H_{S_2}}(Q_2)\otimes \hat{1}_{S_1},\quad \forall Q_{2}\in L(\mathcal{H}_{S_{2}}).
 \end{split}
\end{equation}
Furthermore
\begin{equation}
\begin{split}
&  [\hat{1}_{S_2}\otimes V^{H_{S_1}}(Q_1)] [V^{H_{S_2}}(Q_2)\otimes \hat{1}_{S_1}] \\
&  =  [V^{H_{S_2}}(Q_2)\otimes \hat{1}_{S_1}][\hat{1}_{S_2}\otimes V^{H_{S_1}}(Q_1)],
\end{split}
\end{equation}
for all $Q_1\in L(\mathcal{H}_{S_{1}})$ and $Q_{2}\in L(\mathcal{H}_{S_{2}})$.
\end{Lemma}
\begin{proof}
Let $|\psi_{n}^{(1)}\rangle$ and $sz^{(1)}_{n}$ be orthonormal eigenvectors and corresponding eigenvalues of $H_{S_1}$, and let $|\psi_{m}^{(2)}\rangle$ and $sz^{(2)}_{m}$ be orthonormal eigenvectors and corresponding eigenvalues of $H_{S_2}$. For a compact notation we define $|\psi^{(2)}_{m},\psi^{(1)}_{n}\rangle:=|\psi^{(2)}_{m}\rangle|\psi^{(1)}_{n}\rangle$. By using equation (\ref{bijection}) we can write 
\begin{equation*}
\begin{split}
& V^{H_S}(\hat{1}_{S_2}\otimes Q_1)  \\
 & =   \sum_{nm,n'm'}\langle \psi^{(2)}_m,\psi^{(1)}_n| ( \hat{1}_{S_2}\otimes Q_1) |\psi^{(2)}_{m'},\psi^{(1)}_{n'}\rangle \\
 & \quad\times  |\psi^{(2)}_m,\psi^{(1)}_n\rangle\langle \psi^{(2)}_{m'},\psi^{(1)}_{n'}|\otimes \Delta^{z^{(1)}_{n'}- z^{(1)}_{n} + z^{(2)}_{m'}- z^{(2)}_{m}}\\
  & =    \hat{1}_{S_2}\otimes V^{H_{S_1}}(Q_1).
 \end{split}
 \end{equation*}
 An analogous reasoning holds for $V^{H_S}(Q_2\otimes \hat{1}_{S_1})$.  By Lemma \ref{nvmxcnv} we know that $V$ preserves operator multiplication, and thus
 \begin{equation*}
 \begin{split}
&  [\hat{1}_{S_2}\otimes V^{H_{S_1}}(Q_1)] [V^{H_{S_2}}(Q_2)\otimes \hat{1}_{S_1}]  \\
& =   V^{H_S}(\hat{1}_{S_2}\otimes Q_1) V^{H_S}(Q_2\otimes \hat{1}_{S_1})\\
 & =  V^{H_S}(Q_2\otimes Q_1)\\
 & =   V^{H_S}(Q_2\otimes \hat{1}_{S_1})V^{H_S}(\hat{1}_{S_2}\otimes Q_1) \\
 & =  [V^{H_{S_2}}(Q_2)\otimes \hat{1}_{S_1}][\hat{1}_{S_2}\otimes V^{H_{S_1}}(Q_1)]. 
\end{split}
\end{equation*}
\end{proof} 

\begin{figure}[h!]
 \centering
 \includegraphics[width= 8cm]{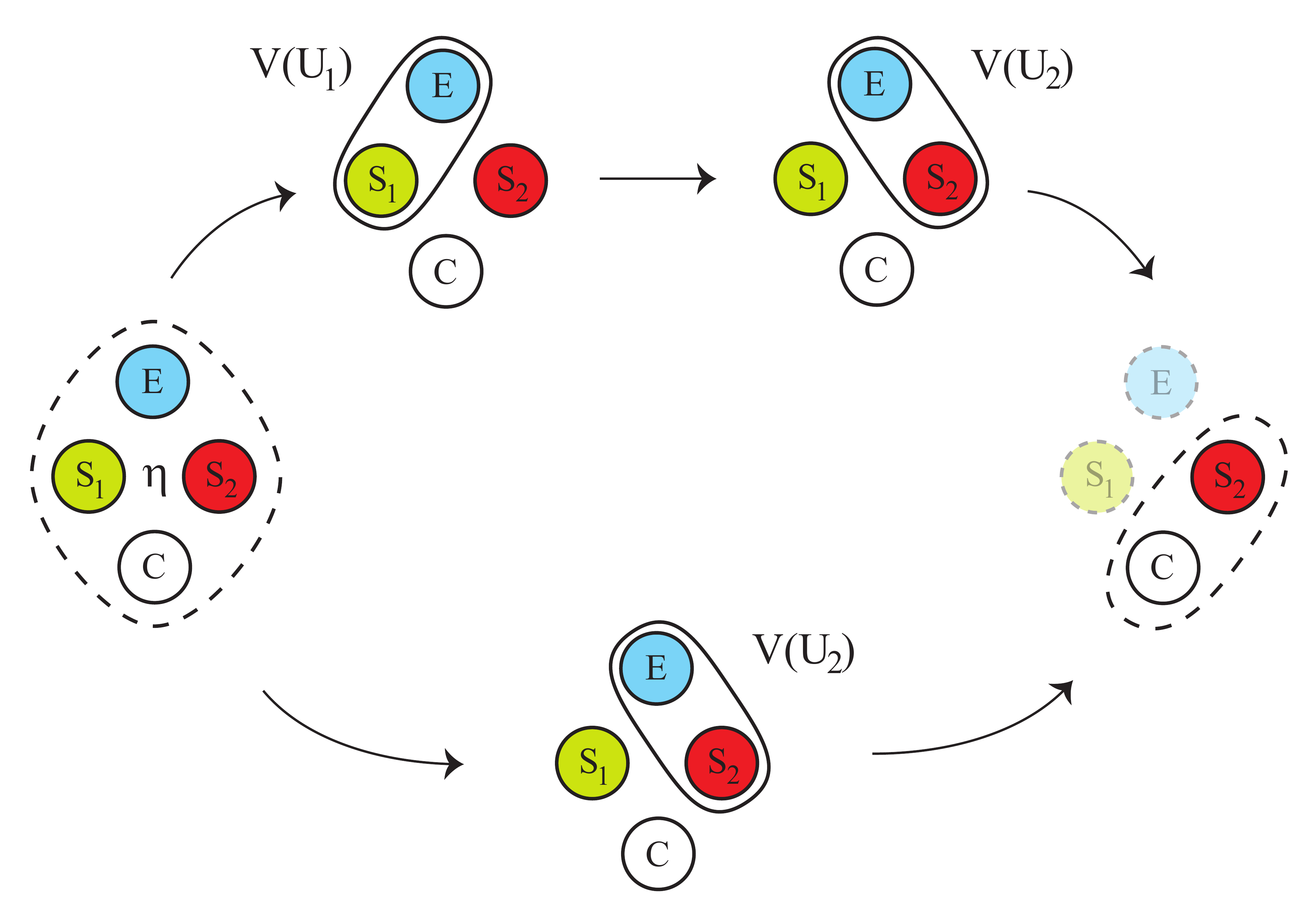} 
   \caption{\label{FigureStrong2} {\bf Strong catalytic property.} 
 The energy reservoir $E$, two systems $S_1$, $S_2$, and a reference system $C$ are initially in an arbitrary joint state $\eta$.  Here we compare two ways (the upper and lower path) to operate on these systems. If we first operate on $ES_1$ with $V(U_1)$ followed by $V(U_2)$ on $ES_2$ (upper path), this yields the same state $\rho_{S_2C}$ on $S_2C$ as if we only operate with $V(U_2)$ on $ES_2$ (lower path). In other words, the application of $V(U_1)$ does not decrease the set of states that can be reached on $S_2C$. This result underlines the observation that catalytic coherence primarily is a property of  a specific class of dynamics, rather than a special class of states.  The reference system $C$ allows us to compare this version of catalytic coherence with the channel version in the previous sections. 
   }
\end{figure}

Suppose that we, apart from the energy reservoir $E$, have two systems of interest $S_1$, $S_2$, as well as a reference system $C$. These systems are initially in some arbitrary joint state $\eta$. 
One can compare two different scenarios (see Fig.~\ref{FigureStrong2}). In the first scenario, we first implement a unitary operation $V^{H_{S_1}}$ on  $ES_1$ followed by a unitary operation $V^{H_{S_2}}$ on $ES_2$. 
In the second scenario, we only perform the second operation $V^{H_{S_2}}$ on $ES_2$. The following proposition tells us that the reduced density operator on $S_2C$ is the same in both scenarios. (Note that we do not operate on system $C$. The only purpose of this system is to keep track of correlations.)

\begin{Proposition}[Strong catalytic property]
\label{StrongCatalytic}
Let $s>0$, and let $\mathcal{H}_{S_1}$, $\mathcal{H}_{S_2}$, and $\mathcal{H}_C$ be finite-dimensional Hilbert spaces. Let $H_{S_1}\in H_s(\mathcal{H}_{S_1})$ and $H_{S_2}\in H_s(\mathcal{H}_{S_2})$, and define $\mathcal{H}_S = \mathcal{H}_{S_1}\otimes\mathcal{H}_{S_2}$ and $H_{S} = \hat{1}_{S_2}\otimes H_{S_1} + H_{S_2}\otimes \hat{1}_{S_1}$.   (The Hamiltonian $H_C$ on $\mathcal{H}_C$ can be chosen arbitrarily. The total Hamiltonian is $H_S\otimes\hat{1}_C + \hat{1}_S \otimes H_{C}$.) Let $U_1\in U(\mathcal{H}_{S_1})$ and $U_2\in U(\mathcal{H}_{S_2})$, then 
\begin{equation}
\label{bdfjbnyfj}
\begin{split}
 & \Tr_{ES_1}\Big( V^{H_{S_2}}(U_2) V^{H_{S_1}}(U_1)\eta V^{H_{S_1}}(U_1)^{\dagger}V^{H_{S_2}}(U_2)^{\dagger}\Big) \\
 & =  \Tr_{E}\big( V^{H_{S_2}}(U_2)\Tr_{S_1}(\eta) V^{H_{S_2}}(U_2)^{\dagger}\big),
\end{split}
\end{equation}
for all  $\eta \in \mathcal{S}(\mathcal{H}_{S_1}\otimes\mathcal{H}_{S_2}\otimes\mathcal{H}_C\otimes\mathcal{H}_E)$.
\end{Proposition}
The reason for why the Hamiltonian $H_C$ on $\mathcal{H}_C$ can be arbitrary is that the operations never touch this subsystem. 

We should strictly speaking write `$[V^{H_{S_2}}(U_2)\otimes\hat{1}_{S_1}\otimes \hat{1}_C][\hat{1}_{S_2}\otimes V^{H_{S_1}}(U_1)\otimes \hat{1}_C]$' rather than  `$V^{H_{S_2}}(U_2)V^{H_{S_1}}(U_1)$' in the first line of equation (\ref{bdfjbnyfj}). Similarly, we should write `$V^{H_{S_2}}(U_2)\otimes \hat{1}_C$' rather than `$V^{H_{S_2}}(U_2)$' in the second line.

\begin{proof}
By Lemma \ref{ndfkjbn} and due to cyclic permutation under the (partial) trace $\Tr_{ES_1}$ (this is applied to an operator that acts trivially on $\mathcal{H}_{S_2}\otimes\mathcal{H}_C$) it follows that
\begin{equation*}
\begin{split}
 & \Tr_{ES_1}\Big(V^{H_{S_2}}(U_2) V^{H_{S_1}}(U_1) \eta V^{H_{S_1}}(U_1)^{\dagger}V^{H_{S_2}}(U_2)^{\dagger}\Big) \\
 & =  \Tr_{ES_1}\Big( V^{H_{S_1}}(U_1) V^{H_{S_2}}(U_2)\eta V^{H_{S_2}}(U_2)^{\dagger}V^{H_{S_1}}(U_1)^{\dagger}\Big) \\
 &  =  \Tr_{ES_1}\Big([\hat{1}_{S_2}\otimes V^{H_{S_1}}(U_1)^{\dagger}V^{H_{S_1}}(U_1)] \\
  & \quad \quad \quad \quad \quad \quad V^{H_{S_2}}(U_2)\eta V^{H_{S_2}}(U_2)^{\dagger}\Big).
\end{split}
 \end{equation*}
 According to Lemma \ref{nvmxcnv}, $V$ preserves the Hermitian conjugate, operator multiplication, and the identity. Thus  $V^{H_{S_1}}(U_1)^{\dagger}V^{H_{S_1}}(U_1) = \hat{1}_{S_1}$. From this we can conclude equation (\ref{bdfjbnyfj}).
\end{proof}

\emph{Re-derivation of Proposition \ref{WeakCatalytic} from Proposition \ref{StrongCatalytic}.--}
Here we show that Proposition \ref{WeakCatalytic} can be derived from  Proposition \ref{StrongCatalytic}. Here we make use of the ancillary system $C$ to form a Choi-state of the relevant channel.

Assume that $\mathcal{H}_C \eqsim  \mathcal{H}_{S_2}$, and 
\begin{equation*}
\eta = |\psi\rangle_{S_2C}\langle\psi|\otimes \rho_1\otimes \sigma,
\end{equation*}
where $|\psi\rangle_{S_2C}$ is a maximally entangled state.
With these assumptions it follows from equation (\ref{bdfjbnyfj}) that 
\begin{equation*}
    [\Phi^{H_{S_2}}_{\Lambda_{\rho_1,U_1}(\sigma),U_2}\otimes I_{C}](|\psi\rangle_{S_2C}\langle\psi|)
   =   [\Phi^{H_{S_2}}_{\sigma,U_2}\otimes I_C](|\psi\rangle_{S_2C}\langle\psi|).
\end{equation*}
Hence, the Choi-states of the two channels $\Phi^{H_{S_2}}_{\Lambda_{\rho_1,U_1}(\sigma),U_2}$ and $\Phi^{H_{S_2}}_{\sigma,U_2}$ are equal, which implies that the channels themselves are equal. We thus obtain the statement of Proposition \ref{WeakCatalytic}.

\emph{Operating several times on the same system.--}
Imagine that we operate on the same system $S$ with the same energy reservoir $E$, by a sequence of operations $V^{H_S}(U_1)$, $V^{H_S}(U_2)$, \ldots, $V^{H_S}(U_L)$. It follows directly by Lemma \ref{nvmxcnv} that the resulting operation $V^{H_S}(U_L)\cdots V^{H_S}(U_2)V^{H_S}(U_1)$ is equal to $V^{H_S}(U_L\cdots U_2U_1)$. In other words, we can view a repeated application on the same system, as if it was a single application. Hence, in this sense, the number of times that we use the reservoir on the same system does not matter. 
It is maybe worth to point out that since we let the two systems interact repeatedly, and thus generically build up correlations, we cannot in general describe this sequential evolution as a concatenation of channels on $S$ alone. In other words, the dynamics on system $S$ is in the general case not Markovian.

One can also imagine that we have a number of subsystems $S_1,\ldots, S_K$, but that we operate on these systems more than once, in an arbitrary sequence. By combining the above reasoning with the commutative property in Lemma \ref{ndfkjbn}, one can realize that any such sequence of operations is equivalent to a scenario where we operate on each system once.
As an example
\begin{equation*}
\begin{split}
& V^{H_{S_2}}(U_5)V^{H_{S_1}}(U_4)V^{H_{S_3}}(U_3)V^{H_{S_1}}(U_2)V^{H_{S_3}}(U_1)\\
= & V^{H_{S_3}}(U_3)V^{H_{S_3}}(U_1)V^{H_{S_2}}(U_5)V^{H_{S_1}}(U_4)V^{H_{S_1}}(U_2)\\
= & V^{H_{S_3}}(U_3U_1)V^{H_{S_2}}(U_5)V^{H_{S_1}}(U_4U_2).
\end{split}
\end{equation*}

\emph{Characterization of relevant properties.--}
In the previous section, where we considered reservoirs that are initially uncorrelated with the system, we could characterize the relevant aspects of the reservoir via the numbers $\Tr(\Delta^{a}\sigma)$ for $a\in\mathbb{Z}$. 
Here, this role is in some sense taken over by the operators $\Tr_E([\hat{1}_{SC}\otimes \Delta^{a}]\eta)$. These operators `determine' which states that can be reached on $SC$.  One can also see that this set of operators is not affected by   any previous application of the reservoir on other systems.

\subsection{\label{AnExample}An example: Sequential preparation}
Here we consider a simple illustration of the catalytic implementation of channels. Consider a collection of non-interacting two-level systems. We let $|\psi_0\rangle,|\psi_1\rangle$ denote the eigenvectors with corresponding eigenvalues $z_0 = 0, z_1 = 1$ of each system Hamiltonian. (We let $s=1$.)
Suppose that all two-level systems initially are in the ground states $|\psi_0\rangle$. Then the channels (\ref{channelOnS}) and (\ref{channelOnE}) take the form
\begin{equation*}
\begin{split}
\Phi_{\sigma,U}(|\psi_0\rangle\langle\psi_0|)  = & |U_{00}|^{2} |\psi_{0}\rangle\langle\psi_{0}| +|U_{10}|^2|\psi_{1}\rangle\langle\psi_{1}| \\
 &+ U_{10}^{*}U_{00}\Tr (\Delta\sigma)|\psi_{0}\rangle\langle\psi_{1}| \\
 &+ U_{10}U_{00}^{*}\Tr (\Delta^{\dagger}\sigma)|\psi_{1}\rangle\langle\psi_{0}|, \\ 
\Lambda_{|\psi_0\rangle\langle\psi_0|,U}(\sigma)   = & |U_{00}|^{2} \sigma  + |U_{10}|^{2}\Delta^{\dagger}\sigma \Delta,
\end{split}
\end{equation*} 
where $U_{nn'} = \langle\psi_n|U|\psi_{n'}\rangle$. 

Suppose that we ideally would like to put all the two-level systems into the state $|\phi\rangle = (|\psi_0\rangle-i|\psi_1\rangle)/\sqrt{2}$, where we sequentially use one and the same reservoir, initially in state $\sigma^{(0)}$. As a measure of the quality of this preparation we  take the (square of the) fidelity $F_U :=\langle \phi|\Phi_{\sigma^{(0)},U}(|\psi_0\rangle\langle\psi_0|)|\phi\rangle$, and optimize over all unitary $U$. 

An optimizing unitary $U^{\textrm{opt}}$ is obtained by $U^{\textrm{opt}}_{10} = 1/\sqrt{2}$ and $U^{\textrm{opt}}_{00} = i\exp[-i\arg(\Tr (\Delta\sigma^{(0)}))]/\sqrt{2}$, and leads to the optimal value $F_{\textrm{opt}}^{(0)} = \frac{1}{2} + \frac{1}{2}|\Tr(\Delta\sigma^{(0)})|$.
The prepared state is 
\begin{equation}
\label{nfdjkbin}
\begin{split}
\Phi_{\sigma,U^{\textrm{opt}}}(|\psi_0\rangle\langle\psi_0|)  = & \frac{1}{2}|\psi_0\rangle\langle \psi_0| + \frac{1}{2}|\psi_1\rangle\langle\psi_1| \\
 & + i\frac{1}{2}|\psi_{0}\rangle\langle\psi_{1}||\Tr (\Delta\sigma)|\\
 & -i \frac{1}{2}|\psi_{1}\rangle\langle\psi_{0}||\Tr (\Delta\sigma)|.
\end{split}
\end{equation}
With the optimizing unitary $U^{\textrm{opt}}$, the next state of the reservoir is $\sigma^{(1)} := \frac{1}{2}\sigma^{(0)} + \frac{1}{2}\Delta^{\dagger}\sigma^{(0)} \Delta$. As seen (and as we already know) $\Tr(\Delta \sigma^{(1)}) = \Tr(\Delta\sigma^{(0)})$. Hence, in the second step we can use the same optimizing unitary, resulting in the same optimum $F_{\textrm{opt}}^{(1)} = F_{\textrm{opt}}^{(0)}$. 
 This process can be repeated indefinitely, with a constant sequence of fidelities $F_{\textrm{opt}}^{(k)} =  \frac{1}{2} + \frac{1}{2}|\Tr(\Delta\sigma^{(0)})|$ for all $k$. 

Although the quality of each individual preparation thus remains constant over the repetitions, the state of the reservoir changes. At the $k$th step the state of the reservoir is 
\begin{equation}
\label{dmnvb}
\sigma^{(k)} = \frac{1}{2^{k}}\sum_{l=0}^{k}\binom{k}{l} {\Delta^{l}}^{\dagger}\sigma^{(0)}\Delta^{l}.
\end{equation} 
This provides an example of the spreading of the state of the reservoir mentioned in the main text. 
As a special case one can choose $\sigma^{(0)} := |\eta_{L,l_0}\rangle\langle\eta_{L,l_0}|$, for $|\eta_{L,l_0}\rangle =\sum_{l=0}^{L-1}|l+ l_0\rangle/\sqrt{L}$. In this case we obtain $F_{\textrm{opt}}^{(k)} =  1-1/(2L)$.

Equation (\ref{dmnvb}) describes an incoherent mixture of translations of the input state. In the generic case, as described by equation (\ref{channelOnE}), the action on the reservoir is more general. However, whenever the initial state $\rho$ of system $S$ is diagonal with respect to the energy eigenspaces (as it is in this example) one  obtains an incoherent mixture of translations. 

%%%%%%%%%%%%%%%%%%%%%%%%%%
% Generalization to Abelian representations      %
%%%%%%%%%%%%%%%%%%%%%%%%%%

\subsection{\label{Sec:Abelian} Generalization to unitary representations of Abelian groups}
In this section we make a bit of a detour in an attempt to shed some light on the underlying mathematical structure of catalytic coherence.  We here abandon questions of energy conservation and Hamiltonians, and make the observation that all the results of Sections  \ref{ACatalyticProperty} and \ref{StongerCatalyticProperty} can be regained if the rigid translations $\Delta^{a}$ along the energy ladder are replaced by a unitary representation of an Abelian group. 

For an Abelian group $\langle G, +\rangle$, let $\{U(g)\}_{g\in G}$ be a unitary representation of $G$ on a Hilbert space $\mathcal{H}_E$. $U$ satisfies the standard properties, i.e., $U(g)$ is unitary and $U(g+g') = U(g)U(g')$, and $U(0)= \hat{1}$, where $0$ denotes the identity of $\langle G, +\rangle$. If $-g$ denotes the group inverse of $g$, then $U(-g)= U(g)^{\dagger}$.

For a subset $S\subseteq G$ let $\{|\psi_{g}\rangle\}_{g\in S}$ be an orthonormal basis spanning a Hilbert space $\mathcal{H}_S$. It is maybe worth emphasizing that $S$ does \emph{not} have to form a subgroup of $G$. Define the following mapping 
\begin{equation}
\tilde{V}(Q) = \sum_{g,g'\in S}|\psi_{g}\rangle\langle\psi_{g}|Q|\psi_{g'}\rangle\langle\psi_{g'}|\otimes U(g'-g).
\end{equation}
The results of the previous sections can be phrased as the special case of $\langle \mathbb{Z},+\rangle$, where the choice of the subset $S\subseteq G$ corresponds to the collection integers $z_n$ describing the eigenvalues of $H_S$.

Analogous to the mapping $V$ in Lemma \ref{nvmxcnv}, $\tilde{V}$ preserves the operator algebra on $\mathcal{H}_S$, i.e.,  $\tilde{V}$ satisfies the properties in (\ref{preservation}). We can also define the channels on $S$ and $E$,
\begin{equation}
\begin{split}
\tilde{\Phi}_{\sigma,U}(\rho)  := &  \Tr_E \tilde{V}(U)\rho\otimes\sigma \tilde{V}(U)^{\dagger}\\
 = & \sum_{g,g',\overline{g},\overline{g}'}\langle\psi_{g}|U|\psi_{g'}\rangle\langle\psi_{g'}|\rho |\psi_{\overline{g}'}\rangle\langle\psi_{\overline{g}'}|U^{\dagger}|\psi_{\overline{g}}\rangle \\
 & \times\Tr \big(U(g'-g-\overline{g}'+\overline{g})\sigma\big) |\psi_{g}\rangle\langle\psi_{\overline{g}}|, 
\end{split}
\end{equation}
\begin{equation}
\label{AbelianchannelOnE}
\begin{split}
\tilde{\Lambda}_{\rho,U}(\sigma)  := & \Tr_{S}\tilde{V}(U)\rho\otimes \sigma \tilde{V}(U)^{\dagger}\\
 = & \sum_{g,g',\overline{g}'\in S}\langle\psi_{g}|U|\psi_{g'}\rangle\langle\psi_{g'}| \rho|\psi_{\overline{g}'}\rangle\langle\psi_{\overline{g}'}|U^{\dagger}|\psi_{g}\rangle   \\
 &\times  U(g'-g)\sigma U(\overline{g}'-g)^{\dagger}.
\end{split}
\end{equation} 
As an analogue of Lemma \ref{Invariance}, one can verify that 
\begin{equation}
\label{AbelianInvariance}
\Tr\big(U(g'')\tilde{\Lambda}_{\rho,U}(\sigma)\big) =   \Tr\big(U(g'')\sigma\big),
\end{equation}
for all $\sigma \in \mathcal{S}(\mathcal{H}_E)$, $U\in U(\mathcal{H}_S)$, $\rho\in\mathcal{S}(\mathcal{H}_S)$, and $g''\in G$, where $\Lambda_{\rho,U}$ is as defined in equation (\ref{AbelianchannelOnE}).
Using this, one can prove the analogue of Proposition  \ref{WeakCatalytic}.

One can also generalize the strong catalytic property in Proposition \ref{StrongCatalytic}. 
If $S_1\subseteq G$ and $S_2\subseteq G$ are two arbitrary subsets, and $\{|\psi^1_{g}\rangle\}_{g\in S_1}$ and $\{|\psi^2_{f}\rangle\}_{f\in S_2}$ are orthonormal bases of $\mathcal{H}_{S_1}$ and $\mathcal{H}_{S_2}$, respectively, we can generalize the mapping $\tilde{V}$ as
\begin{equation}
\begin{split}
 & \tilde{V}^{21}(Q)  \\
 & =  \sum_{g,g'\in S_1}\sum_{f,f'\in S_2}   |\psi^2_{f},\psi^1_{g}\rangle\langle\psi^2_{f},\psi^1_{g}|Q|\psi^2_{f'},\psi^1_{g'}\rangle\langle\psi^2_{f'},\psi^1_{g'}|\\
 & \quad\quad\otimes U(g'+f'-g-f).
 \end{split}
\end{equation}
where $ |\psi^2_{f},\psi^1_{g}\rangle := |\psi^2_{f}\rangle|\psi^1_{g}\rangle$.
With this extended definition one can prove the analogue of Lemma \ref{ndfkjbn}, which directly leads to a  strong catalytic property as in Proposition \ref{StrongCatalytic},
\begin{equation}
\begin{split}
& \Tr_{ES_1}\Big( \tilde{V}(U_2)\tilde{V}(U_1)\eta \tilde{V}(U_1)^{\dagger}\tilde{V}(U_2)^{\dagger}\Big) \\
 & =  \Tr_{E}\big( \tilde{V}(U_2)\Tr_{S_1}(\eta) \tilde{V}(U_2)^{\dagger}\big).
\end{split}
\end{equation}
(In the above equation we should strictly speaking write `$[\tilde{V}(U_2)\otimes \hat{1}_{S_1}\otimes\hat{1}_C][\hat{1}_{S_2}\otimes\tilde{V}(U_1)\otimes\hat{1}_C]$' rather than `$\tilde{V}(U_2)\tilde{V}(U_1)$', where $C$ is the ancillary reference system introduced in Sec.~\ref{StongerCatalyticProperty}.)

Since the generalization to Abelian groups goes through, an obvious question is if non-Abelian groups also would work.
One can indeed construct a new map
 \begin{equation}
\overline{V}(Q) = \sum_{g,g'\in S}|\psi_{g}\rangle\langle\psi_{g}|Q|\psi_{g'}\rangle\langle\psi_{g'}|\otimes U(g^{-1}g').
\end{equation}
which again satisfy $\overline{V}(A)\overline{V}(B) = \overline{V}(AB)$, $\overline{V}(A)^{\dagger} = \overline{V}(A^{\dagger})$, and $\overline{V}(\hat{1}_{S}) = \hat{1}_{S}\otimes\hat{1}_{E}$. However, the property (\ref{AbelianInvariance}) does in general not hold for $\overline{V}$.

%%%%%%%%%%%%%%
%     Reference system    %
%%%%%%%%%%%%%%

\section{\label{TranslationPhaseRef} Reformulation in terms of correlations with a reference system}

In the context of reference frames it is known that references can be treated both explicitly and implicitly \cite{Bartlett07} (see also the discussion on coherence in \cite{Bartlett05}). In an attempt to gain a bit more understanding of the nature of catalytic coherence we here make a specific construction along these lines, where we `translate' superposition of energy eigenstates and coherence into correlations with an explicit reference system. On this extended system, all operations can be performed with states and measurements that are block diagonal with respect to the total Hamiltonian, i.e., the effect of superpositions can in some sense  be `simulated' on the joint system. (The word `simulation' may potentially suggest that this scenario somehow would be artificial and less `real'. However, this choice of terminology is merely intended as a convenient way to make it unambiguous which of the two pictures that is under discussion.)

The general idea can be pictured by rephrasing the general setting of the previous sections as a game that we play against a referee. This referee prepares all our  resources (i.e. all systems $S$, as well as the energy reservoir $E$) and gives them to us. After we have made our energy preserving operations, we hand back the systems to the referee, who makes measurements to check the quality of our operations, e.g., how well we have been able to perform energy mixing unitary operations. From the viewpoint of the previous sections, we can achieve coherent operations of high quality when the referee has provided us with a high degree of coherence in terms of broad superpositions of consecutive energy eigenstates in the reservoir. 

However, there is an alternative approach that would yield identical results, when it comes to  measurement statistics of the referee's tests. In this second version of game, the referee does not give us any systems that are in superposition between energy eigenspaces. Instead, the referee prepares a specific type of states correlated to an additional reference system $R$, in such a way that the joint state is diagonal with respect to the total Hamiltonian. As before, the referee gives us the prepared systems $S$ and $E$, but keeps the reference $R$. After we have made our operations and handed back the systems, the referee makes joint measurements on $R$ and the systems, in such a way that the measurement statistics becomes equivalent to the first version of the game. In this sense one can view the second version as a `simulation' of the first version. In this simulation, the role of superpositions with respect to energy eigenspaces is in some sense taken over by correlations to the reference $R$. Thus, instead of regarding coherence, or `off-diagonality' as a resource, we should in this scenario view correlations with the reference as the relevant asset.  

In the following we shall construct such an explicit reference system, tailored to the doubly infinite ladder model in Section \ref{DoublyInfinite}, and translate catalytic coherence to this setting. We do this via  specific mappings that send density operators, observables, and POVMs to corresponding operators in the simulation. In Section \ref{SeqSeparable} we shall consider the special case of sequential state preparation from initial states with definite energy (analogously to what we did in Section \ref{AnExample}). We will see that in this case superpositions in the prepared states  correspond to correlated but separable states with respect to the reference. Since separable correlations, as opposed to entanglement, are not monogamous, this observation may provide some additional intuitive understanding of why catalytic coherence is possible

\subsection{Simulating coherence}

Let $H_Q\in\mathcal{H}_{s}(\mathcal{H}_Q)$ and let  $P_{j}$ be the eigenprojectors of $H_Q$ corresponding to distinct eigenvalues  $sz_j$ (i.e., $z_j\neq z_{j'}$ if $j\neq j'$). Let $\mathcal{H}_{R}$ be infinite-dimensional with the Hamiltonian $H_{R} = s\sum_{j\in\mathbb{Z}}j|j\rangle\langle j|$. As one may note, $H_R$ is structurally identical to the energy reservoir hamiltonian $H_E$ in equation (\ref{DefHEreservoir}), but it will serve a somewhat different purpose. The full simulation of  catalytic coherence will involve both $E$ and $R$. 

On $\mathcal{H}_{Q}\otimes\mathcal{H}_R$ we define
\begin{equation}
Y =  \sum_{j}\sum_{l\in \mathbb{Z}}P_{j}\otimes |l-z_j\rangle\langle l|  =  \sum_{j}P_{j}\otimes \Delta^{-z_j}.
\end{equation}
This operator is unitary, but not energy conserving with respect to $H_Q + H_{R}$.

For $A\in \mathcal{L}(\mathcal{H}_Q)$ and $B\in\mathcal{L}(\mathcal{H}_{R})$ define
\begin{equation}
\label{ThetaDef}
\Theta^{Q:R}_{B}(A) := Y[A\otimes B] Y^{\dagger}.
\end{equation}
For $A'\in \mathcal{L}(\mathcal{H}_Q)$ and $B'\in\mathcal{L}(\mathcal{H}_{R})$ one can see that 
\begin{equation}
\begin{split}
 & \Tr\big(\Theta^{Q:R}_{B}(A)\Theta^{Q:R}_{B'}(A') \big) =  \Tr(AA')\Tr(BB'), \\
 &\Theta^{Q:R}_{B}(A)^{\dagger}  =   \Theta^{Q:R}_{B^{\dagger}}(A^{\dagger}),\\
 & \Theta^{Q:R}_{\hat{1}}(\hat{1}) =  \hat{1}\otimes\hat{1},\\
 & A\geq 0,\quad B\geq 0\quad \Rightarrow \quad \Theta^{Q:R}_{B}(A)\geq 0. 
 \end{split}
\end{equation}
Note that the mapping $V$ defined in Lemma \ref{nvmxcnv} is nothing but $V(A) = \Theta^{S:E}_{\hat{1}}(A)$. Although identical operations, we will keep the distinct notation, as these mappings have different roles and act on different systems.

From these observations we can conclude the following:
\begin{itemize}
\item If $\rho\in \mathcal{S}(\mathcal{H}_Q)$ and $\xi\in\mathcal{S}(\mathcal{H}_R)$ then $\Theta^{Q:R}_{\xi}(\rho)\in\mathcal{S}(\mathcal{H}_{Q}\otimes\mathcal{H}_R)$.
\item If $\{A_{k}\}_{k}$ is a POVM on $\mathcal{H}_Q$, then $\{\Theta^{Q:R}_{\hat{1}}(A_k)\}_{k}$ is a POVM on $\mathcal{H}_Q\otimes\mathcal{H}_R$.
\item If $\rho\in \mathcal{S}(\mathcal{H}_Q)$ and $\xi\in\mathcal{S}(\mathcal{H}_R)$ and $\{A_{k}\}_{k}$ is a POVM on $\mathcal{H}_Q$, then 
\begin{equation}
\Tr\big(\Theta^{Q:R}_{\hat{1}}(A_k)\Theta^{Q:R}_{\xi}(\rho)\big) = \Tr(A_k\rho).
\end{equation}
\end{itemize}
In other words, the effect of preparing density operators $\rho$ and measuring POVMs $A_k$ on $Q$ can be simulated by preparing density operators $\Theta^{Q:R}_{\xi}(\rho)$ and measuring POVMs $\Theta^{Q:R}_{\hat{1}}(A_k)$ on $\mathcal{H}_Q\otimes\mathcal{H}_R$.

For $\xi\in\mathcal{S}(\mathcal{H}_R)$ such that $[\xi]_{H_{R}} = \xi$, it is the case that 
\begin{equation}
[\Theta^{Q:R}_{\xi}(\rho)]_{H_{Q}+ H_R} = \Theta^{Q:R}_{\xi}(\rho),\quad \forall \rho\in\mathcal{S}(\mathcal{H}_Q),
\end{equation}
where $[\cdot]_{H_{Q}+ H_R} = \sum_{l\in\mathbb{Z}}P^{(l)}\cdot P^{(l)}$, with $P^{(l)} = \sum_{n}P_n\otimes |l-z_n\rangle\langle l-z_n|$.
In other words, if $\xi$ is diagonal with respect to $H_R$, then $\Theta^{Q:R}_{\xi}(\rho)$ is diagonal with respect to $H_Q + H_R$, irrespective of whether $\rho$ is diagonal with respect to $H_Q$ or not.

Analogously, for any POVM $\{A_k\}_k$ on $\mathcal{H}_Q$, it is the case that 
\begin{equation}
[\Theta^{Q:R}_{\hat{1}}(A_k)]_{H_{Q}+ H_R} = \Theta^{Q:R}_{\hat{1}}(A_k),\quad \forall k.
\end{equation}
Hence, every POVM on $Q$ is mapped to a POVM that is diagonal with respect to $H_Q + H_R$.

The above observations mean that we can simulate the preparation of arbitrary states, irrespective of how `off-diagonal' they are, and analogously we can measure arbitrary POVMs on $Q$, by preparing and measuring other states and POVMs that are diagonal with respect to the Hamiltonian on the larger system $QR$. In other words, even if we have no access to superpositions,  we can, in this sense, simulate the effect of superpositions.  

In this setting one should not think of $\Theta$ as a physical operation. (It is disqualified due to the assumption of energy conservation.) One should rather imagine that we input a \emph{description} of a state $\rho$ to a preparation device. Rather than preparing $\rho$, this device prepares $\Theta_{\xi}^{Q:R}(\rho)$. Similarly, given a description of a POMV $\{A_k\}_k$ a measurement device implements the POVM $\{\Theta^{Q:R}_{\hat{1}}(A_k)\}_{k}$. 

\subsection{\label{RefFrExamples}Examples}
The nature of the states that the mapping $\Theta_{\xi}^{Q:R}$ generates depends on the choice of state $\xi$. Here we consider the special case that $\xi$ is in a pure energy eigenstate of $H_R$, e.g., $|0\rangle$. We furthermore let $H_Q$ be non-degenerate with eigenstates $\{|\psi_j\rangle\}_{j}$.
 In this case a diagonal state corresponds to a separable state, as seen by
\begin{equation*}
\Theta_{|0\rangle\langle 0|}^{Q:R}(\sum_{j}\lambda_j |\psi_j\rangle\langle\psi_j|) = \sum_{j}\lambda_{j}|\psi_j\rangle\langle\psi_{j}|\otimes |-z_{j}\rangle\langle -z_{j}|.
\end{equation*}
A pure superposition $|\phi\rangle = \sum_{j}c_{j}|\psi_j\rangle$ is mapped to an entangled state $\Theta_{|0\rangle\langle 0|}^{Q:R}(|\phi\rangle\langle\phi|) = |\tilde{\phi}\rangle\langle\tilde{\phi}|$, where $|\tilde{\phi}\rangle = \sum_{j}c_j|\psi_j\rangle|-z_{j}\rangle$.

\subsection{\label{Sec:equivalenceRelation}An equivalence relation}
Under the assumption that we only measure observables/POVM elements of the form $\Theta^{Q:R}_{\hat{1}}(A)$ there are classes of states on $\mathcal{H}_{Q}\otimes\mathcal{H}_{R}$ that we cannot distinguish. More precisely, if $\eta,\eta'\in \mathcal{S}(\mathcal{H}_Q\otimes\mathcal{H}_R)$, then we can define the equivalence relation
\begin{equation*}
\eta\sim\eta' \, \Leftrightarrow \, \Tr(\Theta^{Q:R}_{\hat{1}}(A)\eta) = \Tr(\Theta^{Q:R}_{\hat{1}}(A)\eta'),\, \forall A\in \mathcal{L}(\mathcal{H}_Q).
\end{equation*} 
Using the definition of $\Theta$ in equation (\ref{ThetaDef}), one can see that this condition reduces to 
\begin{equation*}
\sum_{j,j'}P_{j'}\Tr_{R}(\Delta^{z_{j'}-z_j} \eta) P_{j}  =   \sum_{j,j'}P_{j'}\Tr_{R}(\Delta^{z_{j'}-z_j} \eta') P_{j}.
\end{equation*}

\subsection{Energy conserving unitary operations}
 If $V\in U(\mathcal{H}_Q)$ is such that $[V]_{H_Q} = V$, then $\Theta^{Q:R}_{\hat{1}}(V) = V\otimes\hat{1}_R$, and thus
\begin{equation}
\label{EnergyConserve}
(V\otimes \hat{1}_R)\Theta^{Q:R}_{\xi}(\rho)(V\otimes \hat{1}_R)^{\dagger} = \Theta^{Q:R}_{\xi}(V\rho V^{\dagger}),
\end{equation}
for all $\rho\in\mathcal{S}(\mathcal{H}_Q)$ and $\xi\in\mathcal{S}(\mathcal{H}_R)$.

Hence, if we  restrict to energy conserving unitary operations on system $Q$, then we can perform this operation on $Q$ without having access to the reference $R$. Later we will let $V$ be the unitary operations $V(U)$ as defined in equation (\ref{bijection}), which by construction are energy conserving.

\subsection{\label{RefFrPartialTr}Partial trace}

Suppose that $Q$ consists of two subsystems $Q = Q_1Q_2$, with the joint Hamiltonian $H_Q = H_{Q_1} + H_{Q_2}$, where $H_{Q_1}\in H_s(\mathcal{H}_{Q_1})$ and $H_{Q_2}\in H_s(\mathcal{H}_{Q_2})$. 

Let $\{P^{(1)}_{j}\}_{j}$ be the eigenprojectors of $H_{Q_1}$ with respect to the distinct eigenvalues $sz^{(1)}_{j}$, and analogously  let $\{P^{(2)}_{k}\}_{k}$ be the eigenprojectors of $H_{Q_2}$ corresponding to the distinct eigenvalues $sz^{(2)}_{k}$. 
One can see that 
\begin{equation}
 \Theta_{\hat{1}}^{Q_1Q_2:R}(A^{(1)}\otimes\hat{1}_2) =  \Theta_{\hat{1}}^{Q_1:R}(A^{(1)})\otimes \hat{1}_2.
\end{equation}

It follows that partial trace corresponds to partial trace also in the simulation,
\begin{equation}
\label{PartialTr}
 \Tr(A_1\Tr_{Q_2}\rho_{12}) =  \Tr\Big( \Theta_{\hat{1}}^{Q_1:R}(A_1) \Tr_{Q_2}\Theta_{\xi}^{Q_1Q_2:R}(\rho_{12})\Big)
\end{equation}
and this leads to
\begin{equation}
\Tr_{Q_2}\Theta_{\xi}^{Q_1Q_2:R}(\rho_{12}) \sim \Theta_{\xi}^{Q_1:R}(\Tr_{Q_2}\rho_{12}).
\end{equation}

One may note that the explicit form of the partial trace of the representation $\Theta_{\xi}^{Q_1Q_2:R}(\rho_{12})$ is 
\begin{equation*}
\begin{split}
 \Tr_{Q_2}\Theta_{\xi}^{Q_1Q_2:R}(\rho_{12}) 
 = & \sum_{j,j',k}P^{(1)}_{j} \Tr_2 (P^{(2)}_{k}\rho_{12})P^{(1)}_{j'} \\
 & \otimes\Delta^{-z^{(1)}_{j}-z^{(2)}_{k}}\xi {\Delta^{-z^{(1)}_{j'}-z^{(2)}_{k}}}^{\dagger},
\end{split}
\end{equation*}
while 
\begin{equation*}
\Theta_{\xi}^{Q_1:R}(\Tr_{Q_2}\rho_{12}) = \sum_{j,j'}P^{(1)}_{j} \Tr_2 (\rho_{12})P^{(1)}_{j'} \Delta^{-z^{(1)}_{j}}\xi {\Delta^{-z^{(1)}_{j'}}}^{\dagger}.
\end{equation*} 
Hence, although equivalent in terms of $\sim$, it is generally the case that  $\Tr_{Q_2}\Theta_{\xi}^{Q_1Q_2:R}(\rho_{12})\neq \Theta_{\xi}^{Q_1:R}(\Tr_{Q_2}\rho_{12})$.

\subsection{Appending systems with definite energy}
Again we assume $Q = Q_1Q_2$, with the same notation
 as in the previous section.
Suppose that $\rho_2\in\mathcal{S}(\mathcal{H}_2)$ has a distinct energy, i.e., $P^{(2)}_{t}\rho_2 P^{(2)}_{t} = \rho_2$ for some $t$, with corresponding energy eigenvalue $sz_{t}$. Then
\begin{equation}
\label{Append}
\begin{split}
\rho_{2}\otimes \Theta^{Q_1:R}_{\xi}(\rho_1)  = & \Theta^{Q_2Q_1:R}_{\xi'}(\rho_2\otimes \rho_1),\\
 \xi' := & \Delta^{z^{(2)}_t} \xi {\Delta^{z^{(2)}_t}}^{\dagger}.
 \end{split}
\end{equation}

Note that the transformation from $\xi$ to $\xi'$ is merely a rigid translation in energy, and does intuitively not change the `degree' of coherence in the state. Moreover
\begin{equation}
\label{snvdfkjnv}
\rho_{2}\otimes \Theta^{Q_1:R}_{\xi}(\rho_1) \sim  \Theta^{Q_2Q_1:R}_{\xi}(\rho_2\otimes \rho_1).
\end{equation}
One may note that in the case when $\rho_2$ is a convex combination of states of definite energy, i.e., if $\rho_2 = \sum_{t}\lambda_{t}\rho_{2}^{(t)}$ with $P^{(2)}_{t}\rho_2^{(t)} P^{(2)}_{t} = \rho_2^{(t)}$, then (\ref{snvdfkjnv}) still holds, although we cannot relate the two states by any operation on the reference $R$ alone.

\subsection{Reformulation of strong catalytic coherence within the simulation}
We are now ready to rephrase catalytic coherence in terms of correlations with the reference system. We  go directly for the strong version in Section \ref{StongerCatalyticProperty} rather than the weaker version in Section \ref{ACatalyticProperty},  because the channel formulation does not translate very elegantly into this picture. 

Let $Q = S_2S_1CE$, with joint Hamiltonian $H_{Q} =H_{S_2} + H_{S_1} + H_{C}+ H_E$, with $H_{S_1}\in H_{s}(\mathcal{H}_{S_1})$, $H_{S_2}\in H_{s}(\mathcal{H}_{S_2})$, $H_C\in H_s(\mathcal{H}_C)$, and $H_{E}$ our standard energy reservoir Hamiltonian as in equation (\ref{DefHEreservoir}), and the same for $H_R$. In Sec.~\ref{StongerCatalyticProperty} we allowed arbitrary $H_C$, but for the sake of simplicity with respect to the formalism, we here restrict to $H_C\in H_s(\mathcal{H}_C)$.

Proposition \ref{StrongCatalytic} does in essence compare two different procedures. In the first procedure, the energy reservoir interacts with system $S_1$ before it interacts with $S_2$, while in the second procedure it only interacts with $S_2$. Proposition \ref{StrongCatalytic} shows that these two procedures lead to the same state on $S_2C$. In the following we shall translate and compare these two procedures, and we begin with the two-step version.  

\begin{enumerate}
\item The referee prepares the state $\Theta^{S_2S_1CE:R}_{\xi}(\eta)$.
\item The referee gives us systems $E$ and $S_1$, and we perform the operation $V^{H_{S_1}}(U_1)$ on $S_1E$.
\item The referee gives us system $S_2$, and we perform the operation $V^{H_{S_2}}(U_2)$ on $S_2E$.
\item We give back system $S_2$ to the referee, who checks the resulting state by measurements on the form $\Theta^{S_2C:R}_{\hat{1}}(A)$.   
\end{enumerate}
Due to the fact that $V_{1}:= V^{H_{S_1}}(U_1)$ and $V_{2}:= V^{H_{S_2}}(U_2)$ are energy conserving it follows that 
\begin{equation*}
V_{2}V_{1}\Theta_{\xi}^{S_2S_1CE:R}(\eta) V_{1}^{\dagger} V_{2}^{\dagger}
= \Theta_{\xi}^{S_2S_1CE:R}(V_{2}V_{1}\eta V_{1}^{\dagger} V_{2}^{\dagger}),
\end{equation*}
i.e., the above described physical procedure results in the same state as if the referee had performed the two operations $V_1$ and $V_2$ before $S_2S_1E$ are given to us. In the particular case we consider here, we only give back system $S_2$ to the referee, and by using  the results of Section \ref{RefFrPartialTr} one can conclude that the state $\eta_{S_2CR}^{\textrm{two step}}$ on $S_2CR$ satisfies the following
  \begin{equation}
  \label{TwoStep}
\begin{split}
\eta_{S_2CR}^{\textrm{two step}}  := & \Tr_{S_1E}(V_{2}V_{1}\Theta_{\xi}^{S_2S_1CE:R}(\eta) V_{1}^{\dagger} V_{2}^{\dagger}) \\
 \sim & \Theta_{\xi}^{S_2C:R}(\Tr_{ES_1}V_{2}V_{1}\eta V_{1}^{\dagger} V_{2}^{\dagger}).
 \end{split}
\end{equation}
Hence, measurements on the form $\Theta_{\hat{1}}^{S_2C:R}(A)$ on the state $\eta_{S_2CR}^{\textrm{two step}}$ simulates measurements $A$ on $\Tr_{ES_1}(V_{2}V_{1}\eta V_{1}^{\dagger} V_{2}^{\dagger})$.

We can now compare the above procedure with the following. 
 \begin{enumerate}
\item The referee prepares the state $\Theta^{S_2CE:R}_{\xi}(\Tr_{S_2}\eta)$.
\item The referee gives us systems $E$ and $S_2$, and we perform the operation $V^{H_{S_2}}(U_2)$ on $S_2E$.
\item We give back system $S_2$ to the referee, who checks the resulting state by measurements on the form $\Theta^{S_2C:R}_{\hat{1}}(A)$.   
\end{enumerate}
Similarly to the above case, one can see that the resulting state  $\tilde{\eta}_{S_2CR}^{\textrm{one step}}$ on $S_2CR$ satisfies 
 \begin{equation}
 \label{OneStep}
 \begin{split}
\tilde{\eta}_{S_2CR}^{\textrm{one step}} := & \Tr_{E}(V_{2}\Theta_{\xi}^{S_2CE:R}(\Tr_{S_1}\eta) V_{2}^{\dagger}) \\
 \sim & \Theta_{\xi}^{S_2C:R}(\Tr_{E}V_{2} \Tr_{S_1}(\eta) V_{2}^{\dagger}).
\end{split}
\end{equation}
By comparing equation (\ref{OneStep}) with (\ref{TwoStep}) and Proposition \ref{StrongCatalytic}, we can conclude that the two procedures are equivalent with respect to what the referee can see in terms of measurements $\Theta^{S_2C:R}_{\hat{1}}(A)$. In other words
 \begin{equation}
 \eta_{S_2CR}^{\textrm{two step}} \sim \eta_{S_2CR}^{\textrm{one step}}. 
 \end{equation}
 Hence, from the point of view of the simulation, we have reconstructed the result of Proposition \ref{StrongCatalytic}.

 \subsection{\label{SeqSeparable}Sequential preparations from states with definite energy}
 Here we focus on the special case of preparations of a sequence of states on systems that initially have definite energies (i.e., analogous to what we did in Section \ref{AnExample}). In this case the initial coherence is only carried by the energy reservoir $E$. Restated in the simulation picture, this means that initially only $E$ and $R$ are correlated.

In Section \ref{AnExample} we have seen that it is possible to create arbitrarily many systems $S_k$ that all have the same reduced density operator, via a sequential application of the energy reservoir. Furthermore, if the coherence in $E$ is sufficiently strong, we can make these states very close to pure superpositions of energy eigenstates. 
In view of the examples in Section \ref{RefFrExamples} this may appear paradoxical. There we saw that  a pure superposition $|\phi\rangle$ in the case of $\xi = |0\rangle\langle 0|$ corresponds to a state $\Theta_{\xi}^{S:R}(|\phi\rangle\langle\phi|)$ that is highly entangled between $S$ and $R$. This may thus seem to suggest that we would be able to entangle arbitrarily many systems $S_1,S_2,\ldots$, with the reference $R$. However, this can  not be the case, as this would increase the degree of entanglement between $S$ and $R$ by local operations. 
The resolution of this apparent paradox lies in the equivalence relation $\sim$. More precisely, there are many  states $\mu$ on $SR$  such that $\mu \sim \Theta^{S:R}_{\xi}(\rho)$ and thus also representing $\rho$, while not necessarily being entangled. In the following we shall show that this indeed is the case for sequential preparations.

Let systems $S_1,\ldots, S_M$ have the Hamiltonians $H_{S_k}\in H_s(\mathcal{H}_{S_k})$. Let $\{P^{(k)}_{n_k}\}_{n_k=1}^{N_k}$ denote the eigen-projectors of $H_{S_k}$ with corresponding distinct eigenvalues $sz^{(k)}_{n_k}$. 
The initial states on these systems are $\rho_k\in\mathcal{S}(\mathcal{H}_{S_k})$, such that $P^{(k)}_{t_k}\rho_{k} P^{(k)}_{t_k} = \rho_{k}$, i.e., system $S_k$ is initially in a state with definite energy $sz^{(k)}_{t_k}$. As usual we let the energy reservoir $E$ have the Hamiltonian $H_E = s\sum_{j\in\mathbb{Z}}j|j\rangle\langle j|_E$. On each system $S_k$ we pick a unitary operator $U_k$ and construct the unitary operation
\begin{equation}
V :=  V^{H_{S_M}}(U_M)\cdots V^{H_{S_1}}(U_1).
\end{equation}
The first step in the procedure is to prepare the coherence resource $\Theta^{E:R}_{\xi}(\sigma)$. Next, the systems $S_1,\ldots,S_M$ are appended and the sequence of operations $V$ is applied, which result in the state
\begin{equation}
\eta:=  V \rho_{M}\otimes \cdots \otimes\rho_{1}\otimes \Theta^{E:R}_{\xi}(\sigma) V^{\dagger}.
\end{equation}
The reduced state on $S_1\cdots S_M$ and the reference $R$ reads
\begin{equation*}
\begin{split}
 \Tr_E\eta   =  & \sum_{j,j'} \sum_{n_M,m_M} \cdots \sum_{n_1,m_1} \sigma_{j,j'}\\
 & \times \delta_{ j' -z^{(M)}_{m_M} -\cdots -z^{(1)}_{m_1}, j-z^{(M)}_{n_M} -\cdots -z^{(1)}_{n_1} } \\
 & \times P^{(M)}_{n_M}U_M\rho_M U_M^{\dagger}P^{(M)}_{m_M}   \otimes\cdots\otimes  P_{n_1}^{(1)}U_1\rho_1 U_1^{\dagger}P_{m_1}^{(1)} \\
 & \otimes \Delta^{-j}\xi {\Delta^{-j'}}^{\dagger}.
\end{split}
\end{equation*}
For $x\in \{1,\ldots, M\}$ the reduced state on $S_{x}R$ can be written
\begin{equation}
\begin{split}
\eta_{S_xR}  =  & \sum_{j,j'}\sigma_{j,j'} \sum_{n_x,m_x} \delta_{j' -z^{(x)}_{m_x}, j-z^{(x)}_{n_x}} \\ 
& \times P_{n_x}^{(x)}U_x\rho_x U_x^{\dagger}P_{m_x}^{(x)}  \otimes \Delta^{-j}\xi {\Delta^{-j'}}^{\dagger}. 
 \end{split}
 \end{equation}
From the previous sections we know that $\eta_{S_xR}$ represents $\rho_{S_x}$, i.e., 
 \begin{equation*}
 \begin{split}
\eta_{S_xR}  \sim & \Theta^{S_x:R}_{\xi}(\rho_{S_x}),\\
 \rho_{S_x}  := & \Tr_{S_1\cdots S_{x+1}S_{x-1}\cdots S_1E}(V \rho_{M}\otimes \cdots \otimes\rho_{1}\otimes\sigma V^{\dagger}).
\end{split}
\end{equation*}
We shall next prove that $\eta_{S_xR} $ always is a separable state.

Let us recall the characterization of bipartite separable states as those that are infinitely symmetrically extendible. (For a quick overview, see e.g. \cite{Brandao12}). A bipartite state $\eta$ on $SR$ is called symmetrically $M$-extendible, if there exists a state $\rho$ on $S^{\otimes M}R$ such that $\eta = \Tr_{S^{\otimes M}} \rho$ and that $\rho$ is invariant under all permutations of the $M$ subsystems $S$. All separable states have a trivial symmetric extension for every $M$. Thus, if a state fails to be symmetrically extendible for some $M$ it must be entangled \cite{Doherty04}. 
More generally it turns out that a state is separable if and only if it has symmetric $M$-extensions for all $M$, see Theorem 1 in \cite{Doherty04}. (This also follows from various versions of the quantum de-Finetti theorem \cite{Stromer69,Hudson76,Raggio89,Werner89,Caves02,Konig05,Christandl07}. See also \cite{Fannes88}.)

Sequential preparations provide a construction of a symmetric extension of $\eta_{S_xR}$ for every $M$. 
Consider a new preparation procedure where we prepare the same state on a sequence of copies $\tilde{S}_1,\ldots, \tilde{S}_{M}$ of $S_{x}$ with $\tilde{H}_{S_1} = \cdots =\tilde{H}_{S_M} = H_{S_{x}}$. In other words, let $\tilde{N} =  \tilde{N}_1 = \cdots = \tilde{N}_M:= N_x$, and $\tilde{z}_{n}  = \tilde{z}^{(1)}_{n} = \cdots = \tilde{z}^{(M)}_{n}:= z^{(x)}_{n}$, $\tilde{\rho} = \tilde{\rho}_{1} = \cdots  = \tilde{\rho}_{M} := \rho_{x}$, and $\tilde{U} = \tilde{U}_{1} = \cdots = \tilde{U}_{M} := U_{x}$, which yields the global state 
\begin{equation*}
\begin{split}
\Tr_E\tilde{\eta}  = & \sum_{j,j'} \sum_{n_M,m_M} \cdots \sum_{n_1,m_1}\sigma_{j,j'} \\
& \times\delta_{j' -\tilde{z}_{m_M} -\cdots -\tilde{z}_{m_1}, j-\tilde{z}_{n_M} -\cdots -\tilde{z}_{n_1}} \\
 &\times P_{n_M}\tilde{U}\tilde{\rho} \tilde{U}^{\dagger}P_{m_M}   \otimes\cdots\otimes  P_{n_1}\tilde{U}\tilde{\rho} \tilde{U}^{\dagger}P_{m_1} \\
 &\otimes \Delta^{-j}\xi {\Delta^{-j'}}^{\dagger}.
 \end{split}
\end{equation*}
As seen, this state is invariant under permutations of the subsystems $\tilde{S}_1,\ldots,\tilde{S}_M$. Furthermore, if we trace out all subsystems $\tilde{S}$ except one single, we obtain $\eta_{S_xR}$ by construction. Since this is true for all $M$, we can conclude that $\eta_{S_xR}$ is infinitely symmetrically extendible and thus separable.

As a side-remark one may note that this proof holds for any choice of $\xi$, and is not limited to $\xi$ that are diagonal with respect to $H_R$.

%%%%%%%%%%%%%%%%%%%%%%
%            Half-infinite ladder                   %
%%%%%%%%%%%%%%%%%%%%%%

\section{\label{Harmonic} Half-infinite ladder (harmonic oscillator)}
As mentioned earlier, a problem with the doubly-infinite energy ladder is that it is somewhat unphysical, as it has no ground state energy. However, by `cutting away' the lower half of the spectrum we get the spectrum of the harmonic oscillator. Here we show that we can, with some modifications, reconstruct the catalytic properties of the doubly-infinite energy ladder in this harmonic oscillator model.

\subsection{The half-infinite ladder model}

Here we use an harmonic oscillator as model  (or any Hamiltonian that is iso-spectral to the harmonic oscillator) 
\begin{equation}
H_E^{+} := s\sum_{j=0}^{+\infty} j|j\rangle\langle j|.
\end{equation}
As before, we define an $N$-level system with a Hamiltonian $H_S\in H_s(\mathcal{H}_S)$, with eigenvalues $h_n = sz_n$. We also let
\begin{equation}
z_{\mathrm{min}}:= \min_{n=1,\ldots,N} z_n,\quad z_{\mathrm{max}}:= \max_{n=1,\ldots,N}z_n.
\end{equation}
The projectors onto the eigenspaces of $H_S + H_E^{+}$ can be written
\begin{equation}
\label{Pplusdef}
P^{(l)}_{+} := \sum_{n = 1,\ldots,N: l\geq z_{n}}|\psi_n\rangle\langle\psi_n|\otimes|l-z_n\rangle\langle l-z_n|,
\end{equation}
for all $l\geq z_{\mathrm{min}}$. (It is maybe slightly annoying that we let the index $l$ start at $z_{\mathrm{min}}$ rather than at zero, but we do this for the sake of compactness of formulas and coordination with the notation in the previous sections.) 
Note that 
\begin{equation}
P^{(l)}_{+} = \sum_{n = 1,\ldots,N}|\psi_n\rangle\langle\psi_n|\otimes|l-z_n\rangle\langle l-z_n|,\quad \forall l\geq z_{\textrm{max}}.
\end{equation}
In other words, $P^{(l)}_{+} = P^{(l)}$ for $l\geq z_{\mathrm{max}}$. 
(We will in this section often compare operators on the spaces $\mathcal{H}_E^{+}$ and $\mathcal{H}_E$, where the former can be regarded as a subspace of the latter. Without further comments we do in these cases assume that the operators on $\mathcal{H}_E^{+}$ are extended trivially to the orthogonal complement of $\mathcal{H}_E^{+}$ in $\mathcal{H}_E$, i.e., they act trivially on the negative half-ladder.) 

For every $U\in U(\mathcal{H}_S)$ we define the unitary operator
\begin{equation}
\label{tzuitzuit}
\begin{split}
V_{+}(U)  :=  & \sum_{l\geq z_{\mathrm{max}}}V_{l}(U) + \sum_{l=z_{\mathrm{min}}}^{z_{\mathrm{max}}-1}X_l,\\
&  X_{l}X_{l}^{\dagger} = X_{l}^{\dagger}X_{l} = P_{+}^{(l)},
\end{split}
\end{equation}
where 
\begin{equation}
\label{svrtbsb}
\begin{split}
V_{l}(U)  := &  \sum_{n,n'=1}^{N}|\psi_{n}\rangle\langle\psi_{n}|U|\psi_{n'}\rangle\langle\psi_{n'}|\otimes|l-z_n\rangle\langle l-z_{n'}|,\\
& l\geq z_{\mathrm{max}}
\end{split}
\end{equation}
is as in equation (\ref{nfkjnb}). The unitarity of $V_{+}(U)$ follows from the fact that $V_{l}(U)V_{l}(U)^{\dagger} = V_{l}(U)^{\dagger}V_{l}(U) = P_{+}^{(l)}$. We will not pay any particular attention to the choice of the operators $X_l$, as we will focus on states outside the supports of these operators. As seen, we have, by the very construction, made sure that $V_{+}(U)$ in (\ref{nfkjnb}) and $V(U)$ in (\ref{tzuitzuit}) act identically for states with support sufficiently far away from the ground state. The rest of this section is devoted to the formalization of  the simple idea that if we start the energy reservoir in a state that is sufficiently far from the ground state, we can maintain it so by pumping energy into the reservoir. By  this we retrieve all the relevant properties of the doubly-infinite ladder model.

Analogous to $\Delta$ we define the operator
\begin{equation}
\Delta_{+} := \sum_{j=0}^{+\infty}|j+1\rangle\langle j|.
\end{equation}
Note that $\Delta_{+}$ is not unitary, but rather is a partial isometry that takes $\mathcal{H}^{+}_{E}$ into the subspace $\Sp\{|j\rangle\}_{j\geq 1}$.

\subsection{\label{SufficientlyFar}Sufficiently far from the ground state the two models are equivalent}
Analogous to the channel $\Phi_{\sigma,U}$ in Section \ref{Sec:inducedch} we define
\begin{equation}
\label{PluschannelOnS}
\Phi^{ +}_{\sigma,U}(\rho): =   \Tr_E V_{+}(U)\rho\otimes\sigma V_{+}(U)^{\dagger}.
\end{equation}
Similarly, we denote the set of channels that can be reached on $S$, using $\sigma$ as a resource, by
\begin{equation}
\mathcal{C}^{S,H_S}_{+}(\sigma) := \Big\{\Phi^{+}_{\sigma,U}: U\in U(\mathcal{H}_S)\Big\}.
\end{equation}
The effect of $V_{+}(U)$ on the energy reservoir is expressed via the channel 
\begin{equation}
\label{PluschannelOnE}
\Lambda^{+}_{\rho,U}(\sigma) :=  \Tr_{S}V_{+}(U)\rho\otimes \sigma V_{+}(U)^{\dagger}.
\end{equation} 
Define the family of projectors 
\begin{equation}
P_{\geq m} := \sum_{j\geq m}|j\rangle\langle j|, \quad m\geq 0. 
\end{equation}

\begin{Lemma}
\label{equivalence}
Let $s>0$, $\mathcal{H}_S$ finite-dimensional, and $H_S\in H_s(\mathcal{H}_S)$ with maximal eigenvalue $sz_{\textrm{max}}$ and minimal eigenvalue $sz_{\textrm{min}}$.
Let $\sigma\in\mathcal{S}(\mathcal{H}_E)$ be such that $P_{\geq z_{\textrm{max}}-z_{\textrm{min}}}\sigma P_{\geq z_{\textrm{max}}-z_{\textrm{min}}} = \sigma$, then 
\begin{equation}
\Phi^{+}_{\sigma,U}  =  \Phi_{\sigma,U},\quad\quad \Lambda^{+}_{\rho,U}(\sigma)  =  \Lambda_{\rho,U}(\sigma),
\end{equation}
for all $U\in U(\mathcal{H}_S)$.
\end{Lemma}
\begin{proof}
By comparing with the definition of $P_{+}^{(l)}$ in (\ref{Pplusdef}) one can see that 
\begin{equation}
\label{bhjvvdf}
P_{+}^{(l)}[\hat{1}_S\otimes P_{\geq z_{\textrm{max}}-z_{\textrm{min}}}] = 0,\quad z_{\textrm{min}}\leq l< z_{\textrm{max}}.
\end{equation}
Analogously 
\begin{equation}
P^{(l)}[\hat{1}_S\otimes P_{\geq z_{\textrm{max}}-z_{\textrm{min}}}] = 0,\quad  l< z_{\textrm{max}}.
\end{equation}
Note that $P_{+}^{(l)} = P^{(l)}$ and $V_{l}(U)P^{(l)} = V_{l}(U)$ for $l\geq z_{\textrm{max}}$,  and $X_lP_{+}^{(l)} = X_l$.  By comparing (\ref{bijection}), (\ref{nfkjnb}), and (\ref{tzuitzuit}), it follows that 
\begin{equation}
\label{bhjvdf}
V_{+}(U)[\hat{1}_S\otimes P_{\geq z_{\textrm{max}}-z_{\textrm{min}}}] = V(U)[\hat{1}_S\otimes P_{\geq z_{\textrm{max}}-z_{\textrm{min}}}].
\end{equation}
The statements of the lemma follows directly from this.
\end{proof}

\begin{Lemma}
\label{qweqwer}
Let $s>0$, $\mathcal{H}_S$ finite-dimensional, and $H_S\in H_s(\mathcal{H}_S)$ with maximal eigenvalue $sz_{\textrm{max}}$ and minimal eigenvalue $sz_{\textrm{min}}$. 
For $m\geq z_{\textrm{max}}-z_{\textrm{min}}$
\begin{equation}
 V_{+}(U)[\hat{1}_S\otimes P_{\geq m}]  = [\hat{1}_S\otimes P_{\geq m-z_{\textrm{max}}+z_{\textrm{min}}}]V_{+}(U)[\hat{1}_S\otimes P_{\geq m}].
\end{equation}
Hence, if $\sigma\in\mathcal{S}(\mathcal{H}_E)$ is such that $P_{\geq m}\sigma P_{\geq m} = \sigma$, then 
\begin{equation}
\begin{split}
& P_{\geq m-z_{\textrm{max}}+z_{\textrm{min}}}\Lambda^{+}_{\rho,U}(\sigma)P_{\geq m-z_{\textrm{max}}+z_{\textrm{min}}}  = \Lambda^{+}_{\rho,U}(\sigma), \\
 & \quad\quad\quad\quad\quad\quad  \forall \rho\in\mathcal{S}(\mathcal{H}_S),\quad \forall U\in U(\mathcal{H}_S).
\end{split}
\end{equation}
\end{Lemma}
\begin{proof}
Since $m\geq z_{\textrm{max}}-z_{\textrm{min}}$ it follows by equation (\ref{bhjvvdf}) that 
\begin{equation}
\begin{split}
& V_{+}(U)[\hat{1}_S\otimes P_{\geq m}]  \\
& =  \sum_{l\geq z_{\textrm{max}}}\sum_{n,n'=1}^{N}|\psi_{n}\rangle\langle\psi_{n}|Q|\psi_{n'}\rangle\langle\psi_{n'}|\otimes|l-z_n\rangle\langle l-z_{n'}|P_{\geq m}.
\end{split}
\end{equation}
Let $|x\rangle$ be an energy eigenstate of the reservoir. For $\langle x|V_{+}(U)[\hat{1}_S\otimes P_{\geq m}]\neq 0$, a necessary condition is that $x = l-z_{n}$ and $l-z_{n'}\geq m$ for some $l,n,n'$. We can thus write
$x + z_{\textrm{max}} \geq x + z_n = l \geq m + z_{n'} \geq m + z_{\textrm{min}}$. Hence,  $\langle x|V_{+}(U)[\hat{1}_S\otimes P_{\geq m}] = 0$ for all  $x < m-z_{\textrm{max}}  + z_{\textrm{min}}$.
It follows that $ [\hat{1}_S\otimes (\hat{1}_E-P_{\geq m-z_{\textrm{max}}+z_{\textrm{min}}})]V_{+}(U)[\hat{1}_S\otimes P_{\geq m}] =0$. This proves the lemma.
\end{proof}

The following Proposition tells us that a reservoir with a state that is $L$ times removed from the ground state with respect to the maximal possible energy change $D$, can be used $L$ times before there is any noticeable difference in the set of channels that can be reached by using the reservoir. We refer to this as a quasi-catalytic property.
Here $\lambda_{\textrm{max}}(H)$ and $\lambda_{\textrm{min}}(H)$ denote the maximal and minimal eigenvalues of the Hermitian operator $H$.
\begin{Proposition}[Quasi-Catalytic states]
\label{QuasiCatalytic}
Let $s>0$, $D\in\mathbb{N}$, and let $\mathcal{H}_{S_1},\ldots,\mathcal{H}_{S_L}$ be finite-dimensional, with $H_{S_l}\in H_s(\mathcal{H}_{S_l})$, such that $\lambda_{\textrm{max}}(H_{S_l})-\lambda_{\textrm{min}}(H_{S_l})\leq sD$ for $l=1,\ldots, L$.
Suppose $\sigma^{(1)}\in\mathcal{S}(\mathcal{H}_E)$ is such that
\begin{equation}
P_{\geq LD}\sigma^{(1)}P_{\geq LD} = \sigma^{(1)}.
\end{equation}
Define the sequence $\sigma^{(l+1)} =\Lambda^{+}_{\rho_l,U_l}(\sigma^{(l)})$, for $l=1,\ldots, L-1$, for $\rho_l\in\mathcal{S}(\mathcal{H}_{S_l})$ and $U_l\in U(\mathcal{H}_{S_l})$.
Then 
\begin{equation}
\mathcal{C}^{S_l,H_{S_l}}_{+}(\sigma^{(l)}) = \mathcal{C}^{S_l,H_{S_l}}_{+}(\sigma^{(1)}),\quad l = 1,\ldots,L.
\end{equation}
\end{Proposition}
\begin{proof}
By a repeated use of Lemma \ref{qweqwer} it follows that $P_{\geq (L-l+1)D}\sigma^{(l)} P_{\geq (L-l+1)D} = \sigma^{(l)}$. Thus, by Lemma \ref{equivalence} we know that $\Lambda^{+}_{\rho_l,U_l}(\sigma^{(l)}) = \Lambda_{\rho_l,U_l}(\sigma^{(l)})$. By Lemma \ref{equivalence} we can also conclude that $\mathcal{C}_{+}^{S_l,H_{S_l}}(\sigma^{(l)}) = \mathcal{C}^{S_l,H_{S_l}}(\sigma^{(l)})$. Furthermore, by Proposition \ref{WeakCatalytic} it follows that $\mathcal{C}^{S_l,H_{S_l}}\boldsymbol{(}\Lambda_{\rho_{l-1},U_{l-1}}(\sigma^{(l-1)})\boldsymbol{)}
=   \mathcal{C}^{S_l,H_{S_l}}(\sigma^{(l-1)})$. By combining the above observations, we see that 
\begin{equation}
\begin{split}
\mathcal{C}_{+}^{S_l,H_{S_l}}(\sigma^{(l)})  = &  \mathcal{C}_{+}^{S_l,H_{S_l}}\boldsymbol{(}\Lambda^{+}_{\rho_{l-1},U_{l-1}}(\sigma^{(l-1)})\boldsymbol{)}\\
 = &  \mathcal{C}^{S_l,H_{S_l}}\boldsymbol{(}\Lambda^{+}_{\rho_{l-1},U_{l-1}}(\sigma^{(l-1)})\boldsymbol{)}\\
 = &  \mathcal{C}^{S_l,H_{S_l}}\boldsymbol{(}\Lambda_{\rho_{l-1},U_{l-1}}(\sigma^{(l-1)})\boldsymbol{)}\\
 = &  \mathcal{C}^{S_l,H_{S_l}}(\sigma^{(l-1)})\\
 = &  \mathcal{C}_{+}^{S_l,H_{S_l}}(\sigma^{(l-1)}).
 \end{split}
\end{equation}
The above reasoning can now be iterated, which yields the statement of the proposition.
\end{proof}

\subsection{\label{Regenerative}Regenerative cycles}
Here we construct a protocol to maintain the set of reachable channels $\mathcal{C}_{+}$. Basically, what we do is to keep the state of the energy reservoir sufficiently far away from the ground state by pumping energy into it.

The following lemma tells us that we can implement the rigid translation $\sigma \mapsto \Delta_{+}^{m}\sigma {\Delta_{+}^{m}}^{\dagger}$ within our model, if we have access to an ancillary system in a pure excited state.

By the observation that $P^{(l)}_{+}|a_{\textrm{max}}\rangle|j\rangle =  0$ for $z_{\textrm{max}}>l\geq z_{\textrm{min}}$ and all $j\geq 0$ we can directly use (\ref{tzuitzuit}) and (\ref{svrtbsb}) to prove the following Lemma.
\begin{Lemma}
\label{RigidTranslation}
Let $H_A\in H_s(\mathcal{H}_A)$ have the largest eigenvalue $sz_{\mathrm{max}}$ with corresponding eigenstate $|a_{\mathrm{max}}\rangle$, and smallest eigenvalue $sz_{\mathrm{min}}$ with corresponding eigenstate $|a_{\mathrm{min}}\rangle$. Assume $z_{\mathrm{max}}\neq z_{\mathrm{min}}$. Let $U_{A}\in U(\mathcal{H}_A)$ be such that 
\begin{equation}
U_{A}|a_{\mathrm{max}}\rangle = |a_{\mathrm{min}}\rangle.
\end{equation}
Then
\begin{equation}
\Lambda^{+}_{|a_{\mathrm{max}}\rangle\langle a_{\mathrm{max}}|,U_A}(\sigma) = \Delta_{+}^{z_{\mathrm{max}}-z_{\mathrm{min}}}\sigma {\Delta_{+}^{z_{\mathrm{max}}-z_{\mathrm{min}}}}^{\dagger},
\end{equation}
for all $\sigma\in\mathcal{S}(\mathcal{H}_E^{+})$
where $\Lambda^{+}$ is as defined in equation (\ref{PluschannelOnE}).
\end{Lemma}
One may note that in this implementation there is no need to avoid any `border zone'.

In essence, the following proposition tells us that if the state of the energy reservoir has a support that is sufficiently far from the ground state, then we can maintain its power to induce channels on $S$ (as characterized by the set $\mathcal{C}^{S,H_S}_{+}$) merely by feeding energy into it between each use.
\begin{Proposition}[Protocol for a catalytic cycle]
\label{ProtocolCatalyticCycle}
Let $s>0$, $D\in\mathbb{N}$, $\mathcal{H}_{S_1}, \mathcal{H}_{S_2}$ be finite-dimensional, and $H_{S_1}\in H_s(\mathcal{H}_{S_1})$, $H_{S_2}\in H_s(\mathcal{H}_{S_2})$, where $\lambda_{\mathrm{max}}(H_{S_1})-\lambda_{\mathrm{min}}(H_{S_1}) \leq sD$ and $\lambda_{\mathrm{max}}(H_{S_2})-\lambda_{\mathrm{min}}(H_{S_2}) \leq sD$, with $D\in \mathbb{N}$.  
 Let $\sigma_{\mathrm{in}}\in\mathcal{S}(\mathcal{H}_E^{+})$ be such that 
\begin{equation}
\label{ynvjknnkj2}
P_{\geq 2D}\sigma_{\mathrm{in}} P_{\geq 2D} = \sigma_{\mathrm{in}}. 
\end{equation}
For an arbitrary $\rho_1\in\mathcal{S}(\mathcal{H}_{S_1})$ and an arbitrary $U_1\in U(\mathcal{H}_{S_1})$, let
\begin{equation}
\overline{\sigma} := \Lambda^{+}_{\rho_1,U_1}(\sigma_{\mathrm{in}}), 
\end{equation}
where the channel $\Lambda^{+}_{\rho_1,U_1}$ is as defined in equation (\ref{PluschannelOnE}).
Let $\mathcal{H}_A$ be two-dimensional, with $H_A\in H_s(\mathcal{H}_A)$ such that $H_A$ has the largest eigenvalue $sD$ with corresponding eigenvector $|a_1\rangle$, and lowest eigenvalue $0$ with corresponding eigenvector $|a_0\rangle$.
Let $U_A = |a_1\rangle\langle a_0| + |a_0\rangle\langle a_1|$, and define
\begin{equation}
\sigma_{\mathrm{out}} := \Lambda^{+}_{|a_1\rangle\langle a_1|,U_A}(\overline{\sigma}),
\end{equation}
where $\Lambda^{+}_{|a_{1}\rangle\langle a_{1}|,U_A}$ is as defined in equation (\ref{PluschannelOnE}).
Then
\begin{equation}
\begin{split}
& \mathcal{C}^{S_2,H_{S_2}}_{+}(\sigma_{\mathrm{out}}) = \mathcal{C}^{S_2,H_{S_2}}(\sigma_{\mathrm{out}}) \\
&=  \mathcal{C}^{S_2,H_{S_2}}(\sigma_{\mathrm{in}})  = \mathcal{C}^{S_2,H_{S_2}}_{+}(\sigma_{\mathrm{in}})
\end{split}
\end{equation}
and
\begin{equation}
\label{zuoijtozkn2}
P_{\geq 2D}\sigma_{\mathrm{out}} P_{\geq 2D} = \sigma_{\mathrm{out}}. 
\end{equation}
\end{Proposition}

\begin{proof}
Combining (\ref{ynvjknnkj2}) and Lemma \ref{qweqwer} yields  
\begin{equation}
\label{buitbnitubn}
P_{\geq D}\overline{\sigma} P_{\geq D} = \overline{\sigma}.
\end{equation}
By further combining this with Lemma \ref{RigidTranslation} gives (\ref{zuoijtozkn2}). Due to the latter we can, by Lemma \ref{equivalence}, conclude that  
\begin{equation}
\mathcal{C}_{+}^{S_2,H_{S_2}}(\sigma_{\textrm{out}}) = \mathcal{C}^{S_2,H_{S_2}}(\sigma_{\textrm{out}}).
\end{equation}
By (\ref{buitbnitubn}) and Lemma \ref{qweqwer} it follows that $\sigma_{\textrm{out}} = \Lambda^{+}_{|a_1\rangle\langle a_1|,U_A}(\overline{\sigma}) =  \Lambda_{|a_1\rangle\langle a_1|,U_A}(\overline{\sigma})$.
Hence Proposition \ref{WeakCatalytic} tells us that
\begin{equation}
 \mathcal{C}^{S_2,H_{S_2}}\boldsymbol{(}\Lambda_{|a_1\rangle\langle a_1|,U_A}(\overline{\sigma})\boldsymbol{)}
 =   \mathcal{C}^{S_2,H_{S_2}}(\overline{\sigma})
\end{equation}
Due to (\ref{ynvjknnkj2}) and Lemma \ref{qweqwer} it follows that $\Lambda^{+}_{\rho_1,U_1}(\sigma_{\mathrm{in}}) = \Lambda_{\rho_1,U_1}(\sigma_{\mathrm{in}})$. By Proposition \ref{WeakCatalytic} 
\begin{equation}
      \mathcal{C}^{S_2,H_{S_2}}\boldsymbol{(}\Lambda_{\rho_1,U_1}(\sigma_{\mathrm{in}})\boldsymbol{)}
   =    \mathcal{C}^{S_2,H_{S_2}}(\sigma_{\mathrm{in}}).
\end{equation}
Furthermore, due to (\ref{ynvjknnkj2}) it follows by Lemma \ref{equivalence} that 
\begin{equation}
 \mathcal{C}^{S_2,H_{S_2}}(\sigma_{\mathrm{in}})  = \mathcal{C}_{+}^{S_2,H_{S_2}}(\sigma_{\mathrm{in}}).
\end{equation}
By combining the above chain of equalities, we obtain the proposition.
\end{proof}

%%%%%%%%%%%%%%%
%  More general models   %
%%%%%%%%%%%%%%%

\section{\label{MoreGeneral}More general models}
In the previous sections we demonstrated that there in principle exist systems where coherence can be made catalytic. Here we briefly touch upon the question to what extent these phenomena survive in more general  systems.

\subsection{General energy preserving unitary operations}
Here we consider Hamiltonians on the half-infinite ladder model, but we relax the restriction on the dynamics, and only demand that it preserves the total energy.

One can realize that a unitary operator $V$ on the system and half-infinite ladder is block-diagonal with respect to the energy eigenspaces, i.e., $V = \sum_{l\geq z_{\textrm{min}}}P_{+}^{(l)}VP_{+}^{(l)}$, if and only if 
\begin{equation}
\label{jhvjvhj2}
\begin{split}
V  = & \sum_{l\geq z_{\mathrm{max}}} \sum_{n,n'=1}^{N} V_{n,n'}^{(l)}|\psi_{n}\rangle\langle\psi_{n'}|\otimes|l-z_n\rangle\langle l-z_{n'}| \\
 & + \sum_{l=z_{\mathrm{min}}}^{z_{\mathrm{max}}-1}X_l,\quad\quad  X_{l}X_{l}^{\dagger} = X_{l}^{\dagger}X_{l} = P_{+}^{(l)},
 \end{split}
\end{equation} 
where each $V^{(j)} = [V_{n,n'}^{(j)}]_{n,n'=1}^{N}$ is a unitary matrix.

Recall that the harmonic oscillator model in Section \ref{Harmonic} has the feature that it, sufficiently far from the ground state, essentially implements copies of a unitary interaction. We thus regain the model in Section \ref{Harmonic} if we  for all $l$ choose $V^{(l)} = V$ for a fixed unitary matrix $V$. This `uniformity' over the energy ladder appears to be closely connected to the catalytic properties of these models. It thus seems reasonable to suspect that in models that have some type of approximate uniformity, one should be able to find some approximate version of the type of behaviors that we have seen in the previous sections.

As an example, suppose that the set of unitary matrices $V^{(l_0)},\ldots, V^{(l_0+L)}$ would become more `similar' to each other as $l_0$ increases (for a fixed $L$). One way in which this could happen is if $V^{(l)}$ would converge to a specific limit matrix as $l\rightarrow + \infty$. However, this is not the only possibility. For example, let $H$ be an $N\times N$ Hermitian matrix, and define $V^{(l)} := \exp[-if(l)H]$, where $f:\mathbb{N}\rightarrow\mathbb{N}$ is such that $f(l+L)-f(l)\rightarrow 0$ as $l\rightarrow +\infty$. [An example of this  is $f(x) = \sqrt{x}$,  for which $\sqrt{x+L}-\sqrt{x}\sim L/(2\sqrt{x})$.] In this case $(V^{(l_0)},\ldots, V^{(l_0+L)}) \sim_{l_0\rightarrow \infty} (V^{(l_0)},\ldots,V^{(l_0)})$. Hence, for a limited range of levels, the transformations become asymptotically uniform (in spite the fact that they do not converge to a limit operator). In Section \ref{JCmodel} we show that the Jaynes-Cummings model satisfies this asymptotic uniformity.  

As a final remark we note that in Section \ref{Harmonic} we could maintain the capacity of the coherence by injecting energy into the reservoir. Furthermore, we did this solely `within' the model, i.e., we used the same class of unitaries $V_+(U)$ for the pumping, as for the operations on system $S$. It is a relevant question for which classes of these more general unitary operations that this is possible (at least approximately). However, we leave this as an open question.
 
\subsection{\label{JCmodel}The Jaynes-Cummings model}
Here we analyze numerically the Jaynes-Cummings model \cite{JC63,Shore93} of a two-level system in resonant interaction with an electromagnetic field mode.  As we shall see, numerical tests suggest that the higher the energy of a (suitably chosen) initial state of the reservoir is, the longer its coherence survives sequential interactions with a collection of two-level systems.

We consider the Jaynes-Cummings model \cite{JC63,Shore93} for the interaction between a two-level system $H_S = h_1|\psi_0\rangle\psi_0| + h_2|\psi_1\rangle\psi_1|$ and an electromagnetic field mode $H_E = \hbar \omega a^{\dagger}a$, with the interaction Hamiltonian
\begin{equation}
\label{TheJCHamiltonian}
H_{\mathrm{JC}} =  g\sigma_{+}\otimes a + g\sigma_{-}\otimes a^{\dagger}, 
\end{equation}
where $a,a^{\dagger}$ are the standard annihilation and creation operators $[a,a^{\dagger}] = \hat{1}_E$, $a = \sum_{l=1}^{\infty}\sqrt{l}|l-1\rangle\langle l|$ and $\sigma_{+} = |\psi_1\rangle\langle \psi_0|$, $\sigma_{-} = |\psi_0\rangle\langle \psi_1|$.

We furthermore assume resonant conditions (i.e., the `matching' of the energy levels), so that $h_1-h_0 = \hbar\omega$, and choose $h_1 = \hbar\omega z_1$, $h_0 = \hbar\omega z_0$, with $z_0 = 0, z_1= 1$.
By standard textbook calculations, where we introduce the unit-free evolution parameter $x := tg/\hbar$ for the evolution time $t$, the evolution induced by $H_{\mathrm{JC}}$ can be written

\begin{equation*}
\begin{split}
e^{-ixH_{\mathrm{JC}}}  = &  \sum_{l=1}^{\infty}\sum_{n,n'=0}^{1}V^{(l)}_{n,n'}|\psi_{n}\rangle\langle \psi_{n'}|\otimes |l-n\rangle\langle l-n'|\\
 &  + |\psi_0\rangle\langle \psi_0|\otimes |0\rangle\langle 0|, \\
V^{(l)}  := & \left[\begin{matrix} \cos(x\sqrt{l}) &  -i  \sin(x\sqrt{l}) \\
 -i  \sin(x\sqrt{l}) & \cos(x\sqrt{l})
\end{matrix}\right]  
\end{split}
\end{equation*}
where we can write
\begin{equation}
V^{(l)} = e^{-ix\sqrt{l}H},\quad H := \left[\begin{matrix} 0 &  1 \\
 1 & 0
\end{matrix}\right].
\end{equation}
Hence, the dynamics in the JC-model provides an example of the type of asymptotically uniform unitary operators discussed in the previous section.
The evolution on the reservoir is given by
\begin{equation*}
\begin{split}
\Lambda_{x,|\psi_0\rangle\langle\psi_0|}(\sigma)  := & \Tr_S (e^{-ixH_{\mathrm{JC}}}|\psi_0\rangle\langle\psi_0|\otimes\sigma e^{ixH_{\mathrm{JC}}})  \\
 = &  A_x \sigma A_{x}^{\dagger} + B_x \sigma B_x^{\dagger}, \\
 A_x  := & \sum_{l=0}^{\infty}\cos(x\sqrt{l})|l\rangle\langle l|,\\
 B_x   := & \sum_{l=1}^{\infty}\sin(x\sqrt{l})|l-1\rangle\langle l|.
 \end{split}
\end{equation*}
The sequential preparation procedure results in a evolution process on the reservoir, which corresponds to an iteration of the  channel $\Lambda_{x,|\psi_0\rangle\langle\psi_0|}$. The nature of this dynamics is depicted in Fig.~\ref{fig:dynamics}.

Here we attempt to mimic the sequential preparation procedure that we investigated in Section \ref{AnExample}.
Again we thus use the (square of) the fidelity to measure how well  the superposition $|\phi\rangle := (|\psi_0\rangle - i|\psi_1\rangle)/\sqrt{2}$ can be prepared from the ground state $|\psi_0\rangle$, by using the state $\sigma$ on the reservoir as a resource. 
The preparation is characterized by the channel $\Phi_{x,\sigma}(\rho) := \Tr_{E}(e^{-ixH_{\textrm{JC}}}\rho\otimes\sigma e^{ixH_{\textrm{JC}}})$. For the initial state $|\psi_0\rangle\langle\psi_0|$, a given $x$, and $\sigma$, our measure of the `quality' of the preparation is
\begin{equation}
\label{nsdfkjgbnksdf}
\begin{split}
 Q_{x}(\sigma)  := &  \langle \phi|\Phi_{x,\sigma}(|\psi_0\rangle\langle\psi_{0}|)|\phi\rangle\\
 = &  \frac{1}{2} + \sum_{l=0}^{\infty}\sin(x\sqrt{l+1})\cos(x\sqrt{l}) \Real\langle l|\sigma |l+1\rangle.
 \end{split}
\end{equation}

 \begin{figure}[h]
 \centering
 \includegraphics[width= 8cm]{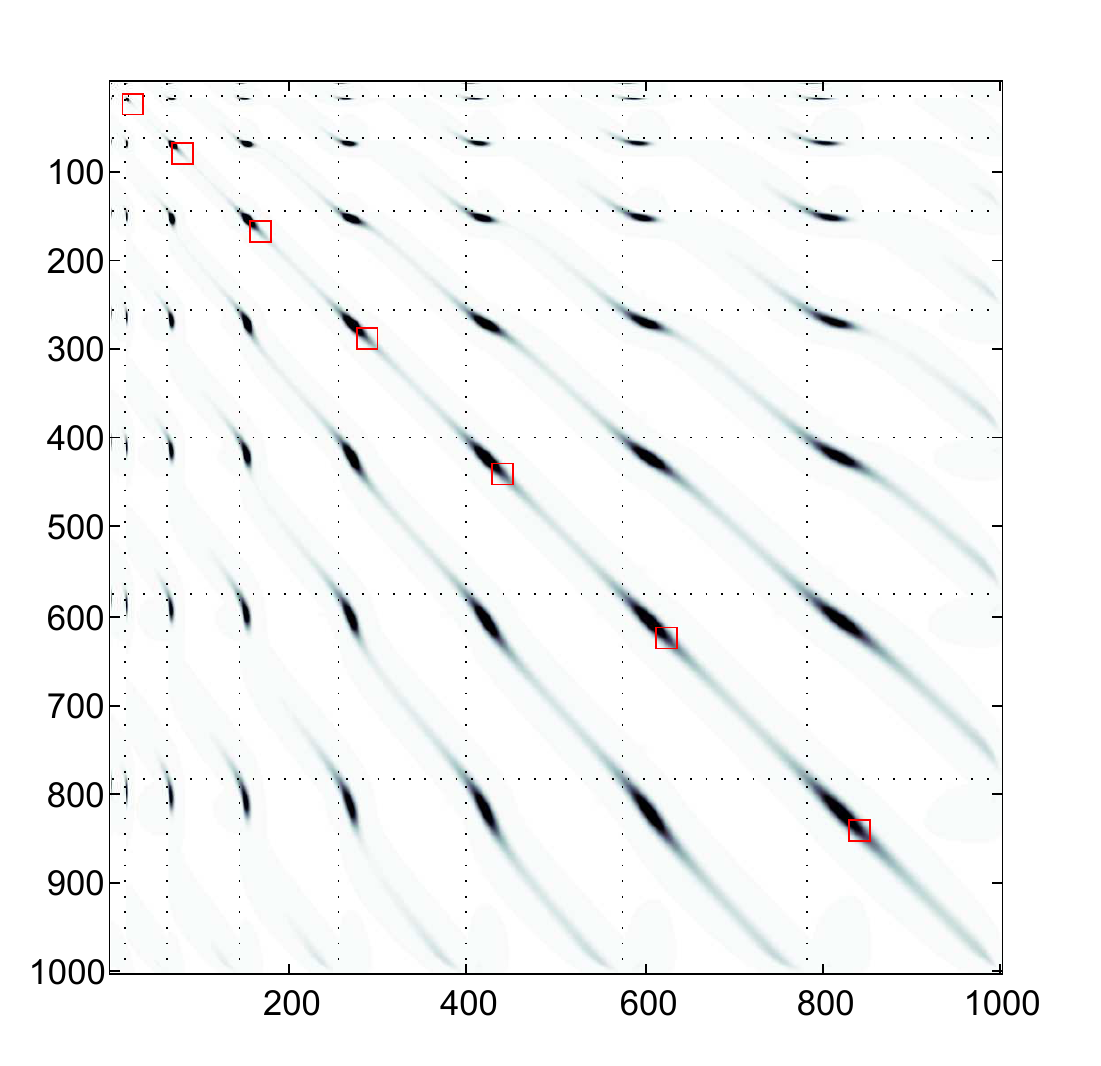} 
   \caption{\label{fig:dynamics} {\bf The state of the reservoir after iterated use.} 
    For an initial state $\sigma^{(0)}= |\eta\rangle\langle\eta|$, with a very wide superposition $|\eta\rangle = \sum_{l=0}^{999}|l\rangle/\sqrt{1000}$, the figure depicts $|\langle l|\sigma^{(100)}|l'\rangle|$ (black corresponds to high values, and white low) where $\sigma^{(100)}$ is the $100$th iteration of $\sigma^{(k+1)} = \Lambda_{\pi/4,|\psi_0\rangle\langle\psi_0|}(\sigma^{(k)})$, for $x = \pi/4$. The dotted lines corresponds to the values of $l$ (and $l'$) for which $B_{\pi/4}|l\rangle = 0$ ($B_{\pi/4}|l'\rangle =0$). Since $B_x$ induces loss of quanta from the reservoir, this suggests that weight will tend to accumulate at the crossings of the dotted lines. This figure shows an intermediate step in this process. The squares show the points $l_0^{\textrm{max}}(m)= (4m+1)^2$ around which $Q_{\pi/4}$ takes a high value (for relatively narrow distributions).
  On these points we center the family of initial states that lead to the curves in Fig.~3 in the main text. (As opposed to the initial state $\sigma^{(0)}$ that yields the present figure, those initial states only contains a relatively narrow distribution of number states.)
   }
\end{figure}

Analogously to what we did in Section \ref{AnExample} (where we optimized over all $U$) we should in principle  optimize $Q_{x}(\sigma)$ over all $x\geq 0$. However, the dynamics of the JC-model can unfortunately be rather irregular (see e.g. \cite{Eberly80,Shore93}), and due to this an optimization may not be feasible. Instead we here opt for a simpler (but possibly suboptimal) procedure, where we keep the value of $x$ fixed (corresponding to one fixed unitary). For the sake of simplicity we choose $x = \pi/4$. For an initial state $\sigma^{(0)}$ we construct the sequence $\sigma^{(k+1)} = \Lambda_{\pi/4,|\psi_0\rangle\langle\psi_0|}(\sigma^{(k)})$ and the corresponding fidelities $Q_{\pi/4}(\sigma^{(k)})$. For a suitably chosen $l_0$ the initial state is $\sigma^{(0)}:= |\eta_{L,l_0}\rangle\langle\eta_{L,l_0}|$. As we have seen, the evolution is asymptotically uniform. Hence, if $L$ is small compared to $l_0$, one can  make the approximation 
\begin{equation}
Q_{\pi/4}(\sigma^{(0)}) \approx \frac{1}{2} +  \frac{1}{2}(1-\frac{1}{L})\sin(\pi\sqrt{l_0}/2),
\end{equation}
 and thus obtain the (approximate) maxima $1-1/(2L)$ at  
\begin{equation}
l_0^{\textrm{max}}(m) := (4m + 1)^2,\quad m = 1,2,\ldots.
\end{equation}

\emph{Some further details concerning Fig.~3 in the main text.--}
Each graph in Fig.~3 in the main text is obtained from an initial state that is a uniform superposition over $50$ number states, centered around $l_0^{\textrm{max}}(m)$. Hence, $\sigma^{(0)}_{m} = |\eta_{L,l_0(m)}\rangle\langle\eta_{L,l_0(m)}|$, for $l_{0}(m) =  l_0^{\textrm{max}}(m) - 24 = (4m+1)^2-24$, i.e., $|\eta_{L,l_0(m)} \rangle = \sum_{j=-24}^{25}|l_0^{\textrm{max}}(m) + j\rangle$.
 Each graph corresponds to the value $F^{(m)}_k := Q_{\pi/4}(\sigma^{(k)}_{m})$ on the vertical axis versus the number of iterations $k$ on the horizontal axis. The different graphs correspond to the different values of $m$, where in the rightmost part of the figure, the curves correspond to $m = 5,7,9,11,14,18,24,31,40,52,67,87,113,147,191$ counted from below. 
These graphs suggest that the `lifetime' of the coherence can be made longer, merely by increasing the average energy of the state.

Since the graphs in Fig.~3 in the main text are the result of an approximate optimization, one should maybe not get too surprised by the fact that one see an initial increase of the fidelity. As seen, the largest increase happens for the initial states with a relatively low number of quanta, which is where we would expect the approximation to be the worst.

The dotted line in the figure corresponds to the fidelity we would have obtained for the doubly-infinite ladder-model, i.e., $1-1/(2L)$ (see Section \ref{AnExample}). Note that there is no particular reason to expect this to be a bound for what can be achieved in the JC model. However,  the gradual approach of the curves to this value suggests that the applied approximation is reasonable.

Concerning the numerical evaluation to obtain Fig.~3 in the main text, one can note that the quantity we wish to evaluate, equation (\ref{nsdfkjgbnksdf}), only depends on the first upper-diagonal $\langle l|\sigma|l+1\rangle$ of the density matrix. Furthermore, the evolution map $\Lambda_{\pi/4,|\psi_0\rangle\langle\psi_0|}$ transforms the elements $\langle l|\sigma|l+1\rangle$ among themselves, and thus  one never needs to store or calculate the entire density matrix.

\emph{Why not using coherent states?}
One may wonder why we here use $|\eta_{L,l_0}\rangle$ as initial states, rather than coherent states, $|\alpha\rangle := e^{-|\alpha|^2/2}\sum_{l=0}^{+\infty}\alpha^l |l\rangle/\sqrt{l!}$ \cite{Glauber63a,Glauber63b,MandelWolf}, which  would be the both convenient and conventional choice for this type of model. For the coherent states we could obtain a sequence of states with increasing average energy by increasing $|\alpha|$. However, as we increase $|\alpha|$ we also increase  the width of the superposition over the energy eigenstates. Hence, in some intuitive sense the degree of coherence increases as we increase $|\alpha|$, and for this reason it appears risky to use these states for the questions we investigate here. 
One should keep in mind though, that since we have not developed any operational measure of the degree of coherence (although Sec.~\ref{Sec:approxEnergymix} provides some steps in this direction), discussions of this kind are a bit shaky. (For the doubly-infinite and the half-infinite ladder models, we side-step the issue of measures of coherence by focusing on the set of induced channels $\mathcal{C}(\sigma)$.)

\subsection{\label{GJC}The half-infinite ladder as a generalized Jaynes-Cummings model}
Here we briefly return to the half-infinite ladder model in Section \ref{Harmonic}, to point out that in the special case of system $S$ being a qubit, one can obtain unitary operators of the $V_{+}(U)$-form via Hamiltonians in a generalized Jaynes-Cummings model.  

This generalization corresponds to letting the coupling parameter ($g$ in equation (\ref{TheJCHamiltonian})) be dependent on the field strength, leading to interaction Hamiltonians of the form $H = g\sigma_{+}\otimes af(a^{\dagger}a) + g\sigma_{-}\otimes f(a^{\dagger}a)a^{\dagger}$, where $f$ is a function.  Comparing with Section \ref{JCmodel} one can see that we obtain the uniform sequence of unitary matrices $V^{(1)} = V^{(2)} = \cdots$ in the case $f(l) = 1/\sqrt{l}$, i.e., for the Hamiltonian
\begin{equation}
H_{GJC}  = g\sigma_{+}\otimes a(a^{\dagger}a)^{-1/2} + g\sigma_{-}\otimes (a^{\dagger}a)^{-1/2}a^{\dagger}.
\end{equation}
(This should not be confused with the case $f(l) = \sqrt{l}$ that has been considered in the literature \cite{Buck81,Singh06,Alexio00, Buzek, Rybin99}.)
 It certainly is an interesting question to what extent such a generalized Jaynes-Cummings interaction can be obtained, or at least approximated, within realistic systems. However, this is left as an open question.

%%%%%%%%%%%%%%%%%%%%%%%%%%%%%%%%%%%%%%%%%%%%%%%%%%%%%
%			Expected	work extraction	 without energy reservoir			%
%%%%%%%%%%%%%%%%%%%%%%%%%%%%%%%%%%%%%%%%%%%%%%%%%%%%%

\section{\label{WithoutReservoir}Prelude: Expected work-extraction without an energy reservoir}

Here we turn to the question of work extraction. (For introductions to work extraction and the closely related Landauer's principle for information erasure, see e.g., \cite{Szilard,Landauer61,Procaccia,Benett03,LeffRexI,LeffRexII,Takara,Esposito06,Maroney09,Reeb13b}.) Before we begin with the actual purpose of analyzing the role of coherence in expected work extraction by the use of explicit energy reservoirs (which we do in Section \ref{ExplicitEnergyReservoirs}) we first consider some simpler types of models that do not contain an explicit energy reservoir. For these models one can show that the optimal expected work extraction (suitably defined) is given by the relative von Neumann entropy $kTD\boldsymbol{(}\rho\Vert G(H_S)\boldsymbol{)}$, where $D(\rho\Vert \sigma) = \Tr(\rho\ln\rho)-\Tr(\rho\ln\sigma)$. The purpose of this exposition is partially to highlight the difference between using, and not using, an explicit energy reservoir, but also to gain some understanding for the ideas behind the (somewhat technically messy) derivations in Section \ref{ExplicitEnergyReservoirs}. On a technical level, the only things that we need from this section are Lemma \ref{ztioutruor} and Corollary \ref{jgfkgfj}.

In Section  \ref{TimeDependent} we begin with a short `prelude to the prelude', where we briefly recall the common approach to implement expected work extraction using time-dependent Hamiltonians. More precisely, the system traverses through a path of Hamiltonians. These paths of Hamiltonians will implicitly play a role in the subsequent derivations, as we, in some sense, will simulate such paths.  In Section \ref{passivity} we recall the concept of passivity, and the related setting for work extraction. In Section \ref{SimulationPass} we implement one version of the above mentioned `simulation' (not to be confused with the simulation in Section \ref{TranslationPhaseRef}), and show how this leads to optimal expected work extraction. 

\subsection{\label{TimeDependent} Expected work extraction by varying the Hamiltonian}
A common approach (that can take many different guises on the level of assumptions and technical details) is to define expected work extraction in terms of changes of the Hamiltonian of the system. (For a handful of examples, see e.g. \cite{Alicki04,Esposito06,Henrich07,Maroney09,Sagawa09,Anders12,Gelbwaser13}, were \cite{Anders12} provides an introduction.) In this setting we are given a quantum state $\rho$ and a Hamiltonian $H$, and are allowed to perform changes of the Hamiltonian $H\mapsto H'$. The expected work gain of such a step is defined as
\begin{equation}
\label{PathGain}
W_{\textrm{yield}}(H_S\mapsto H'_S,\rho) := \Tr(H_S\rho)-\Tr(H'_S\rho).
\end{equation}
The contact with the heat bath is modeled as a replacement map $\rho\mapsto \Phi_{H}(\rho) := G(H)$. (As we will come back to below, there are variations on this, which can  implement `softer' types of thermalizations.) 
A process is defined as a sequence of  shifts of Hamiltonians, sandwiched by thermalizations, where the total work gain is defined as the sum of all the individual gains as in (\ref{PathGain}). This model is, for expected work extraction, a quantum generalization of the discrete classical model used in \cite{TrulyWorkLike}. (Unfortunately, this generalization does not make much sense as a generalization for $\epsilon$-deterministic work extraction.)  Very much analogous to \cite{TrulyWorkLike}, one can show that for a cyclic change of Hamiltonians, the optimal expected work extraction in this setting becomes $kTD\boldsymbol{(}\rho\Vert G(H_S)\boldsymbol{)}$. To see this, one can use the relation $H = F(H)-kT\ln G(H)$, where $F(H)= -kT\ln Z(H)$, and $Z(H) = \Tr e^{-\beta H}$, to show that the total work gain of an $L$-step  cyclic  process ($H^{(0)}= H^{(L)} = H_S$), with $\rho^{(0)}:=\rho$  and $\rho^{(l+1)}= \Phi_{H^{(l+1)}}(\rho^{(l)})$ becomes 
\begin{equation}
\label{PathTotalGain}
\begin{split}
W_{\textrm{yield}}  = & \sum_{l=0}^{L-1}[\Tr(H^{(l)}\rho_{l})-\Tr(H^{(l+1)}\rho_{l})]\\
 = & \frac{1}{\beta}D\boldsymbol{(}\rho\Vert G(H_S)\boldsymbol{)} -\frac{1}{\beta}D\boldsymbol{(}\rho_{L-1}\Vert G(H_S)\boldsymbol{)}\\
 & -\frac{1}{\beta}\sum_{l=1}^{L-1}\Big[ D\boldsymbol{(}\rho_{l-1}\Vert G(H^{(l)})\boldsymbol{)}-D\boldsymbol{(}\rho_l\Vert G(H^{(l)})\boldsymbol{)} \big].
\end{split}
\end{equation}
Using $\Phi_{H}\boldsymbol{(}G(H)\boldsymbol{)} = G(H)$, together with the fact that the von Neumann relative entropy is contractive under the action of channels, i.e., $D\boldsymbol{(}\Gamma(\rho)\Vert \Gamma(\sigma)\boldsymbol{)} \leq D(\rho\Vert\sigma)$, it follows that  $D\boldsymbol{(}\rho_l\Vert G(H^{(l)})\boldsymbol{)} = D\boldsymbol{(}\Phi_{H^{(l)}}(\rho_{l-1})\Vert \Phi_{H^{(l)}}(G(H^{(l)})\boldsymbol{)} \leq  D\boldsymbol{(}\rho_{l-1}\Vert G(H^{(l)})\boldsymbol{)}$, and thus 
$W_{\textrm{yield}} \leq  kTD\boldsymbol{(}\rho\Vert G(H_S)\boldsymbol{)} -kTD\boldsymbol{(}\rho_{L-1}\Vert G(H_S)\boldsymbol{)}$. Due to the fact that $D(\cdot\Vert\cdot)\geq 0$, it follows that within this model, for any cyclic process, we get the general bound  $W_{\textrm{yield}} \leq  kTD\boldsymbol{(}\rho\Vert G(H_S)\boldsymbol{)}$. 

As a side-remark one may note that, for the above derivation, we do not necessarily need to use the replacement map $\rho\mapsto G(H)$. As seen, we could, instead of $\Phi_{H}$, use \emph{any} channel that has $G(H)$ as a fix-point, i.e., $\Phi\boldsymbol{(}G(H)\boldsymbol{)} = G(H)$. (This implies a considerable freedom in implementing the softer types of thermalizations mentioned above. However, we shall not need this freedom in our derivations, but will stay put with the, rather brutal, complete thermalization implemented by the replacement map.)

The next question is if it is possible to saturate the bound $W_{\textrm{yield}} \leq  kTD\boldsymbol{(}\rho\Vert G(H_S)\boldsymbol{)}$ within this model. A (limit) process that does achieve the bound can be described as follows: Change the initial Hamiltonian $H_S$ into $H'$, where the latter is such that $\rho = G(H')$. (Strictly speaking, this only works if $\rho$ has full support. This issue can be handled via yet another limiting process, with a sequence of Gibbs states that converges to the state, akin to what was done in \cite{TrulyWorkLike}.) Next, we thermalize the system and perform an alternating sequence of  thermalizations and small changes of the Hamiltonians, keeping the eigenbasis fixed, until we reach a Hamiltonian $H''$ that is iso-spectral to $H_S$. Finally, we perform an alternating sequence of thermalizations and small unitary transformations, until $H''$ has been rotated back to $H_S$. The first of these steps has the expected work yield $W^{(1)}_{\textrm{yield}} = F(H_S)-F(H')+kTD\boldsymbol{(}\rho\Vert G(H_S)\boldsymbol{)}$. The second step does in the limit of an infinitely fine discretization  yield $W^{(2)}_{\textrm{yield}} = F(H')-F(H'')$. (Since we for this step only have states that are diagonal in a fixed energy eigenbasis, the results for the ITR in \cite{TrulyWorkLike} are applicable.)
For the final step of the protocol let $U$ be a unitary that transforms an eigenbasis of $H''$ into an eigenbasis of $H_S$. We diagonalize $U= \sum_{n}e^{-i\phi_n}|\xi_n\rangle\langle\xi_n|$, for $0\leq \phi_n <2\pi$, and define the Hermitian operator $A := \sum_{n}\phi_n|\xi_n\rangle\langle\xi_n|$. We can thus define a sequence of Hamiltonians $H^{(l)}:= e^{-iAl/L}H''e^{iAl/l}$, for $l=0,\ldots,L$, such that $H^{(0)}= H''$ and $H^{(L)} = H_S$. Consider a  process the sequentially  changes $H^{(l)}$ to $H^{(l+1)}$, sandwiched by thermalizations. The work yield of this process is.  
\begin{equation}
\begin{split}
W_{\textrm{yield}}^{(3)}  = & \sum_{l=0}^{L-1}\Tr[(H^{(l)}-H^{(l+1)})G(H^{(l)})] \\
 = & L\Tr[( H''-e^{-iA/L}H'' e^{iA/L})G(H'')].
 \end{split}
\end{equation}
 By inserting the expansion $e^{-iA/L}H'' e^{iA/L} = H'' -i[A,H'']/L  + O(1/L^2)$ into the above equation, one sees that $W_{\textrm{yield}}^{(3)} \rightarrow 0$ as $L\rightarrow \infty$. Hence, the total work yield $W_{\textrm{yield}} = W^{(1)}_{\textrm{yield}} + W^{(2)}_{\textrm{yield}} + W^{(3)}_{\textrm{yield}}$ can be made arbitrarily close to $kTD\boldsymbol{(}\rho\Vert G(H_S)\boldsymbol{)}$.

\subsection{\label{passivity}Work extraction and passive states}

The concept of passive states were introduced in the context of axiomatic characterizations of Gibbs states \cite{PuszWoronowicz,Lenard}. See also recent results in \cite{AlickiFannes12}  and \cite{Hovhannisyan13}, which study the role of entanglement in this type of setting.

We are given a quantum system $S$, with a Hamiltonian $H_S$, and a quantum state $\rho$. We are allowed to perform arbitrary unitary operations (no matter whether they mix energy levels or not). The work gain of a unitary operation $U$ is defined as 
\begin{equation}
\label{nsdjfkbnv}
W_{\mathrm{yield}}(U,\rho,H_S) := \Tr(H_S\rho)- \Tr(H_{S}U\rho U^{\dagger}). 
\end{equation}

In light of the results in the previous sections one may suspect that this construction implicitly assumes the  access to ideal coherence. To analyze this issue in a quantitative manner is precisely the purpose of Section \ref{ExplicitEnergyReservoirs}.

A state $\rho$ is \emph{passive} with respect to a given Hamiltonian $H_S$, if we cannot extract any energy from it in the sense of (\ref{nsdjfkbnv}), i.e.,  $W_{\mathrm{yield}}(U,\rho,H_S)\leq 0$ for all unitary $U$ \cite{PuszWoronowicz,Lenard}. It turns out \cite{PuszWoronowicz,Lenard}  that a state is passive if and only if the state $\rho$ and the Hamiltonian $H_S$ has a common eigenbasis, and that in this basis the eigenvalues of the density operator are ordered non-increasingly with increasing energy eigenvalues  (or equivalently if the eigenvalues of $\rho$ and $G(H_S)$ co-decrease). 

In this model it is straightforward to calculate the optimal extraction (e.g., by use of Lemma \ref{kvjdbkv}) 
\begin{equation}
\label{PassivityOptimum}
\begin{split}
W^{\mathrm{optimal}}_{\mathrm{yield}}  := & \sup_{U\in U(\mathcal{H}_S)}W_{\mathrm{yield}}(U)\\
 = & \Tr(H_S\rho_S )- \inf_{U\in U(\mathcal{H}_S)}\Tr(H_{S}U\rho_S U^{\dagger})\\
 = & \frac{1}{\beta}D\boldsymbol{(}\rho_S \Vert G(H_S)\boldsymbol{)} -\frac{1}{\beta}D\boldsymbol{(}\lambda^{\downarrow}(\rho_S) \Vert \lambda^{\downarrow}(G(H_S))\boldsymbol{)}.
 \end{split}
\end{equation}
Here $\lambda_{1}^{\downarrow}(\rho)\geq \cdots \geq \lambda_{N}^{\downarrow}(\rho)$ means the eigenvalues of the operator $\rho$ ordered non-increasingly. As seen from (\ref{PassivityOptimum}) we obtain the optimal work extraction precisely when we put the system in a passive state.
One might wonder from where the $\beta$ in (\ref{PassivityOptimum}) comes from, as we have not yet introduced any heat bath. At this stage it is merely a mathematical identity. [All the $\beta$ on the right hand side of  (\ref{PassivityOptimum}) cancel, for any $\beta>0$.]

Due to the fact that relative entropy is non-negative $D(\cdot\Vert \cdot)\geq 0$, the optimal value can never be larger than $kTD\boldsymbol{(}\rho\Vert G(H_S)\boldsymbol{)}$.  
If we somehow could `arrange' for the term $kTD\boldsymbol{(}\lambda^{\downarrow}(\rho) \Vert \lambda^{\downarrow}(G(H_S))\boldsymbol{)}$ to become zero, we would obtain $kTD\boldsymbol{(}\rho\Vert G(H_S)\boldsymbol{)}$ as the optimum. 
 
An important observation is that a combination of passive states is not necessarily passive \cite{PuszWoronowicz,Lenard}. As an example, for a passive state $\rho$, the tensor product $\rho\otimes \rho$ may not necessarily be passive. For this we assume that the two subsystems have identical and non-interaction Hamiltonians, resulting in the total Hamiltonian $H_{S_1}\otimes \hat{1}_{S_2} + \hat{1}_{S_1}\otimes H_{S_2}$. We furthermore allow arbitrary unitary operations on the \emph{joint} system $\mathcal{H}_{S_1}\otimes\mathcal{H}_{S_2}$. The Gibbs state $G(H)$ have the special property that it is completely passive  \cite{PuszWoronowicz,Lenard}, meaning that it remains passive for arbitrary tensor products $G(H)^{\otimes n}$. 

Combining copies of a system is not the only way in which we can change a passive state into a non-passive state. Another method is to add on an ancillary system $A$ in a Gibbs state $G(H_A)$ for some Hamiltonian $H_A$. We furthermore assume that the ancillary system is  non-interacting with $S$, such that the total Hamiltonian is $H_{SA} := H_S\otimes\hat{1}_A + \hat{1}_S\otimes H_A$. Substituting `$S$' in (\ref{PassivityOptimum}) with `$SA$' with the new Hamiltonian $H_{SA}$, and the new initial state $\rho\otimes G(H_A)$ [where we let the $\beta$ in (\ref{PassivityOptimum}) be the same as the $\beta$ in the Gibbs state $G(H_A)$ of the ancillary system] yields
 \begin{equation}
 \label{passoptimHeat}
 \begin{split}
W^{\mathrm{opt}}_{\mathrm{yield}} = &\Tr\big((H_S+ H_A)\rho\otimes G(H_A)\big) \\
 & -\inf_{U\in U(\mathcal{H}_S\otimes \mathcal{H}_A)} \Tr\big((H_S+H_A)U\rho\otimes G(H_A) U^{\dagger}\big)\\
 =  & \frac{1}{\beta}D\boldsymbol{(}\rho \Vert G(H_S)\boldsymbol{)} \\
  &-\frac{1}{\beta}D\Big(\lambda^{\downarrow}\big(\rho\otimes G(H_A)\big) \Vert \lambda^{\downarrow}\big(G(H_S)\otimes G(H_A)\big)\Big). 
\end{split}
\end{equation}
Comparing (\ref{passoptimHeat}) with (\ref{PassivityOptimum}) we see that the first term is unchanged [due to fact that $D\boldsymbol{(}\rho\otimes G(H_A)\Vert G(H_S)\otimes G(H_A)\boldsymbol{)} = D\boldsymbol{(}\rho\Vert G(H_S)\boldsymbol{)}$]  while  the second term is altered. In general $D\boldsymbol{(}\lambda^{\downarrow}(\rho\otimes G(H_A)) \Vert \lambda^{\downarrow}(G(H_S)\otimes G(H_A))\boldsymbol{)}$ can be smaller than $D\boldsymbol{(}\lambda^{\downarrow}(\rho) \Vert \lambda^{\downarrow}(G(H_S))\boldsymbol{)}$, reflecting the fact that we can make a passive state non-passive by adding an ancillary Gibbs state. 
The purpose of the following section is to show that the former term can be made arbitrarily small by choosing the ancillary Hamiltonian $H_A$ appropriately. 

The justification for the use of these ancillary Gibbs states is that we in principle always can obtain them  by equilibrating the desired ancillary system with respect to the given heat bath of temperature $T$, i.e., an ancillary Gibbs state is a `free resource' \cite{Janzing00,Brandao11,Horodecki11}. More precisely, one can imagine that the ancilla system is put in contact with the heat bath. Next, we let it continuously equilibrate with the heat bath while we slowly de-connect it. (Since we here deal with questions about ideal thermodynamics we do not worry about very large equilibration times.)

\subsection{\label{SimulationPass}Simulation of time-dependent Hamiltonians using a sequence of subsystems}

The purpose of this section is to show that the term $D\boldsymbol{(}\lambda^{\downarrow}(\rho\otimes G(H_A)) \Vert \lambda^{\downarrow}(G(H_S)\otimes G(H_A))\boldsymbol{)}$  in equation (\ref{passoptimHeat}) can be made arbitrarily small by choosing $H_A$ appropriately, which implies that we can achieve the expected work yield $kTD\boldsymbol{(}\rho\Vert G(H_S)\boldsymbol{)}$. In the construction of this Hamiltonian $H_A$, the aforementioned  `simulation' of time dependent Hamiltonians plays an important role (see Fig.~\ref{FigureSimulation}).

Let us reformulate this in a slightly more precise way. Define the probability distributions $q:= \lambda^{\downarrow}(\rho)$ and $p:= \lambda^{\downarrow}\boldsymbol{(}G(H_S)\boldsymbol{)}$. 
For $h= (h_1,\ldots, h_N)\in \mathbb{R}^{N}$, we furthermore define $G_{n}(h) := e^{-\beta h_n}/Z(H)$, $Z(h) := \sum_{n}e^{-\beta h_n}$, $F(h):= -\frac{1}{\beta}\ln Z(h)$ (much analogous to how we defined $G(H)$, $Z(H)$, and $F(H)$, for operators).

If $(h_1,\ldots,h_M) \in \mathbb{R}^{M}$ are the eigenvalues of $H_A$, it thus follows that we can write 
\begin{equation}
\begin{split}
& D\Big(\lambda^{\downarrow}\big(\rho\otimes G(H_A)\big) \Vert \lambda^{\downarrow}\big(G(H_S)\otimes G(H_A)\big)\Big) \\
& = D\Big( \big(qG(h)\big)^{\downarrow} \Vert \big(pG(h)\big)^{\downarrow}\Big).
\end{split}
\end{equation}
 One can thus realize that the problem of making this relative entropy small  is not particularly `quantum' to its nature, as it can be reformulated entirely as a question concerning probability distributions. (This is essentially due to the optimal unitary transformation $U$ in equation (\ref{passoptimHeat}), which makes the two states diagonal in the same eigenbasis.) 
In other words, the question whether we can find the appropriate Hamiltonians $H_A$ boils down to the question whether we can find $h$ for given $q$ and $p$ such that $D\big( (qG(h))^{\downarrow} \Vert (pG(h))^{\downarrow}\big)$ becomes small. It is thus enough to perform the `simulation'  entirely within a classically discrete setting (in the sense of \cite{TrulyWorkLike}), where we consider sequences of energy level configurations $h$ and probability distributions over these.

\begin{figure}[h!]
\center
 \includegraphics[width= 8cm]{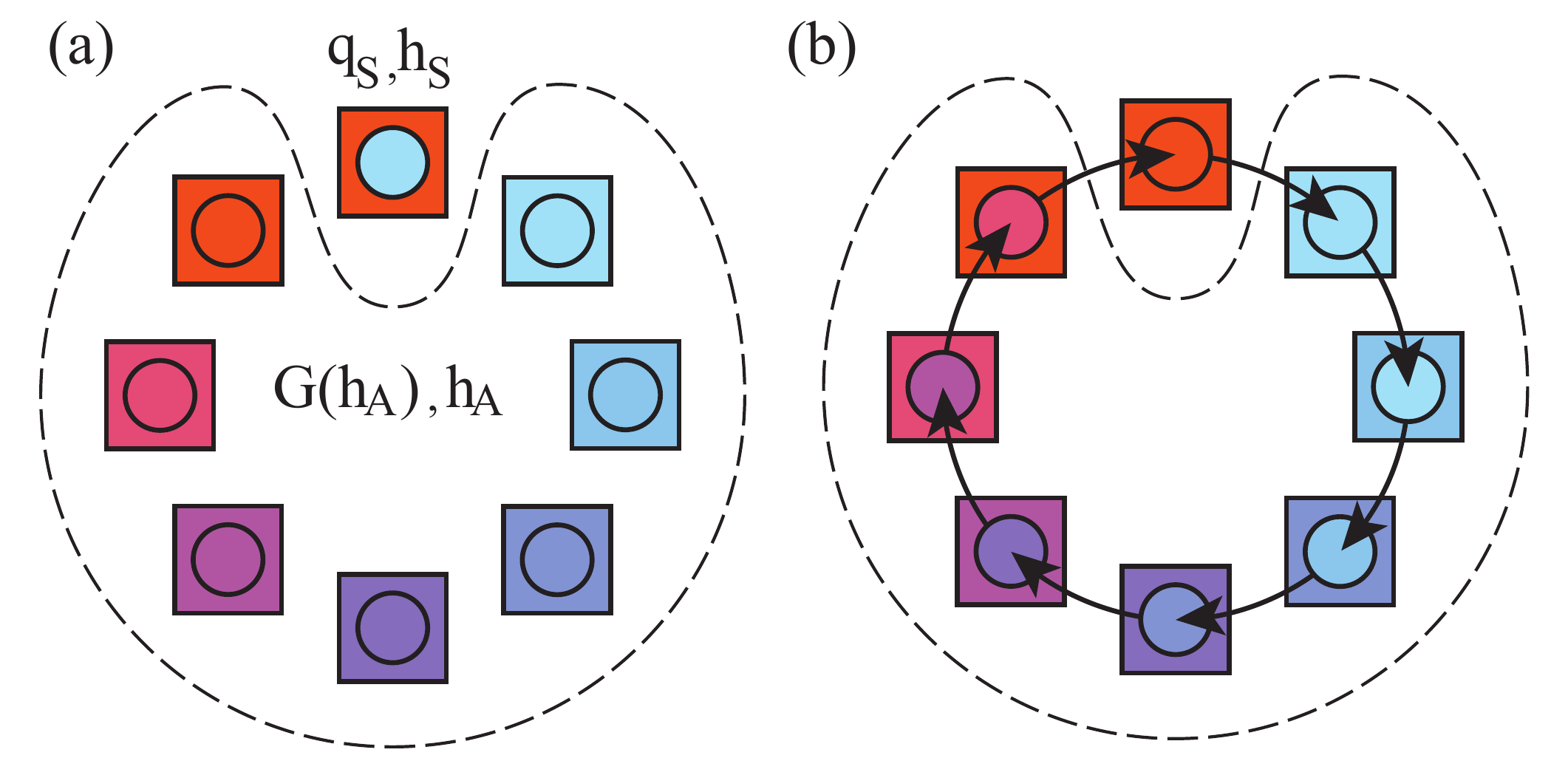} 
   \caption{\label{FigureSimulation} {\bf Simulation.} 
    Each circle represents the state of a subsystem, and the corresponding square its Hamiltonian. The circle and square have the same color when the system is Gibbs distributed with respect to its Hamiltonian.
 (a)  We are given a probability distribution $q$ over a set of states with energy levels $h_S$. To this we add a collection of ancillary systems, each with Hamiltonian $h^k$ and state distribution $G(h^k)$, forming a joint ancillary system $A$ which is Gibbs distributed $G(h_A) = G(h^0)\cdots G(h^K)$ with respect to the Hamiltonian $h^A_{n_0,\ldots,n_K} = h^0_{n_0} + \cdots + h^K_{n_K}$. Arranged in a circle, these systems start with $h^{0}$ such that $G(h^{0}) = q$, and ends with $h^K := h_S$. We can view  $h^0,\ldots,h^K$ as a discretization of a smooth path in the space of energy level configurations.
(b) By a cyclic permutation $\Pi$ of the states along this cycle, $S$ is put in equilibrium, while most of the ancillary systems are shifted slightly out of equilibrium. 
 The relative entropy $D\boldsymbol{(}\Pi(qG(h_A))\Vert G(h_S)G(h_A) \boldsymbol{)}$ goes to zero in the limit of an infinitely fine discretization along a smooth path of energy level configurations.
 }
\end{figure}

The relative entropy does not satisfy the triangle inequality. 
For the construction of our proofs, we make use of this `failure' of the relative entropy. Lemma \ref{ztioutruor}  shows that if we would regard the relative entropy as a (strange) distance measure, then the total sum of the `lengths' of the pieces of a path can be made arbitrarily small merely by making the division finer. (This is closely related to the `isothermal reversible paths' as described in \cite{TrulyWorkLike}.)

\begin{Lemma}
\label{erererre}
Let $h:[0,1]\rightarrow \mathbb{R}^{N}$ have component-wise continuous second derivatives, then
\begin{equation}
\bigg|\frac{d^2}{dx^2}F(h(x))\bigg| \leq  \max_{n}|h''_{n}(x)|  +2\beta \max_{n}|h'_{n}(x)|^2. 
\end{equation}
\end{Lemma}
By combining the above lemma with the Taylor expansions, $f(x) =  f(x_0) + (x-x_0)f'(x_{0})+ r(x)$, with the error term in the Lagrange form $r(x) := \frac{(x-x_0)^2}{2}f''(\xi)$ for $\xi\in[\min(x,x_0),\max(x,x_0)]$ one can obtain the following lemma. 
\begin{Lemma}
\label{ztioutruor}
Let $N\in\mathbb{N}$ and $\beta>0$ be given, and let $h:[0,1]\rightarrow \mathbb{R}^{N}$ have component-wise continuous second derivatives. For each integer $K\geq 2$
\begin{equation}
\begin{split}
 0  \leq & \sum_{k=0}^{K-1}D\Big( G\big(h(\frac{k}{K})\big)\Big\Vert G\big(h(\frac{k+1}{K})\big)\Big) \\
 \leq & \frac{1}{K} \Big[\beta\max_{x\in[0,1]} \max_{n}|h''_{n}(x)| +\beta^2\max_{x\in[0,1]}\max_{n}|h'_{n}(x)|^2\Big].
 \end{split}
\end{equation}
\end{Lemma}
In the following we let $\mathbb{P}(N)$ denote the set of probability distributions over $N$ symbols, and $\mathbb{P}_{+}(N)$ the subset of probability distributions with full support.
\begin{Lemma}[Corollary II.4.4 on p.~49 in \cite{Bhatia}]
\label{kvjdbkv}
Let $a,b\in\mathbb{R}^{M}$, then 
\begin{equation}
\sum_{m=1}^{M}a^{\downarrow}_mb^{\uparrow}_{m} \leq \sum_{m=1}^{M}a_{m}b_{m}\leq \sum_{m=1}^{M}a^{\uparrow}_{m}b^{\uparrow}_{m}.
\end{equation}
\end{Lemma}

Every permutation $\pi$ on the set $\{1,\ldots,M\}$ induces an operation $\Pi$ on $\mathbb{P}(M)$ by $\Pi_{\pi}(p)_m:= p_{\pi(m)}$, i.e., we permute the elements in the index of $p$.
A direct consequence of Lemma \ref{kvjdbkv} is the following corollary.
\begin{Corollary}
\label{jgfkgfj}
Let $q\in \mathbb{P}(M)$ and $p\in\mathbb{P}_{+}(M)$, and $\pi$ be any permutation on $\{1,\ldots,M\}$, then
\begin{equation}
D(q^{\downarrow}\Vert p^{\downarrow}) \leq D\boldsymbol{(}\Pi_{\pi}(q)\Vert p\boldsymbol{)}.
\end{equation}
\end{Corollary}
Suppose that $M := N^{L}$, and suppose moreover that we choose to represent the elements of $\mathbb{P}(N^{L})$ as $L$-dimensional tensors, i.e., as $(p_{n_1,\ldots,n_{L}})_{(n_1,\ldots,n_{L})}$, where $(n_1,\ldots,n_{L})\in \{1,\ldots,N\}^{\times L}$. We can obtain a special permutation on the index set $\{1,\ldots,N\}^{\times L}$ by permuting the `index of the indices', i.e., $(n_1,\ldots,n_{L})$ is mapped to $(n_{\pi(1)},\ldots,n_{\pi(L)})$, where $\pi$ is a permutation of $\{1,\ldots,L\}$. (Another way to put it is to say that we permute $L$ subsystems.) This is merely a special case of a general permutation of $\{1,\ldots,N \}^{\times L}$. The permutation of the index of indices induces a corresponding operation $\Pi$ on $\mathbb{P}(N^{L})$, by $\Pi\big((p_{n_1,\ldots,n_{L}})_{n_1,\ldots,n_{L}}\big) = (p_{n_{\pi(1)},\ldots,n_{\pi(L)}})_{n_1,\ldots,n_{L}}$. As this is a special case of the more general permutation, it follows that Corollary \ref{jgfkgfj} is applicable, which is important for the proof of the following proposition.

\begin{Lemma}
\label{ptozuirpuztoi}
Let $N$ and $q,p\in \mathbb{P}_{+}(N)$ be given. For each integer $K\geq 2$  there exists a $h^{K}\in \mathbb{R}^{N^{K+1}}$ such that 
\begin{equation}
\label{rtuzrtu}
D\Big( \big(q G(h^K)\big)^{\downarrow} \Vert \big(p G(h^K)\big)^{\downarrow}\Big) \leq \frac{1}{K}\max_{n}|\ln\frac{q_n}{p_n}|^2. 
\end{equation}
\end{Lemma}
\begin{proof}
Since, $q,p \in \mathbb{P}_{+}(N)$, it follows that $\ln q_n$ and $\ln p_n$ are well defined and finite. Thus we can define $h^{i}_n:= -\frac{1}{\beta}\ln q_n$, $h^{f}_n: = -\frac{1}{\beta}\ln p_n$. We connect these two points with the path $h:[0,1]\rightarrow \mathbb{R}^N$ by $h(x):= (1-x)h^{i}+ xh^{f}$.  Lemma \ref{ztioutruor} is applicable, and thus delivers $h^{k:K} := h(\frac{k}{K})$ with $k=0,1,\ldots, K$ such that 
\begin{equation}
\label{dksvjbds}
  \sum_{k=0}^{K-1}D\big( G(h^{k:K})\big\Vert G(h^{k:K})\big) 
 \leq  \frac{1}{K}\max_{n}|\ln \frac{q_n}{p_n}|^2.
\end{equation}
Define 
\begin{equation}
\begin{split}
h^{K}_{n_0,\ldots,n_K}  := & h^{0:K}_{n_0} + \cdots + h^{K:K}_{n_K},\\
&  (n_0,\ldots,n_K)\in \{1,\ldots,N\}^{K+1}.
\end{split}
\end{equation}
Hence $h^K\in \mathbb{R}^{N^{K+1}}$, and
\begin{equation}
G_{n_0,n_1,\ldots,n_K}(h^K) = G_{n_0}(h^{0:K})\cdots G_{n_K}(h^{K:K}).
\end{equation}
On $\mathbb{R}^{N^{K+2}}$ we define the operation $\Pi$ by
\begin{equation}
(\Pi a)_{n_s, n_0,n_1,\ldots,n_{K-1},n_K} := a_{n_0,n_1,\ldots,n_{K-1},n_K,n_s}.
\end{equation}
In other words, $\Pi$ is defined via a cyclic permutation of the subsystems, i.e., via a permutation of the index of the indices of $a$. 
By Corollary \ref{jgfkgfj} it follows that 
\begin{equation}
\begin{split}
 & D\Big(\big(qG(h^{K})\big)^{\downarrow}\Big\Vert \big(pG(h^{K})\big)^{\downarrow} \Big) \\
  & \leq  D\Big(\Pi\big(qG(h^{K})\big)\Big\Vert pG(h^{K})\Big) \\
 & =  D\Big(G(h^{K:K})qG(h^{0:K})\cdots G(h^{K-1:K})\Big\Vert \\
 & \quad \quad \quad\quad pG(h^{0:K})\cdots G(h^{K-1:K})G(h^{K:K})\Big) \\
 & =  D\boldsymbol{(}G(h^{K:K})\Vert p \boldsymbol{)} + D\boldsymbol{(}q\Vert G(h^{0:K})\boldsymbol{)} \\
 & +\sum_{k=0}^{K-1}D\boldsymbol{(}G(h^{k:K})\Vert G(h^{k+1:K})\boldsymbol{)}.
 \end{split}
\end{equation}
By construction $q = G(h^{0:K})$ and $p = G(h^{K:K})$, and by (\ref{dksvjbds}) the statement of the lemma follows.
\end{proof}

\begin{Proposition}
Let $\rho\in\mathcal{S}_{+}(\mathcal{H}_S)$ and let $H_S$ be a Hermitian operator on $\mathcal{H}_S$. Then there exists a sequence of ancillary Hilbert spaces $\mathcal{H}^K_A$ and Hermitian operators $H_A^K$ on $\mathcal{H}^K_A$ such that 
\begin{equation}
\lim_{K\rightarrow +\infty}W^{\mathrm{opt, K}}_{\mathrm{yield}} = \frac{1}{\beta}D\boldsymbol{(}\rho \Vert G(H_S)\boldsymbol{)},
\end{equation}
where
\begin{equation*}
\begin{split}
 W^{\mathrm{opt, K}}_{\mathrm{yield}} := & \Tr\big((H_S+ H_A^K)\rho\otimes G(H_A^K)\big)\\
 & -\inf_{U\in U(\mathcal{H}_S\otimes \mathcal{H}_A^K)} \Tr\big((H_S+H_A^K)U\rho\otimes G(H_A^K) U^{\dagger}\big).
\end{split}
\end{equation*}
\end{Proposition}
\begin{proof}
Let $q$ be the eigenvalues of $\rho$ and $p$ the eigenvalues of $G(H_S)$. Since $\rho\in\mathcal{S}_{+}(\mathcal{H}_S)$ it follows that $q\in \mathbb{P}_+(N)$. As $p$ are the eigenvalues of a Gibbs state, it follows directly that $p\in \mathbb{P}_{+}(N)$. By use of Lemma \ref{ptozuirpuztoi} we can construct a sequence $(h^{K})_{K\in\mathbb{N}}$, $h^{K}\in\mathbb{R}^{N^{K+1}}$, where each $h^K$ satisfies equation (\ref{rtuzrtu}). Let $\{|a^{K}_k\rangle\}_{k=1}^{N^{K+1}}$ be an orthonormal basis in an $N^{K+1}$-dimensional Hilbert space $\mathcal{H}_A^{K}$ and define the Hermitian operator $H_A^{K}:= \sum_{k=1}^{N^{K+1}}h^{K}_{k}|a^{K}_k\rangle\langle a^{K}_k|$. One can see that $D\boldsymbol{(}\lambda^{\downarrow}(\rho\otimes G(H_A^K))\Vert \lambda^{\downarrow}(G(H_S)\otimes G(H_A^K))\boldsymbol{)} = D\boldsymbol{(}(q G(h^K))^{\downarrow} \Vert (p G(h^K))^{\downarrow}\boldsymbol{)}$. Combining this observation with equations (\ref{passoptimHeat}) and  (\ref{rtuzrtu}), we obtain $0 \leq kTD\boldsymbol{(}\rho \Vert G(H_S)\boldsymbol{)} - W^{\mathrm{opt},K}_{\mathrm{yield}} \leq  \frac{1}{K}\max_{n}|\ln\frac{q_n}{p_n}|^2$, which proves the proposition.
\end{proof}

%%%%%%%%%%%%%%%%%%%%%%%%%%%%%%%%%%%%%%%%%%%%%%%%%%%%%
%			Expected	work extraction	 with an energy reservoir			%
%%%%%%%%%%%%%%%%%%%%%%%%%%%%%%%%%%%%%%%%%%%%%%%%%%%%%

\section{\label{ExplicitEnergyReservoirs} Expected work extraction with an explicit energy reservoir}
In the context of passivity one defines (as we saw in Section \ref{passivity}) the work gain as the decrease of internal energy of the system of interest under unitary operations. Here we introduce an explicit energy reservoir, and thus model the system into which the extracted energy is to be put. Consequently we define the expected work gain as the increase of expected energy in the reservoir. Throughout this section we will exclusively make use of the model introduced in Section \ref{DoublyInfinite}.

\subsection{Expected work extraction in fixed systems}
\subsubsection{Operations on the system and energy reservoir}
\begin{Definition}
\label{Rdefinition}
Let $s>0$, let $\mathcal{H}_S$ be a finite-dimensional Hilbert space, with $H_S\in H_{s}(\mathcal{H}_S)$, and let $\sigma\in\mathcal{S}(\mathcal{H}_E)$.
Let $s z_k$ and $|\psi_k\rangle$ be eigenvalues and corresponding orthonormal eigenvectors of $H_S$. Define
\begin{equation}
\label{Rdef}
\begin{split}
R^{H_S}_{\sigma}(Q)  :=   & \sum_{n,n'=1}^{N}|\psi_{n}\rangle\langle\psi_{n}|Q|\psi_{n'}\rangle\langle\psi_{n'}|  \Tr(\Delta^{z_n-z_{n'}}\sigma),\\
&  \forall Q\in L(\mathcal{H}_S).
\end{split}
\end{equation}
\end{Definition}
The channel $R^{H_S}_{\sigma}$ should not be confused with $\Phi^{S}_{\sigma,U}$ introduced in Section \ref{Sec:inducedch}.

\begin{Lemma}
\label{Runital}
For $s>0$, $H_{S}\in H_{s}(\mathcal{H}_S)$,and $\sigma\in\mathcal{S}(\mathcal{H}_E)$
\begin{equation}
\label{avkdjsbasdk}
R_{\sigma}^{H_S}(\hat{1}) = \hat{1}.
\end{equation}
In other words, the channel is unital and can thus not decrease the von Neumann entropy.
\end{Lemma}

So far in this investigations we have only dealt with expectation values of bounded operators, and have thus not needed to worry about the existence of these expectation values. However, in this section we need to consider expressions like $\Tr(H^{(s)}_E\sigma)$. Since $H^{(s)}_E$ is an unbounded operator, we cannot guarantee that the expectation value is well defined for all elements $\sigma\in\mathcal{S}(\mathcal{H}_E)$. One way of dealing with this would be to restrict $\mathcal{S}(\mathcal{H}_E)$ to elements for which the product $H^{(s}_E\sigma$ is a trace class operator ($Q$ is trace class if $\Tr\sqrt{Q^{\dagger}Q}< +\infty$). In essence, we need a restriction such that $\langle n|\sigma|n'\rangle$ converges sufficiently fast to zero as $n,n'\rightarrow \pm \infty$. 
However, for the purposes of this investigation there appears to be no strong reason to consider this technical issue in detail, and for this reason we do in the following merely let $\mathcal{S}^{*}(\mathcal{H}_E)$ denote some restriction of $\mathcal{S}(\mathcal{H}_E)$, where we (in one way or another) have made sure that all the relevant expectation values make sense. Note that all the operations we perform on the energy reservoir merely involve the interactions with a finite-dimensional space, and hence there is no reason to expect that these operations brings the state out of (a reasonably defined) $\mathcal{S}^{*}(\mathcal{H}_E)$.

\begin{Lemma}
\label{hjbjfghbfbghj}
Let $s>0$. Let $\mathcal{H}_S$ be a finite-dimensional Hilbert space, $H_S\in H_{s}(\mathcal{H}_S)$, $\rho\in\mathcal{S}(\mathcal{H}_S)$, and $\sigma\in\mathcal{S}^{*}(\mathcal{H}_E)$. Then
\begin{equation}
\begin{split}
& \Tr\big([\hat{1}_S\otimes H_{E}^{(s)}]V(U)\rho \otimes \sigma V(U)^{\dagger}\big) \\
 & =   \Tr(H_{E}^{(s)}\sigma ) + \Tr(H_S\rho) \\
  & \quad\quad- \Tr\big( U^{\dagger}H_SU R_{\sigma}^{H_S}(\rho) \big),
\end{split}
\end{equation}
for all $U\in \mathbb{U}(\mathcal{H}_S)$ and all $\rho\in\mathcal{S}(\mathcal{H}_S)$, where $V$ is defined as in Lemma \ref{nvmxcnv}, and $R_{\sigma}^{H_S}$ is as in Definition \ref{Rdefinition}.
\end{Lemma}
\begin{proof}
The first part of the proof is to rewrite $\Tr([\hat{1}_S\otimes H_{E}^{(s)}]V(U)\rho \otimes \sigma V(U)^{\dagger}) 
 =  \Tr([H_S\otimes \hat{1}_E + \hat{1}_S\otimes H_{E}^{(s)}]\rho \otimes \sigma ) 
 - \Tr([H_S\otimes \hat{1}_E]V(U)\rho \otimes \sigma V(U)^{\dagger})$, where we have used the fact that $V(U)$ is energy conserving. By using equation (\ref{bijection}) in Lemma \ref{nvmxcnv}, and use the fact that $H_S$ is diagonal in the energy eigenbasis $\{|\psi_n\rangle\}_{n=1}^{N}$, one can see that 
 $\Tr([H_S\otimes \hat{1}_E]V(U)\rho \otimes \sigma V(U)^{\dagger}) = \sum_{k,l,l'}s z_{l'}\langle\psi_{l}|U^{\dagger}|\psi_{l'}\rangle \langle\psi_{l'}|U|\psi_{k}\rangle\langle\psi_{k}|\rho |\psi_{l}\rangle\Tr(\Delta^{z_{k}-z_{l}}\sigma)$. The latter can be rewritten as $\Tr( U^{\dagger}H_SU R_{\sigma}^{H_S}(\rho))$.
\end{proof}

Given an Hermitian operator $Q$, we let $\lambda^{\downarrow}(Q)$ denote the collection of eigenvalues of $Q$ (counted with multiplicities) ordered non-increasingly, while $\lambda^{\uparrow}(Q)$ denotes the non-decreasing ordering. The following lemma is a direct consequence of Theorem 4.3.53 on p.~255 in \cite{MatrixAnalysisHJ}.
\begin{Lemma}
\label{optimalUnitary}
Let $A,B$ be Hermitian operators on some finite-dimensional Hilbert space $\mathcal{H}$. Then
\begin{equation}
\inf_{U\in U(\mathcal{H})}\Tr(U^{\dagger}AUB) = \sum_{k}\lambda^{\downarrow}_{k}(A)\lambda^{\uparrow}_{k}(B).
\end{equation}
\end{Lemma}

\begin{Lemma}
\label{ndjkfbnv}
Let $s,\beta>0$ and $M\in\mathbb{N}$ be given. 
Let $\mathcal{H}_S$ be a finite-dimensional Hilbert space, and $H_S\in H_{s}(\mathcal{H}_{S})$. Let $\rho\in\mathcal{S}(\mathcal{H}_S)$ and $\sigma\in\mathcal{S}^{*}(\mathcal{H}_E)$. Then
\begin{equation}
\label{vdfjkbnvkadjf}
\begin{split}
& \sup_{U\in U(H_S)}\Tr([\hat{1}_S\otimes H_E^{(s)}]V(U)\rho\otimes \sigma V(U)^{\dagger})   \\
& \quad =  \Tr(H_E^{(s)}\sigma) \\
 &\quad\quad  + \frac{1}{\beta}D\boldsymbol{(}\rho\Vert G(H_S)\boldsymbol{)} \\
 &\quad\quad  - \frac{1}{\beta}\Big[S\boldsymbol{(}R^{H_S}_{\sigma}(\rho)\boldsymbol{)} -S(\rho)\Big]\\
 &\quad\quad -\frac{1}{\beta}D\Big(\lambda^{\downarrow}\big(R^{H_S}_{\sigma}(\rho)\big) \Big\Vert  \lambda^{\downarrow}\big(G(H_S)\big) \Big).
 \end{split}
\end{equation}
Furthermore, there exists an element in ${U}(\mathcal{H}_S)$ that achieves the above maximum, i.e., the `$\sup$' in the above equation can be replaced by  `$\max$'.
\end{Lemma}
Compared to (\ref{PassivityOptimum}) we recognize two new features in equation (\ref{vdfjkbnvkadjf}). Firstly, the term $S\boldsymbol{(}R^{H_S}_{\sigma}(\rho)\boldsymbol{)} -S(\rho)$. Secondly, that we have $D\boldsymbol{(}\lambda^{\downarrow}(R^{H_S}_{\sigma}(\rho)) \Vert  \lambda^{\downarrow}(G(H_S)) \boldsymbol{)}$, rather than $D\boldsymbol{(}\lambda^{\downarrow}(\rho) \Vert  G(H_S) \boldsymbol{)}$. Both of these changes reflect the fact that the present model explicitly takes into account coherence. 
Due to Eq.~(\ref{avkdjsbasdk}) in Lemma \ref{Runital} it follows that $S\boldsymbol{(}R^{H_S}_{\sigma}(\rho)\boldsymbol{)} \geq S(\rho)$. Furthermore by the properties of relative entropy we have $D(\cdot\Vert \cdot)\geq 0$. Hence, the work yield is bounded from above by $kTD\boldsymbol{(}\rho\Vert G(H_S)\boldsymbol{)}$. As we shall see in Section  \ref{Sec:cohinworkextr}, the term $S\boldsymbol{(}R^{H_S}_{\sigma}(\rho)\boldsymbol{)} -S(\rho)$ is determined by the relation between the degree of coherence in the energy reservoir, and the degree to which the state $\rho$ contains superposition between   energy eigenspaces. Analogously as to what we did in Section \ref{SimulationPass}, the term $D\boldsymbol{(}\lambda^{\downarrow}(R^{H_S}_{\sigma}(\rho)) \Vert  \lambda^{\downarrow}(G(H_S)) \boldsymbol{)}$ can be made small by introducing a suitable ancillary Gibbs state. 

\begin{proof}[Proof of Lemma \ref{ndjkfbnv}.]
By combining Lemmas \ref{hjbjfghbfbghj} and \ref{optimalUnitary} we obtain
$\sup_{U}\Tr([\hat{1}_S\otimes H_E^{(s)}]V(U)\rho\otimes \sigma V(U)^{\dagger})  = \Tr(H_E^{(s)}\sigma) + \Tr(H_S\rho) -\sum_{k}\lambda^{\uparrow}_{k}(H_S)\lambda_{k}^{\downarrow} \boldsymbol{(}R^{H_S}_{\sigma}(\rho)\boldsymbol{)}$. Next we use $H_S = F(H_S)-kT\ln G(H_S)$, and $\lambda^{\uparrow}_{k}(H_S) = F(H_S)-kT\ln \lambda_k^{\uparrow}\boldsymbol{(}G(H_S)\boldsymbol{)}$. The rest of the proof is merely a rearrangement of the terms.

Finally we should show that the `$\sup$' can be replaced by `$\max$'.
Let $G_n^{\downarrow}(H_S)$ be the eigenvalues of $G(H_S)$ ordered non-increasingly, and let $|\psi^{\downarrow}_n\rangle$ be corresponding eigenvectors. Let $|\phi^{\downarrow}_n\rangle$ be eigenvectors of $R^{H_S}_{\sigma}(\rho)$ corresponding to the eigenvalues $\lambda^{\downarrow}_n\boldsymbol{(}R^{H_S}_{\sigma}(\rho)\boldsymbol{)}$. The unitary operator $\tilde{U}:=\sum_{n=1}^{N}|\psi^{\downarrow}_n\rangle\langle\phi^{\downarrow}_n|$ achieves the right hand side of equation (\ref{vdfjkbnvkadjf}).
\end{proof}

\subsubsection{Adding a heat bath}
Analogously to what we did in Section \ref{passivity}, we here model the heat bath by appending an ancillary system in a  Gibbs state. As in Section \ref{passivity} we assume that the system $S$ and the ancilla $A$ are non-interacting, i.e., the total Hamiltonian is $H_{SA} = H_{S}\otimes\hat{1}_A + \hat{1}_{S}\otimes H_{A}$. We are allowed to freely choose the Hamiltonian $H_A$ of the ancillary system, up to the condition that the eigenvalues have to be multiples of the energy-level spacing $s$ in the energy reservoir, i.e., $H_A\in H_s(\mathcal{H}_A)$.
We furthermore assume that the ancillary system always starts in the Gibbs state $G(H_A)$, and that  the total initial state of the joint system is $\rho\otimes G(H_A)$. 

\begin{Lemma}
\label{RonrhoG}
\item Let $s>0$ and let $\mathcal{H}_S$ and $\mathcal{H}_A$ be finite-dimensional Hilbert spaces and $H_S\in H_s(\mathcal{H}_S)$,  $H_A\in H_{s}(\mathcal{H}_{A})$. Let $\sigma\in\mathcal{S}^{*}(\mathcal{H}_{E})$,  then
\begin{equation}
\label{vabdfkjbvk}
R_{\sigma}^{H_S + H_A}\boldsymbol{(}\rho\otimes G(H_A)\boldsymbol{)} = R_{\sigma}^{H_S}(\rho)\otimes G(H_A),
\end{equation}
for every $\rho\in\mathcal{S}(\mathcal{H}_S)$.
\end{Lemma}
\begin{proof}
We let $h^S_n = sz^S_n$ and $|\psi^S_n\rangle$ be the eigenvalues and eigenvectors of $H_S$, and similarly let $h^A_m = sz^A_m$ and $|\psi^A_m\rangle$ be eigenvalues and eigenvectors of $H_A$. Hence, $h_n^S + h_m^A$ and $|\psi^S_n,\psi^A_m\rangle := |\psi^S_n\rangle|\psi^A_m\rangle$ are the eigenvalues and eigenvectors of $H_S + H_A$. The channel $R^{H_S+H_A}_{\sigma}$ can be written as
\begin{equation}
\begin{split}
 & R^{H_S+H_A}_{\sigma}(Q)   \\
 & = \sum_{n,n'=1}^{N}\sum_{m,m'=1}^{N}|\psi^S_{n},\psi^A_m\rangle\langle\psi^S_{n},\psi^A_m|Q|\psi^S_{n'},\psi^A_{m'}\rangle\langle\psi^S_{n'},\psi^A_{m'}| \\
 &   \quad\quad \times\Tr(\Delta^{z^S_{n}-z^S_{n'}+z^A_{m}-z^A_{m'}}\sigma).
\end{split}
\end{equation}
for all $Q\in L(\mathcal{H}_S\otimes\mathcal{H}_A)$. For $Q:=\rho\otimes G(H_A)$ we obtain the right hand side of (\ref{vabdfkjbvk}) by using the fact that $G(H_A)$ is diagonal in $\{|\psi^A_m\rangle\}_m$.
\end{proof}

By a direct combination of Lemma \ref{ndjkfbnv} and \ref{RonrhoG} we obtain the following.
\begin{Proposition}
\label{Optimum}
Let $s>0$, $\mathcal{H}_S$ and $\mathcal{H}_{A}$ be finite-dimensional, and $H_S\in H_{s}(\mathcal{H}_{S})$,  $H_{A}\in H_{s}(\mathcal{H}_{A})$. Let $\rho\in\mathcal{S}(\mathcal{H}_S)$ and $\sigma\in\mathcal{S}^{*}(\mathcal{H}_E)$. Then
\begin{equation}
\begin{split}
  &\sup_{U}\Tr\big([\hat{1}_S\otimes\hat{1}_A\otimes H_E^{(s)}]V(U)\rho\otimes G(H_A)\otimes \sigma V(U)^{\dagger}\big)  \\
 &\quad  =   \Tr(H_E^{(s)}\sigma) \\
 & \quad \quad+ \frac{1}{\beta}D\boldsymbol{(}\rho\Vert G(H_S)\boldsymbol{)} \\
 & \quad\quad - \frac{1}{\beta}\Big[S\boldsymbol{(}R^{H_S}_{\sigma}(\rho)\boldsymbol{)}-S(\rho)\Big] \\
 & \quad\quad -\frac{1}{\beta}D\Big(\lambda^{\downarrow}\big(R^{H_S}_{\sigma}(\rho)\otimes G(H_A)\big) \Big\Vert  \lambda^{\downarrow}\big(G(H_S)\otimes G(H_A)\big) \Big),
 \end{split}
\end{equation}
where the supremum is taken over all elements $U\in U(\mathcal{H}_S)$.

Furthermore, there exists an element in $U(\mathcal{H}_S)$ that achieves the above equality, i.e., `$\sup$'  can be replaced by  `$\max$'.
\end{Proposition}
In comparison with (\ref{passoptimHeat}) we see that we again can affect the term $D\boldsymbol{(}\lambda^{\downarrow}(\cdot) \Vert  \lambda^{\downarrow}(\cdot)\boldsymbol{)}$ by adding an ancillary Gibbs state, while the other terms [including $S\boldsymbol{(}R^{H_S}_{\sigma}(\rho)\boldsymbol{)}-S(\rho)$] remain unaffected.

\subsection{\label{Sec:cohinworkextr}Coherence in expected work extraction}
In light of Section  \ref{Sec:approxEnergymix} it perhaps comes as no surprise that the level of coherence in the energy reservoir affects our ability to extract energy from the system. Here we first consider some special cases, and then proceed to bound $S\boldsymbol{(}R^{H_S}_{\sigma}(\rho)\boldsymbol{)}-S(\rho)$ in terms of the coherence in the energy reservoir.

\subsubsection{\label{diagdiag}Diagonal system states and diagonal energy-reservoir states}

Let $H_S = \sum_{l}\tilde{h}_{l}P_{l}$ be the eigenvalue decomposition of $H_S$ where $P_{l}$ is the projector onto the eigenspace corresponding to eigenvalue $\tilde{h}_l$, where we  let $\tilde{h}_l$ denote the \emph{distinct} eigenvalues of $H_S$, i.e., $\tilde{h}_l\neq \tilde{h}_{l'}$ if $l\neq l'$ (as opposed to $h_n$ which is the complete list of eigenvalues including repetitions). We define the following operation
\begin{equation}
[Q]_{H} := \sum_{l}P_{l}QP_{l},\quad \forall Q\in L(\mathcal{H}).
\end{equation}
In words $[Q]_{H}$ removes the off-diagonal blocks of $Q$ with respect to the eigenspaces of $H$ (i.e., $[\cdot]_{H}$ is what sometimes is referred to as a `pinching').
One can rewrite the channel $R^{H_S}_{\sigma}$ in terms of the eigenprojectors $P_l$ as 
\begin{equation}
R^{H_S}_{\sigma}(Q)
=   \sum_{l,l'}P_{l}Q P_{l'}  \Tr(\Delta^{(\tilde{h}_l-\tilde{h}_{l'})/s}\sigma).
\end{equation}
By using this reformulation of $R^{H_S}_{\sigma}$ one can prove the following lemma.
\begin{Lemma}
\label{abvkjsdbv}
If $Q\in\mathcal{L}(\mathcal{H}_S)$ is such that $Q = [Q]_{H_S}$ then
\begin{equation}
R_{\sigma}^{H_S}(Q) = Q,\quad \forall \sigma\in\mathcal{S}(\mathcal{H}_E).
\end{equation}
If $\sigma$ is such that $\sigma = [\sigma]_{H_E^{(s)}}$, then 
\begin{equation}
R_{\sigma}^{H_S}(Q) = [Q]_{H_S},\quad \forall Q\in\mathcal{L}(\mathcal{H}_S).
\end{equation}
\end{Lemma}

By comparing Lemma \ref{abvkjsdbv} with Proposition \ref{Optimum} we can see that if $\rho$ already is diagonal with respect to an energy eigenbasis, then $R_{\sigma}^{H_S}(\rho) = \rho$ and thus the term $S\boldsymbol{(}R^{H_S}_{\sigma}(\rho)\boldsymbol{)} -S(\rho)$ drops out. Furthermore, we obtain the following corollary, which tells us that if the state of the energy reservoir is diagonal, then the expected energy yield can not depend on the off-diagonal elements of the initial state $\rho$. 
This implies that if the reservoir is diagonal, the expected work extraction can only depend on $[\rho]_{H_S}$, which confirms the finding in \cite{Skrzypczyk13}.
\begin{Corollary}
Let $s>0$, $\mathcal{H}_S$ and $\mathcal{H}_{A}$ be finite-dimensional and $H_S\in H_{s}(\mathcal{H}_{S})$, $H_{A}\in H_{s}(\mathcal{H}_{A})$. Let $\rho\in\mathcal{S}(\mathcal{H}_S)$. If $\sigma\in\mathcal{S}^{*}(\mathcal{H}_E)$ is such that $\sigma = [\sigma]_{H_E^{(\delta)}}$, then
\begin{equation}
\begin{split}
 & \sup_{U}\Tr\big([\hat{1}_S\otimes\hat{1}_A\otimes H_E^{(s)}]V(U)\rho\otimes G(H_A)\otimes \sigma V(U)^{\dagger}\big)  \\
 & \quad =  \Tr(H_E^{(s)}\sigma) \\
 & \quad\quad + \frac{1}{\beta}D\boldsymbol{(}[\rho]_{H_S}\Vert G(H_S)\boldsymbol{)} \\
 & \quad\quad -\frac{1}{\beta}D\Big(\lambda^{\downarrow}\big([\rho]_{H_S}\otimes G(H_A)\big) \Big\Vert  \lambda^{\downarrow}\big(G(H_S)\otimes G(H_A)\big) \Big),
\end{split}
\end{equation}
where the supremum is taken over all elements $U\in U(\mathcal{H}_S\otimes \mathcal{H}_A)$.
\end{Corollary}

\subsubsection{Bounds on $S\boldsymbol{(}R^{H_S}_{\sigma}(\rho)\boldsymbol{)} -S(\rho)$}

As seen from the previous sections, the amount of expected energy that can be extracted is partially determined by how much the effective channel $R^{H_S}_{\sigma}$ increases the entropy of the initial state $\rho$. Here we determine bounds on this quantity in terms of the coherence in the energy reservoir. 

\begin{Proposition}
\label{fdkdfbkkkdf}
Let $s>0$ and $M\in\mathbb{N}$. Let $\dim\mathcal{H}_S=N$, and let $H_S\in H_{s}(\mathcal{H}_S)$ with eigenvectors $|\psi_k\rangle$ and corresponding eigenvalues $s z_k$. Let $\sigma\in\mathcal{S}(\mathcal{H}_E)$. Let 
\begin{equation}
\boldsymbol{F} := [1 -\Tr(\Delta^{z_{k}-z_{k'}}\sigma)]_{k,k'=1}^{N},
\end{equation}
then
\begin{equation}
\sup_{\rho\in\mathcal{S}(\mathcal{H}_S)}\Vert \rho- R_{\sigma}^{H_S}(\rho) \Vert_1 \leq \Vert \boldsymbol{F}\Vert.  
\end{equation}
\end{Proposition}
\begin{proof}
By using (\ref{Rdef}) in Definition \ref{Rdefinition}
\begin{equation}
\rho - R^{H_S}_{\sigma}(\rho) =   \sum_{k,k'=1}^{N}|\psi_{k}\rangle\langle\psi_{k}|\rho|\psi_{k'}\rangle\langle\psi_{k'}| \boldsymbol{F}_{k,k'}.
\end{equation}
If we define the matrix $\boldsymbol{\rho} = [\langle\psi_{k}|\rho|\psi_{k'}\rangle]_{k,k'=1}^{N}$ it follows by equation (\ref{zuiozoui}) that  $\Vert \rho - R^{H_S}_{\sigma}(\rho)\Vert_{1} = \Vert \boldsymbol{\rho}*\boldsymbol{F}\Vert_1 \leq \Vert \boldsymbol{\rho}\Vert_1 \Vert \boldsymbol{F}\Vert 
=  \Vert \boldsymbol{F}\Vert$.
\end{proof}

\begin{Lemma}
\label{distancebound}
Let $\mathcal{H}_S$ be finite-dimensional with $N= \dim(\mathcal{H}_S)\geq 2$, and let $H_S\in H_{s}(\mathcal{H}_S)$, with eigenvalues $s z_m$. Let $|\eta_{L,l_0}\rangle := \frac{1}{\sqrt{L}}\sum_{l=0}^{L-1}|l+l_0\rangle\in\mathcal{H}_{E}$, then  
\begin{equation}
\label{vdfjkbnv}
\sup_{\rho\in\mathcal{S}(\mathcal{H}_S)}\Vert \rho- R_{|\eta_{L,l_0}\rangle\langle\eta_{L,l_0}|}^{H_S}(\rho) \Vert_1 \leq  \frac{N(z_{\mathrm{max}}-z_{\mathrm{min}})}{L},
\end{equation}
where $z_{\mathrm{max}}: = \max_{k=1}^{N}z_k$ and $z_{\mathrm{min}} := \min_{k=1}^{N}z_{k}$.
\end{Lemma}
\begin{proof}
In the case $L < z_{\mathrm{max}}- z_{\mathrm{min}}$, Eq.~(\ref{vdfjkbnv}) is trivially true. This follows since both $\rho$ and $R_{|\eta_{L,l_0}\rangle\langle\eta_{L,l_0}|}^{H_S}(\rho)$ are density operators and thus $\Vert \rho- R_{|\eta_{L,l_0}\rangle\langle\eta_{L,l_0}|}^{H_S}(\rho) \Vert_1  \leq 2 \leq N$.  
We can thus without loss of generality assume that $L \geq z_{\mathrm{max}}- z_{\mathrm{min}}$.

Due to the latter inequality we can conclude [e.g., via (\ref{etaShiftOverlap})] that $\boldsymbol{F}_{k,k'} = |z_k-z_{k'}|/L$.  By combining this observation with Proposition \ref{fdkdfbkkkdf} and Lemma \ref{abvkbja} we can conclude that  $\Vert \rho -R^{H_S}_{\sigma}(\rho)\Vert_1 \leq \Vert\boldsymbol{F}\Vert \leq N(z_{\mathrm{max}}-z_{\mathrm{min}})/L$.
\end{proof}

The following bound for for the von Neumann entropy is proved in \cite{Audenaert}.
\begin{Lemma}[From \cite{Audenaert}]
\label{continuitybound}
Let $\dim\mathcal{H} = N$ and $\rho,\sigma\in\mathcal{S}(\mathcal{H})$, then
\begin{equation}
\begin{split}
 & |S(\rho)-S(\sigma)| \leq \frac{1}{2}\Vert \rho-\sigma\Vert_1\ln(N-1) + \Xi\Big(\frac{1}{2}\Vert \rho-\sigma\Vert_1\Big),\\
 & \Xi(x):= -x\ln x -(1-x)\ln(1-x).
 \end{split}
\end{equation}
\end{Lemma}
(In \cite{Audenaert} this is formulated in terms of the base $2$ logarithm. The inequality remains the same also for other logarithms, as long as the choice of logarithm is used consistently in all quantities.)

Taking into account the fact that $\Xi$ is monotonically increasing on the interval $[0,1/2]$ and thus $\Xi(\Vert \rho-\sigma\Vert_1/2) \leq \Xi(N(z_{\mathrm{max}}-z_{\mathrm{min}})/(2L))$ if $\Vert \rho-\sigma\Vert_1 \leq N(z_{\mathrm{max}}-z_{\mathrm{min}})/L \leq 1$ we obtain the following bound by combining Lemmas \ref{distancebound} and \ref{continuitybound}.
\begin{Proposition}
\label{BoundOnLoss}
Let $s>0$ and let $\mathcal{H}_S$ be a Hilbert space with $\dim\mathcal{H}_S = N$, and let $H_S\in H_{s}(\mathcal{H}_S)$. Let $z_{\mathrm{max}}$ be the largest eigenvalue of $H_S$ divided by $s$, and let $z_{\mathrm{min}}$ be the smallest eigenvalue of $H_S$ divided by $s$.
Assume $L \geq N(z_{\mathrm{max}}-z_{\mathrm{min}})$, and define
\begin{equation}
\sigma_{L}:=|\eta_{L,l_0}\rangle\langle\eta_{L,l_0}|,\quad |\eta_{L,l_0}\rangle := \frac{1}{\sqrt{L}}\sum_{l=0}^{L-1}|l+l_0\rangle_E.
\end{equation}
Then 
\begin{equation}
\begin{split}
0 \leq S\boldsymbol{(}R^{H_S}_{\sigma_{L,l_0}}(\rho)\boldsymbol{)} -S(\rho)  \leq & \frac{z_{\mathrm{max}}-z_{\mathrm{min}}}{2L}N\ln N \\
& + \Xi\Big(\frac{N(z_{\mathrm{max}}-z_{\mathrm{min}})}{2L}\Big),
\end{split}
\end{equation}
with  $\Xi$ as in Lemma \ref{continuitybound}.
\end{Proposition}

\subsection{\label{Sec:standard}Achieving the `standard' optimal expected work extraction}
 
 Here we show that one can regain the `standard' result on expected work extraction, i.e., that the maximum yield is $kTD\boldsymbol{(}\rho\Vert G(H)\boldsymbol{)}$.  We will obtain this  as a limit case, namely of a large density of states in the energy reservoir, high degree of coherence in the energy reservoir, and very large (in Hilbert space dimension)  Gibbs distributed ancillary systems.

In order to do this we introduce a number $J$ that controls the density of states in the reservoir compared to the energy levels in $S$. More precisely, if the system Hamiltonian $H_S$ has eigenvalues that are multiples of $\delta$, then we assume that the energy spacings in the energy ladder of the reservoir is $s :=\delta/J$. To provide coherence, we assume our standard uniform superpositions   
$|\eta_{L,l_0}\rangle = \sum_{l=0}^{L-1}|l+l_0\rangle/\sqrt{L}$, and thus get a number $L$ that in some sense controls the degree of coherence in the reservoir.

To construct the above mentioned limit, we consider a fixed Hamiltonian $H_S\in H_{\delta}(\mathcal{H}_S)$, but a sequence of energy reservoirs, one for each $K\in\mathbb{N}$. We assume that the $K$th reservoir has the energy  spacing $s(K) := \delta/J(K)$ and the initial state $|\eta_{L(K),l_0}\rangle$ (where $l_0$ does not matter), with $J(K)$ and $L(K)$ being functions of $K$.
We furthermore construct a sequence of ancillary systems $A_K$ with Hilbert spaces $\mathcal{H}_A^K$ with dimension $N^{K+1}$, where $N = \dim\mathcal{H}_S$. These ancillary systems are Gibbs distributed with respect to suitably chosen Hamiltonians. 

It maybe should be noted that in these derivations we limit ourselves to states $\rho$ that have full support, $\rho\in\mathcal{S}_+(\mathcal{H}_S)$. This is essentially due to somewhat too weak bounds, which involves the logarithm of the smallest eigenvalue of $\rho$. Reasonably it should be possible to overcome this, but we will not do this here.
\begin{Corollary}
\label{LimitExtraction}
Let $N\in\mathbb{N}$, $\delta,\beta>0$ be given. Let $\mathcal{H}_{S}$ be a Hilbert space with $\dim\mathcal{H}_{S} = N$, let $\rho\in\mathcal{S}_{+}(\mathcal{H}_{S})$, and $H_S\in H_{\delta}(\mathcal{H}_{S})$ not completely degenerate. Let $L,J:\mathbb{N}\rightarrow\mathbb{N}$ be such that $\lim_{K\rightarrow+\infty} [K/J(K)] = 0$ and $\lim_{K\rightarrow+\infty} [J(K)/L(K)] = 0$.
 Then there exist a sequence of ancillary Hilbert spaces $\mathcal{H}_{A}^{K}$ with $\dim\mathcal{H}_A^{K} = N^{K+1}$, and $H_A^{K}\in H_{\delta/J(K)}(\mathcal{H}_A^{K})$ such that 
\begin{equation}
\lim_{K\rightarrow\infty}W_{K,J(K),L(K)} = \frac{1}{\beta}D\boldsymbol{(}\rho\Vert G(H_S)\boldsymbol{)},
\end{equation}
where
\begin{equation}
\begin{split}
W_{K,J(K),L(K)}  := & \sup_{U\in U(\mathcal{H}_S\otimes\mathcal{H}_{A}^{K})}\Tr( H^{\delta/J(K)}_E \tilde{\sigma}_{L(K)}) \\
&  -\Tr(H^{\delta/J(K)}_E\sigma_{L(K)}),
\end{split}
\end{equation} 
\begin{equation}
\tilde{\sigma}_{L(K)} :=   \Tr_{S,A}[V(U)\rho\otimes G(H_A^{K})\otimes \sigma_{L(K)} V(U)^{\dagger}],
\end{equation} 
and where $V(U)$ is as defined in (\ref{bijection}) for the reservoir Hamiltonian $H_E^{(\delta/J(K))}$, and
\begin{equation}
\begin{split}
 \sigma_{L(K)}   := & |\eta_{L(K),l_0}\rangle\langle\eta_{L(K),l_0}|,\\
|\eta_{L(K),l_0}\rangle  := & \frac{1}{\sqrt{L(K)}}\sum_{l=0}^{L(K)-1}|l+l_0\rangle_E,
\end{split}
\end{equation} 
for some fixed $l_0\in\mathbb{Z}$.
\end{Corollary}
This is a corollary of Proposition \ref{DenseExtraction}, and the following subsections are devoted to the proof of Proposition \ref{DenseExtraction}.

Although the specific construction we use to prove this proposition does require both increasing sizes of the ancillary systems, as well as increasing density of states in the energy reservoir, we do not prove that these actually are necessary conditions. Note that this increase in density of states is in line with what was found in \cite{Skrzypczyk13}.
Recent results \cite{Reeb13} may help to determine whether these also are necessary conditions.

\subsubsection{Shifting the path to the grid}
In Section \ref{SimulationPass} we demonstrated that one can obtain a discretized path of probability distributions, such that the sum of all consecutive relative entropies along the path can be made arbitrarily small. These probability distributions we constructed as Gibbs distributions $G(h^k)$ of underlying energy level configurations $h^k$. The problem is that these energy level configurations may not necessarily satisfy our energy matching criterion, i.e., that each $h^k_n$ should be a multiple of the energy reservoir spacing $s$. The purpose of this section is to show that we can shift the energy level configurations $h^k$ into new configurations $\tilde{h}^k$ that are multiples of $s$, but at the same time obtain a sum of consecutive relative entropies  $\sum_{k=0}^{K-1}D\big(G(\tilde{h}^{k})\big\Vert G(\tilde{h}^{k+1})\big)$ that is sufficiently close to the original $\sum_{k=0}^{K-1}D\big(G(h^{k})\big\Vert G(h^{k+1})\big)$. 

In $\mathbb{R}^{N}$, regarded as the space of energy level configurations $h$, the configurations that respect the energy matching form a square lattice (see Fig.~\ref{FigureShiftToLattice}). We construct the new configuration $\tilde{h}$ by $\tilde{h}_n = s\lfloor h_n/s \rfloor$, where $\lfloor \cdot \rfloor$ is the floor-function, which rounds of real numbers to the nearest smaller integer. (There is no particular reason to choose precisely the floor-function. We could equally well use $\lceil \cdot \rceil$, or some other method to shift the original point $h$ to somewhere close in the lattice.)

 \begin{figure}[h!]
 \center
 \includegraphics[width= 8cm]{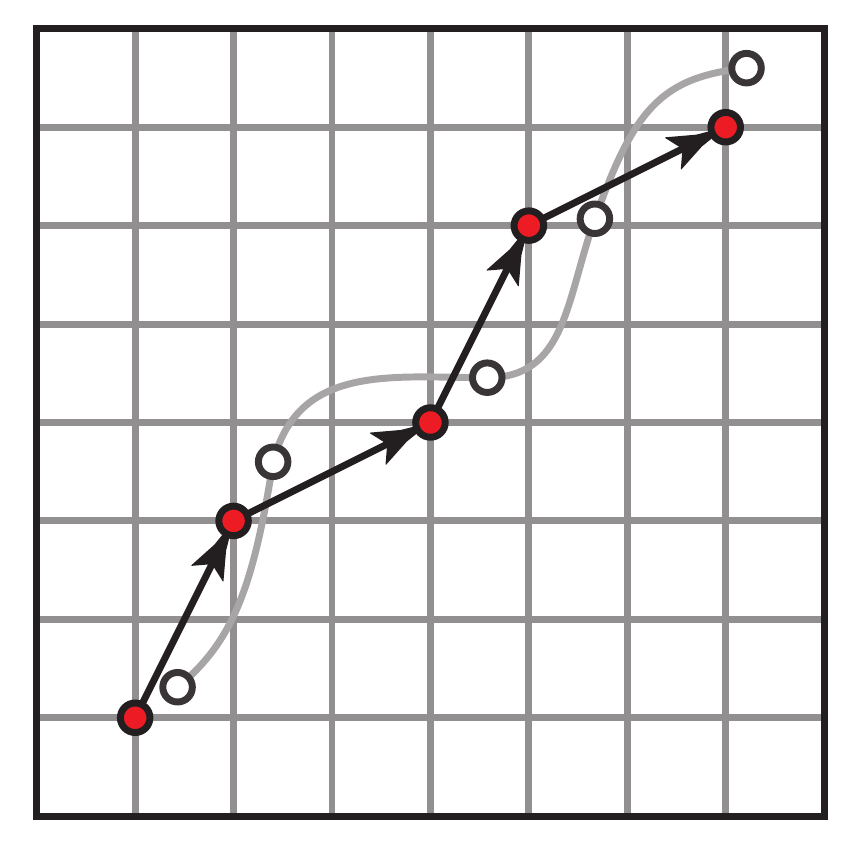} 
   \caption{\label{FigureShiftToLattice} {\bf Shifting to the lattice.}
  To have non-trivial unitary operations on the system $S$, the ancilla $A$, and the energy reservoir $E$, the energy levels of $SA$  must be `matched' to those of $E$. More precisely, we require all energy levels of $SA$ to be multiples of the level spacing in $E$, where the latter has the doubly infinite energy ladder as Hamiltonian. The thus allowed energy levels form a square grid in the space of energy level configurations (which is $\mathbb{R}^M$ for an $M$-level system). An original discretization $h^k\in\mathbb{R}^{M}$ (empty circles) of a smooth path (grey curve) can be shifted to a sequence $\tilde{h}^{k}$ on the lattice (filled circles).  Furthermore, this can be done in such a way that $\sum_{k=0}^{K-1}D\big(G(\tilde{h}^{k})\big\Vert G(\tilde{h}^{k+1})\big)$ is sufficiently close to $\sum_{k=0}^{K-1}D\big(G(h^{k})\big\Vert G(h^{k+1})\big)$. This makes it possible to adapt the results in Section \ref{SimulationPass} to the setting we consider here.
   }
\end{figure}

\begin{Lemma}
\label{NearGridPoints}
Let $N,K\in\mathbb{N}$ and $\beta,s>0$ be given. Let $(h^{k})_{k=0}^{K}$ with $h^{k}\in\mathbb{R}^{N}$.
Then there exists $(\tilde{h}^{k})_{k=0}^{K}$ with $\tilde{h}^{k}\in s \mathbb{Z}^{N}$, such that  
\begin{equation}
\begin{split}
 & \max_{n=1,\ldots,N}|\tilde{h}^{k}_n-h^{k}_n| \leq s,\quad 
D\boldsymbol{(}G(h^{k})\Vert G(\tilde{h}^k)\boldsymbol{)} \leq s \beta,\\
 & D\boldsymbol{(}G(\tilde{h}^{k})\Vert G(h^k)\boldsymbol{)} \leq s \beta,
\end{split}
\end{equation}
for $k=0,\ldots,K,$
and such that
\begin{equation}
\begin{split}
 &  \bigg|\sum_{k=0}^{K-1}D\big(G(\tilde{h}^{k})\big\Vert G(\tilde{h}^{k+1})\big) - \sum_{k=0}^{K-1}D\big(G(h^{k})\big\Vert G(h^{k+1})\big)\bigg|\\
 & \leq  2\beta s (K+1) \\
 &  \quad  + \beta(e^{s\beta}-1)K\max_{k=0,\ldots,K-1}\max_{n=1,\ldots,N}|h^{k+1}_n -h^{k}_n|.
 \end{split}
 \end{equation}
\end{Lemma}

\begin{proof}
For each $k$ let $\tilde{h}^{k}_{n} := s\lfloor \frac{1}{s}h^{k}_n \rfloor$, for $n=1,\ldots,N$. 
By construction, this means that $\tilde{h}^{k}\in s\mathbb{Z}^{N}$, and 
\begin{equation}
\label{sfbjgbnsjklf}
\begin{split}
 &  |\tilde{h}^{k}_{n}-h^{k}_n|\leq s, \quad \beta s > \ln\frac{Z(\tilde{h}^k)}{Z(h^k)} \geq 0,\\
 &  e^{ s\beta}  >   \frac{G_n(\tilde{h}^k)}{G_n(h^k)} \geq e^{- s\beta}.
\end{split}
\end{equation} 
One can moreover realize that 
\begin{equation}
\label{bhvbjvfbv}
|G_n(\tilde{h}^{k}) -G_n(h^{k})|\leq   G_n(h^{k})(e^{ s\beta}-1).
\end{equation}
By Eq.~(\ref{sfbjgbnsjklf}) it follows that $s\beta \geq D\boldsymbol{(}G(h^k)\Vert G(\tilde{h}^{k})\boldsymbol{)}$. By the analogous reasoning one finds $s\beta \geq D\boldsymbol{(}G(\tilde{h}^{k})\Vert G(h^{k})\boldsymbol{)}$.  Observe that  
\begin{equation}
\begin{split}
 &  \sum_{k=0}^{K-1}D\big(G(h^{k})\big\Vert G(h^{k+1})\big) \\
  & = \beta\sum_{k=0}^{K-1}\sum_{n=1}^{N}G_n(h^{k}) (h^{k+1}_n -h^{k}_n)  + \ln Z(h^K)-\ln Z(h^0)
\end{split}
 \end{equation}
and the analogous statement for $\tilde{h}$. By making use of the inequalities in (\ref{sfbjgbnsjklf}) we find
\begin{equation}
\begin{split}
 & \bigg|\sum_{k=0}^{K-1}D\big(G(\tilde{h}^{k})\big\Vert G(\tilde{h}^{k+1})\big) - \sum_{k=0}^{K-1}D\big(G(h^{k})\big\Vert G(h^{k+1})\big)\bigg|\\
 & \leq  \bigg|\beta\sum_{k=0}^{K-1}\sum_{n=1}^{N}G_n(\tilde{h}^{k}) (\tilde{h}^{k+1}_n -\tilde{h}^{k}_n) \\
 & \quad\quad  -\beta\sum_{k=0}^{K-1}\sum_{n=1}^{N}G_n(h^{k}) (h^{k+1}_n -h^{k}_n) \bigg| +2\beta s 
\end{split}
 \end{equation}
\begin{equation*}
\begin{split}
 = & \bigg|\beta\sum_{k=0}^{K-1}\sum_{n=1}^{N}G_n(\tilde{h}^{k}) (\tilde{h}^{k+1}_n- h^{k+1}_n   -\tilde{h}^{k}_n + h^{k}_n) \\
 &  + \beta\sum_{k=0}^{K-1}\sum_{n=1}^{N}[G_n(\tilde{h}^{k}) -G_n(h^{k})](h^{k+1}_n -h^{k}_n) \bigg|\\
 & +2\beta s \\
 \leq & 2\beta s (K+1) + \beta\sum_{k=0}^{K-1}\sum_{n=1}^{N}\Big|G_n(\tilde{h}^{k})-G_n(h^{k})\Big| |h^{k+1}_n -h^{k}_n|\\
 & [\textrm{By Eq.~(\ref{bhvbjvfbv})}] \\
\leq & 2\beta s (K+1)  + \beta(e^{ s\beta}-1)\sum_{k=0}^{K-1}\sum_{n=1}^{N} G_n(h^k) |h^{k+1}_n -h^{k}_n|\\
\leq & 2\beta s (K+1)  + \beta(e^{ s\beta}-1)K\max_{k=0}^{K-1}\max_{n=1,\ldots,N}|h^{k+1}_n -h^{k}_n|.
 \end{split}
 \end{equation*}
\end{proof}

\subsubsection{Simulation of the path}

Here we remind that $\mathbb{P}(N)$ denotes the set of probability distributions over $N$ symbols, and $\mathbb{P}_{+}(N)$ denotes the subset of probability distributions with full support.
\begin{Proposition}
\label{ozuiz}
Let $N\in\mathbb{N}$, $\delta,\beta>0$, and $q,p\in\mathbb{P}_{+}(N)$ be given. For each $K,J\in \mathbb{N}$ there exists a  $h^{K}\in \frac{\delta}{J}\mathbb{Z}^{N^{K+1}}$, such that
\begin{equation}
\begin{split}
 D\Big( \big(qG(h^{K})\big)^{\downarrow}  \Big\Vert \big(pG(h^K)\big)^{\downarrow}\Big)   \leq &  
4\beta\delta\frac{1}{J}   + 2\beta\delta\frac{K}{J}\\
 & + (e^{\frac{\delta\beta}{J}}-1)\max_{n=1,\ldots,N}\big|\ln\frac{q_n}{p_n}\big| \\
& +  \frac{1}{K}\max_{n}\big|\ln\frac{q_n}{p_n}\big|^2.
\end{split}
\end{equation}
\end{Proposition}
\begin{proof}
Since $q,p\in \mathbb{P}_{+}(N)$, it follows that $\ln q_n$ and $\ln p_n$ are well defined and finite. Thus we can define
$h^i := -\frac{1}{\beta}\ln q$, $h^f := -\frac{1}{\beta}\ln p$. 
We furthermore define $h:[0,1]\rightarrow \mathbb{R}^{N}$ by
$h(x):= (1-x)h^{i} + xh^{f}$ for $x\in[0,1]$.
The function $h(x) = (1-x)h^i + xh^f$ clearly has a continuous second derivative on each component.

For this function, let $h^{k}:= h(\frac{k}{K})$, $s:= \frac{\delta}{J}$ in Lemma \ref{NearGridPoints}. Lemma \ref{NearGridPoints} thus delivers $\tilde{h}^{k:K}:= \tilde{h}^{k}$, with 
 $\tilde{h}^{k:K}\in \frac{\delta}{J}\mathbb{Z}^{N}$, such that  
\begin{equation}
\label{pzuiozpio}
\begin{split}
 &0\leq D\big(G\big(h(0)\big)\big\Vert G\big(h^{0:K}\big)\big) \leq \frac{\delta\beta}{J},\\
 & 0\leq D\big(G\big(\tilde{h}^{K:K}\big)\big\Vert G\big(h(1)\big)\big) \leq \frac{\delta\beta}{J},
\end{split}
\end{equation}
and
\begin{equation}
\label{qtwerqrrtq}
\begin{split}
 & \bigg|\sum_{k=0}^{K-1}D\big(G(\tilde{h}^{k:K})\big\Vert G(\tilde{h}^{k+1:K})\big) \\
 & - \sum_{k=0}^{K-1}D\Big(G\big(h(\frac{k}{K})\big)\Big\Vert G\big(h(\frac{k+1}{K})\big)\Big)\bigg|\\
 \leq &  2\beta\delta\frac{K+1}{J}  + (e^{\frac{\delta\beta}{J}}-1)\max_{n=1,\ldots,N}|\ln\frac{q_n}{p_n}|.
\end{split}
 \end{equation}
 A combination of (\ref{qtwerqrrtq}) and Lemma \ref{ztioutruor} yields
 \begin{equation*}
 \begin{split}
0  \leq & \sum_{k=0}^{K-1}D\big( G(\tilde{h}^{k:K} )\big\Vert  G(\tilde{h}^{k+1:K})\big)  \\
 \leq & \bigg|\sum_{k=0}^{K-1}D\big( G(\tilde{h}^{k:K})\Big\Vert  G(\tilde{h}^{k+1:K})\big)  \\
 &   -\sum_{k=0}^{K-1}D\Big( G\big(h(\frac{k}{K})\big)\Big\Vert G\big(h(\frac{k+1}{K})\big)\Big)\bigg|   \\
 & + \sum_{k=0}^{K-1}D\Big( G\big(h(\frac{k}{K})\big)\Big\Vert G\big(h(\frac{k+1}{K})\big)\Big)\\
 \leq &  2\beta\delta\frac{K+1}{J}  \\
 &  + (e^{\frac{\delta\beta}{J}}-1)\max_{n=1,\ldots,N}|\ln\frac{q_n}{p_n}|  +  \frac{1}{K} \max_{n=1,\ldots,N}|\ln\frac{q_n}{p_n}|^2.
\end{split}
\end{equation*}
By combining this with the two inequalities in Eq.~(\ref{pzuiozpio}) we thus obtain
\begin{equation}
\label{vervcezgc}
\begin{split}
 0 \leq & D\Big(G\big(h(0)\big)\Big\Vert G\big(\tilde{h}^{0:K} \big)\Big) + \sum_{k=0}^{K-1}D\Big( G\big(\tilde{h}^{k:K} \big)\Big\Vert  G\big(\tilde{h}^{k+1:K}\big)\Big) \\
 & + D\Big(G\big(\tilde{h}^{K:K} \big) \Big\Vert G\big(h(1)\big)\Big)\\
  \leq &  4\beta\delta\frac{1}{J}   + 2\beta\delta\frac{K}{J}  \\
  & + (e^{\frac{\delta\beta}{J}}-1)\max_{n=1,\ldots,N}|\ln\frac{q_n}{p_n}|+  \frac{1}{K} \max_{n=1,\ldots,N}|\ln\frac{q_n}{p_n}|^2.
\end{split}
\end{equation}
Define
\begin{equation}
\begin{split}
 h^{K}_{n_0,\ldots,n_{K}} := &\tilde{h}^{0:K}_{n_0} + \cdots + \tilde{h}^{K:K}_{n_K} ,\\
 &  (n_0,\ldots, n_{K}) \in \{1,\ldots,N\}^{K+1}.
\end{split}
\end{equation}
Since $h^{k:K}\in \frac{\delta}{J}\mathbb{Z}^{N}$ it follows that 
\begin{equation}
h^{K} \in \frac{\delta}{J}\mathbb{Z}^{N^{K+1}}.
\end{equation}
Furthermore
\begin{equation}
\begin{split}
G_{n_0,\ldots,n_{K}}(h^{K})  =  & G_{n_0}(\tilde{h}^{0:K})\cdots \\
& \cdots G_{n_{K-1}}(\tilde{h}^{K-1:K})G_{n_K}(\tilde{h}^{K:K}).
\end{split}
\end{equation}
On $\mathbb{R}^{N^{K+2}}$ we define the  operation $\Pi$ by
\begin{equation}
(\Pi a)_{n_sn_0,n_1,\ldots,n_{K-1},n_K} := a_{n_0,n_1,\ldots,n_{K-1},n_{K},n_s}.
\end{equation}
In other words, $\Pi$ is defined via a cyclic permutation of the index of the indices of $a$. 
By Corollary \ref{jgfkgfj} it follows that 
\begin{equation}
\begin{split}
& D\Big(\big(qG(h^{K})\big)^{\downarrow}\Big\Vert \big(pG(h^{K})\big)^{\downarrow} \Big) \\
 \leq & D\Big(\Pi\big(qG(h^{K})\big)\Big\Vert pG(h^{K})\Big) \\
 = & D\Big(  \Pi\big( qG(\tilde{h}^{0:K})\cdots G(\tilde{h}^{K-1:K})G(\tilde{h}^{K:K}) \big)   \Big\Vert  \\
 & \quad\quad\quad\quad pG(\tilde{h}^{0:K})\cdots G(\tilde{h}^{K-1:K})G(\tilde{h}^{K:K})\Big) \\
 = & D\Big(G(\tilde{h}^{K:K})qG(\tilde{h}^{0:K})\cdots G(\tilde{h}^{K-1:K})\Big\Vert \\
 & \quad\quad\quad\quad pG(\tilde{h}^{0:K})\cdots G(\tilde{h}^{K-1:K})G(\tilde{h}^{K:K})\Big) \\
 = &  D\boldsymbol{(}G(\tilde{h}^{K:K})\Vert p \boldsymbol{)} + D\boldsymbol{(}q\Vert G(\tilde{h}^{0:K})\boldsymbol{)} \\
 & +\sum_{k=0}^{K-1}D\boldsymbol{(}G(\tilde{h}^{k:K})\Vert G(\tilde{h}^{k+1:K})\boldsymbol{)}.
 \end{split}
\end{equation}
By inserting (\ref{vervcezgc}) we obtain the statement of the proposition. 
\end{proof}

\subsubsection{Assembling it all}

A channel $\Gamma:L(\mathcal{H})\rightarrow L(\mathcal{H})$ is unital (or mixing enhancing) if $\Gamma(\hat{1})= \hat{1}$. For a positive operator $Q$, let $\lambda_{\mathrm{min}}(Q)$ denote the minimal eigenvalue.
\begin{Lemma}
\label{UnitalIncr}
If a channel $\Gamma$ is unital then $\lambda_{\mathrm{min}}\boldsymbol{(}\Gamma(\rho)\boldsymbol{)} \geq \lambda_{\mathrm{min}}(\rho)$ for all  $\rho\in\mathcal{S}(\mathcal{H})$.
\end{Lemma}
\begin{proof}
By \cite{Bapat}, or Theorem 2 in \cite{Chefles}, we know that for unital channels the output is more mixed than the input, i.e., $\lambda\boldsymbol{(}\Gamma(\rho)\boldsymbol{)}\prec \lambda(\rho)$. If the underlying Hilbert space has dimension $N$, this means $\sum_{k=1}^{n}\lambda^{\downarrow}_{k}\boldsymbol{(}\Gamma(\rho)\boldsymbol{)} \leq \sum_{k=1}^{n}\lambda^{\downarrow}_{k}(\rho)$ for $n=1,\ldots,N$. Since $\Gamma$ is a channel and thus trace preserving, we can conclude that $\lambda_{\mathrm{min}}\boldsymbol{(}\Gamma(\rho)\boldsymbol{)} = 1-\sum_{k=1}^{n-1}\lambda^{\downarrow}_{k}\boldsymbol{(}\Gamma(\rho)\boldsymbol{)} \geq 1-\sum_{k=1}^{n-1}\lambda^{\downarrow}_{k}(\rho)  = \lambda_{\mathrm{min}}(\rho)$.
\end{proof}

To avoid confusion in relation to Corollary \ref{LimitExtraction}, one should note that in both Proposition \ref{DenseExtraction} and \ref{nadfskbvnak}, $J$ and $L$ are not regarded as functions of $K$. Note furthermore the change of notation from $R^{H_S}_{\sigma}$ to $R^{\delta/J, H_S}_{\sigma}$. This serves to underline the fact that $R^{\delta/J, H_S}_{\sigma}$ depends on $J$, which should be kept in mind when using Proposition \ref{DenseExtraction}. The reason for this is that for a fixed density operator $\sigma = \sum_{jj'}\sigma_{jj'}|j\rangle\langle j'|$, an increase of $J$ effectively means a `decrease' of the coherence in $\sigma$ compared to $H_S$. To see this, consider a transition from $|\psi_n\rangle$ to $|\psi_{n'}\rangle$. The  corresponding change in energy, $s(z_{n'}-z_{n})$, has to be compensated for by a change of energy $s(z_n-z_{n'}) = sJ(x_n-x_{n'})$ in the reservoir. The larger the $J$, the larger the compensating jumps in the reservoir have to be (counted in numbers of eigenstates). In other words, an increase in $J$ effectively decreases the `width' of $\sigma$ in relation to $H_S$, and in this sense means a decrease in coherence.
\begin{Proposition}
\label{DenseExtraction}
Let $N\in\mathbb{N}$, $\delta,\beta>0$ be given. Let $\mathcal{H}_{S}$ be a Hilbert space with $\dim\mathcal{H}_{S} = N$, let $\rho\in\mathcal{S}_{+}(\mathcal{H}_{S})$, and $H_S\in H_{\delta}(\mathcal{H}_{S})$, and let $\sigma\in\mathcal{S}^{*}(\mathcal{H}_E)$. Then, for every $J,K\in\mathbb{N}$ there exists an ancillary Hilbert space $\mathcal{H}_{A}^{K}$ with $\dim\mathcal{H}_A^{K} = N^{K+1}$, and $H_A^{K}\in H_{\delta/J}(\mathcal{H}_A^{K})$ such that 
\begin{equation}
\begin{split}
 0   \leq & \frac{1}{\beta}D\boldsymbol{(}\rho\Vert G(H_S)\boldsymbol{)}  - \frac{1}{\beta}\Big[S\boldsymbol{(}R^{\delta/J, H_S}_{\sigma}(\rho)\boldsymbol{)}-S(\rho)\Big] - W_{J,K}(\sigma) \\
 \leq &  
4\delta\frac{1}{J}   + 2\delta\frac{K}{J} + \frac{1}{\beta}(e^{\frac{\delta\beta}{J}}-1)|\ln[\lambda_{\mathrm{min}}(\rho)\lambda_{\mathrm{min}}\boldsymbol{(}G(H_S)\boldsymbol{)}]| \\
 & +  \frac{1}{K\beta}\big|\ln[\lambda_{\mathrm{min}}(\rho)\lambda_{\mathrm{min}}\boldsymbol{(}G(H_S)\boldsymbol{)}]\big|^2, 
\end{split}
\end{equation}
where
\begin{equation}
\begin{split}
W_{J,K}(\sigma)  := & \sup_{U\in U(\mathcal{H}_S\otimes\mathcal{H}_{A}^{K})}\Tr( H^{(\delta/J)}_E \tilde{\sigma})  -\Tr(H^{(\delta/J)}_E\sigma), \\
 \tilde{\sigma}  := &  \Tr_{S,A}[V(U)\rho\otimes G(H_A^{K})\otimes \sigma V(U)^{\dagger}],
\end{split}
\end{equation} 
and where $V(U)$ is as defined in (\ref{bijection}) and $R^{\delta/J,H_S}_{\sigma}$ is as in Definition \ref{Rdefinition}, for the reservoir Hamiltonian $H_E^{(\delta/J)}$.
\end{Proposition}

\begin{proof}
Since $R^{\delta/J,H_S}_{\sigma}$ is unital (Lemma \ref{Runital}) we can, by Lemma \ref{UnitalIncr}, conclude that  
\begin{equation}
\label{nbflxg}
\lambda_{\mathrm{min}}\boldsymbol{(}R^{\delta/J,H_S}_{\sigma}(\rho)\boldsymbol{)} \geq \lambda_{\mathrm{min}}(\rho).
\end{equation}
Hence, since $\rho$ by assumption is full rank, it follows that $R^{\delta/J,H_S}_{\sigma}(\rho)$ is full rank. 
Let $q$ denote the eigenvalues of $R^{\delta/J,H_S}_{\sigma}(\rho)$, i.e., $q := \lambda\boldsymbol{(}R^{\delta/J,H_S}_{\sigma}(\rho)\boldsymbol{)}$. Let $p$ denote the eigenvalues of $G(H_S)$, i.e., $p:= \lambda\boldsymbol{(}G(H_S)\boldsymbol{)}$.  Since $\rho$ is full rank, it follows that $R^{\delta/J,H_S}_{\sigma}(\rho)$ is full rank, and thus $q\in \mathbb{P}_{+}(N)$. Furthermore $p\in \mathbb{P}_{+}(N)$ since $G(H_S)$ is a Gibbs distribution. Furthermore,
\begin{equation}
\label{yscbvy}
\begin{split}
\big|\ln\frac{q_n}{p_n}\big|  \leq & -\ln \lambda_{\mathrm{min}}\boldsymbol{(}R^{\delta/J,H_S}_{\sigma}(\rho)\boldsymbol{)} - \ln \lambda_{\mathrm{min}}\boldsymbol{(}G(H_S)\boldsymbol{)}\\
 \leq & -\ln \lambda_{\mathrm{min}}(\rho) - \ln \lambda_{\mathrm{min}}\boldsymbol{(}G(H_S)\boldsymbol{)},
\end{split}
\end{equation}
where the last inequality follows by (\ref{nbflxg}).

Since $q,p\in \mathbb{P}_{+}(N)$, we can conclude that the conditions of Proposition \ref{ozuiz} are satisfied, and we thus know that for each $K,J$ there exists a $h^{K}\in\frac{\delta}{J}\mathbb{Z}^{N^{K+1}}$ such that 
\begin{equation}
\label{bjknnbjk}
\begin{split}
& D\Big( \big(qG(h^{K})\big)^{\downarrow}  \Big\Vert \big(pG(h^K)\big)^{\downarrow}\Big)   \\
& \leq  4\beta\delta\frac{1}{J}   + 2\beta\delta\frac{K}{J}\\
 & \quad + (e^{\frac{\delta\beta}{J}}-1)|\ln[\lambda_{\mathrm{min}}(\rho)\lambda_{\mathrm{min}}\boldsymbol{(}G(H_S)\boldsymbol{)}]| \\
 & \quad +  \frac{1}{K}\big|\ln[\lambda_{\mathrm{min}}(\rho)\lambda_{\mathrm{min}}\boldsymbol{(}G(H_S)\boldsymbol{)}]\big|^2, 
\end{split}
\end{equation}
where we have used (\ref{yscbvy}).
Let $\mathcal{H}_{A}^K$ be a Hilbert space with $\dim\mathcal{H}_A^{K} = N^{K+1}$. We can construct a Hermitian operator $H_A^K$ on $\mathcal{H}_{A}^K$ that has $h^{K}$ as its eigenvalues. By construction it follows that $H_A^K\in H_{\delta/J}(\mathcal{H}_A^K)$.

One can furthermore check that 
\begin{equation}
\label{nbfkjfgbns}
\begin{split}
&  D\Big(\lambda^{\downarrow}\big(R^{\delta/J,H_S}_{\sigma}(\rho)\otimes G(H_A^K)\big)\Vert \lambda^{\downarrow}\big(G(H_S)\otimes G(H_A^K)\big)\Big)    \\
 &  = D\Big(\big(qG(h^{K})\big)^{\downarrow}\Vert\big(pG(h^K)\big)^{\downarrow}\Big).
\end{split}
\end{equation}

By assumption $H_S\in H_{\delta}(\mathcal{H}_S)$, while $H_A^K\in H_{\delta/J}(\mathcal{H}^K_A)$ and the energy reservoir has the energy spacing $s = \delta/J$. Note that 
\begin{equation}
H_{\delta}(\mathcal{H}_{S})\subset H_{\delta/J}(\mathcal{H}_{S}).
\end{equation}
 The assumption $H_S\in H_{\delta}(\mathcal{H}_S)$ implies  $H_S\in H_{\delta/J}(\mathcal{H}_{S})$ and thus allows us to apply Proposition \ref{Optimum} with $s:=\delta/J$ on the state $\rho\otimes G(H^K_A)\otimes \sigma$, which thus yields
\begin{equation}
\label{kjfbvyjkb}
\begin{split}
 & W_{J,K}(\sigma)  \\
 & =  \frac{1}{\beta}D\boldsymbol{(}\rho\Vert G(H_S)\boldsymbol{)}  - \frac{1}{\beta}\Big[S\boldsymbol{(}R^{\delta/J,H_S}_{\sigma}(\rho)\boldsymbol{)}-S(\rho)\Big] \\
 & -\frac{1}{\beta}D\Big(\lambda^{\downarrow}\big(R^{\delta/J,H_S}_{\sigma}(\rho)\otimes G(H_A^K)\big) \Big\Vert  \lambda^{\downarrow}\big(G(H_S)\otimes G(H_A^K)\big) \Big).
\end{split}
\end{equation}
By combining (\ref{bjknnbjk}), (\ref{nbfkjfgbns}), and (\ref{kjfbvyjkb}) we obtain the proposition.
\end{proof}

As was mentioned above, an increase of $J$ effectively means a decrease of the coherence in the reservoir state $\sigma$.
To guarantee that $S\boldsymbol{(}R^{\delta/J,H_S}_{\sigma}(\rho)\boldsymbol{)}-S(\rho)$ is small, we must thus choose the state $\sigma$ as to  compensate for a large $J$. In the case $\sigma = |\eta_{L,l_0}\rangle\langle\eta_{L,l_0}|$, we must thus make sure that $L \gg J$.
\begin{Proposition}
\label{nadfskbvnak}
Let $N\in\mathbb{N}$, $\delta,\beta>0$ be given. Let $\mathcal{H}_{S}$ be a Hilbert space with $\dim\mathcal{H}_{S} = N$, let $\rho\in\mathcal{S}_{+}(\mathcal{H}_{S})$, and $H_S\in H_{\delta}(\mathcal{H}_{S})$.  With $h_n^{S}$ the eigenvalues of $H_S$, we let $x_n := h^{S}_n/\delta$, and $x_{\mathrm{min}} := \min_{n}x_n$ and $x_{\mathrm{max}} := \max_{n}x_n$. Let $J,K,L\in\mathbb{N}$ be such that $L\geq NJ(x_{\mathrm{max}}-x_{\mathrm{min}})$, then there exists an ancillary Hilbert space $\mathcal{H}_{A}^{K}$ with $\dim\mathcal{H}_A^{K} = N^{K+1}$, and $H_A^{K}\in H_{\delta/J}(\mathcal{H}_A^{K})$ such that 
\begin{equation}
\label{vnakdfjbk}
\begin{split}
 0 \leq &\frac{1}{\beta}D\boldsymbol{(}\rho\Vert G(H_S)\boldsymbol{)} -W_{K,J,L}\\
 \leq  &   \frac{J}{\beta L}\frac{x_{\mathrm{max}}-x_{\mathrm{min}}}{2}N\ln N+ \frac{1}{\beta}\Xi\Big(\frac{J}{L}\frac{N(x_{\mathrm{max}}-x_{\mathrm{min}})}{2}\Big)\\
 & + 4\delta\frac{1}{J}   + 2\delta\frac{K}{J}\\
 & + \frac{1}{\beta}(e^{\frac{\delta\beta}{J}}-1)|\ln[\lambda_{\mathrm{min}}(\rho)\lambda_{\mathrm{min}}\boldsymbol{(}G(H_S)\boldsymbol{)}]| \\
 & +  \frac{1}{K\beta}\big|\ln[\lambda_{\mathrm{min}}(\rho)\lambda_{\mathrm{min}}\boldsymbol{(}G(H_S)\boldsymbol{)}]\big|^2, 
\end{split}
\end{equation}
where 
\begin{equation}
W_{K,J,L}  :=  \sup_{U\in U(\mathcal{H}_S\otimes\mathcal{H}_{A}^{K})}\Tr( H^{\delta/J}_E \tilde{\sigma}_{L})  -\Tr(H^{\delta/J}_E\sigma_L),
\end{equation} 
\begin{equation}
\tilde{\sigma}_{L} :=  \Tr_{S,A}[V(U)\rho\otimes G(H_A^{K})\otimes \sigma_L V(U)^{\dagger}]
\end{equation} 
and where $V(U)$ is as defined in (\ref{bijection}) for the reservoir Hamiltonian $H_E^{(\delta/J)}$, and
\begin{equation}
\begin{split}
 \sigma_L  := & |\eta_{L,l_0}\rangle\langle\eta_{L,l_0}|,\\
 |\eta_{L,l_0}\rangle  := & \frac{1}{\sqrt{L}}\sum_{l=0}^{L-1}|l+l_0\rangle_E,
\end{split}
\end{equation} 
for some fixed $l_0\in\mathbb{Z}$.
\end{Proposition}
By letting $J$ and $L$ be functions of the parameter $K$, we can use Proposition \ref{nadfskbvnak} to prove  Corollary \ref{LimitExtraction}. For this one can use the fact that due to $\lim_{K\rightarrow+\infty} [J(K)/L(K)] = 0$,  the inequality $L(K)\geq NJ(K)(x_{\mathrm{max}}-x_{\mathrm{min}})$ will hold for all sufficiently large $K$. This follows from the assumption that $H_S$ is not completely degenerate, and thus $x_{\mathrm{max}} > x_{\mathrm{min}}$.

\begin{proof}[Proof of Proposition \ref{nadfskbvnak}.]
The eigenvalues $h^S_n$ of $H_S$  can be expressed in terms of multiples of the energy spacing $s= \delta/J$ of the energy reservoir, as $h^S_n = z_n\delta/J$. It proves convenient to re-express $h^S_n$ in terms of multiples in $\delta$ (which we can do, due to the assumption $H_S\in H_{\delta}(\mathcal{H}_S)$). Hence, $h^S_n = x_n\delta$ and $z_{n} = x_n J$, and thus $z_{\mathrm{max}} = x_{\mathrm{max}}J$ and $z_{\mathrm{min}} = x_{\mathrm{min}}J$. (The relevant measure of energy differences, as they enter e.g., in Proposition \ref{BoundOnLoss}, is the number of rungs in the energy ladder they correspond to.)

By assumption $L\geq NJ(x_{\mathrm{max}}-x_{\mathrm{min}})$, and thus $L\geq N(z_{\mathrm{max}}-z_{\mathrm{min}})$. Hence,  Proposition \ref{BoundOnLoss} is applicable and yields
\begin{equation}
\label{vdhjfvbjdhv}
\begin{split}
0  \leq & S\boldsymbol{(}R^{\delta/J,H_S}_{\sigma_{L}}(\rho)\boldsymbol{)} -S(\rho) \\
 \leq & \frac{z_{\mathrm{max}}-z_{\mathrm{min}}}{2L}N\ln N+ \Xi\Big(\frac{N(z_{\mathrm{max}}-z_{\mathrm{min}})}{2L}\Big)\\
 = &  \frac{J}{L}\frac{x_{\mathrm{max}}-x_{\mathrm{min}}}{2}N\ln N+ \Xi\Big(\frac{JN}{L}\frac{x_{\mathrm{max}}-x_{\mathrm{min}}}{2}\Big).
 \end{split}
\end{equation}
If we let $\sigma:=\sigma_L$ in Proposition \ref{DenseExtraction} and define $W_{J,K,L} := W_{J,K}(\sigma_L)$, and combine this with (\ref{vdhjfvbjdhv}), we obtain (\ref{vnakdfjbk}).
\end{proof}

\subsection{\label{Sec:ExtractionCoherence}Expected work extraction and catalytic coherence}

Since all the derivations in this section have been based on the model introduced in Section \ref{DoublyInfinite}, i.e., the doubly infinite ladder model and the class of unitary operators given by the mapping $V(U)$, it immediately follows that the coherence resources of the energy reservoir remains unaffected when we use it. It may nevertheless be useful to see explicitly how this manifest itself in the work extraction setting. This is maybe most clearly visible in Lemmas \ref{hjbjfghbfbghj} and \ref{ndjkfbnv}, or in Proposition \ref{Optimum}. As seen, the initial state $\sigma$ of the energy reservoir only affects the final expectation value via the initial expectation value $\Tr(H_{E}^{(s)}\sigma)$ (which only depends on the diagonal elements with respect to the energy eigenbasis) and via the channel $R_{\sigma}^{H_S}$. The latter depends on off-diagonal elements of $\sigma$, but only in terms of our familiar expectation values of powers of $\Delta$, i.e., $\Tr(\Delta^a\sigma)$, as one can see from Definition \ref{Rdefinition}. From Section \ref{DoublyInfinite} we know these expectation values remain invariant under the application of the operations $V(U)$.  Hence, not only does the set of channels that can be implemented on a sequence of systems remain constant for a given initial coherence resource, but also the capacity for work extraction. Hence, to achieve the standard optimal expected work extraction for systems with off-diagonal elements, we need access to an energy reservoir with a high degree of coherence. However, once possessing it, we can utilize it indefinitely.

One can also note that the family of states $|\eta_{L,l_0}\rangle$ that we used in Section \ref{Sec:standard} to regain the standard optimal expected work extraction merely served as a technically  convenient example. We could equally well use some other family of states with a sufficiently high degree of coherence; an obvious example being the output states obtained by using  $|\eta_{L,l_0}\rangle$ for an initial work extraction on some other system.

\section{\label{Discussion} Some additional remarks}
Her we consider further relations to other concepts in the literature. 

\subsection{Why do we not encounter coherence in standard thermodynamics?}  
We have in this investigation argued that coherence is an important thermodynamic quantity. One may thus ask why we not normally encounter coherence in discussions on thermodynamics. In the limit of many independent subsystems in identical states, and identical non-interacting Hamiltonians, the global state has a low degree of `off-diagonality' with respect to the joint energy eigenspaces \cite{Brandao11,Horodecki11,Skrzypczyk13}. Hence, for such states the difference between having, or not having, access to coherence becomes negligible. Since many realistic systems can be argued to be `close' to this situation, we would thus normally not expect to see the effects of coherence. However, as demonstrated in this investigation, this is no longer true when we step outside of this `multi-copy' setting. To obtain a complete theory of quantum thermodynamics, this is arguably a necessary step.

\subsection{\label{SingleShotWorkExtr} Coherence and single-shot work extraction}
In the application of catalytic coherence to work extraction, we have focused on the expected energy gain. These results nevertheless suggest a couple of remarks in relation to single-shot work extraction. As observed in the single-shot setting \cite{Dahlsten,delRio,TrulyWorkLike,Horodecki11,Egloff,EgloffThesis,Faist,Brandao13}, the expected work content may not always correspond to ordered `work-like' energy (see discussions in \cite{TrulyWorkLike}). Although the coherence is preserved in the dynamics of the doubly-infinite ladder-model, the state nevertheless appears to gradually get more broadly distributed over the ladder (see Sec.~\ref{AnExample} for a simple example), and the entropy increases (see Sec.~\ref{Sec:inducedch}). Due to this one may suspect that the `quality' of the energy decreases. This observation maybe gets even more pronounced for the half-infinite ladder model, where we inject energy in the form of pure excited states to maintain the coherence. Since such pure excited states correspond to very ordered energy, the spreading of the distribution again suggests a degradation of the quality of the energy. (One may note that in terms of expected energy, the injection is not lost, but merely adds to the extracted energy.)
These observations raise the question whether there is a cost in terms of ordered energy (as opposed to expected energy) associated with catalytic coherence. If one would compare models that do and do not preserve coherence, would there in some sense be different costs of noise-free energy for performing operations? Such investigations would require the development of a (preferably operational) measure of ordered energy that takes into account coherence in the energy carrier \emph{per se}.

\subsection{\label{SecEmbezzling} Any relation to embezzling states?}
As discussed in Sec.~\ref{TranslationPhaseRef}, one can translate coherence into correlations using a reference system. Although this specific construction merely yields  separable states for sequential preparations (see Sec.~\ref{SeqSeparable}), one can speculate whether there exists some other type of correspondence between coherence and entanglement. There are many results on entanglement transformations that appear to echo properties of coherence transformations (see e.g. \cite{Lo97,Lo99,Harrow09}), and this is perhaps most clearly illustrated by a qualitative analogy between coherent states and embezzling states (which appears to be a bit of a quantum-information folklore).   Embezzling states \cite{vanDam03}  essentially form a special parametric family of bipartite entangled states for which the degree of entanglement goes to infinity. Analogous to how coherent states in the limit of high amplitude allow us to perform coherent operations (in the standard Jaynes-Cummings model) without changing the coherent state much, embezzling states do in the limit of a high degree of entanglement allow us to transform bipartite entangled states into each other, without changing the embezzling state much. In both cases we rely on the limit of high degree of the resource (coherence/entanglement), and in both cases the resources do appear to get degraded with use, although less and less so in the limiting case.  In view of this analogy, one may wonder whether there exists some modification of the embezzling construction that would result in genuinely catalytic transformations, much as we here have turned coherence catalytic. In this context it might be worth to keep in mind that catalytic coherence  primarily is a question of a particular type of \emph{dynamics}, rather than a special class of states.

\subsection{\label{SecEntangCat} Any relation to entanglement catalysis?}
The concept of entanglement catalysis \cite{Jonathan99} refers to the fact that bipartite state transformations by local operations and classical communication (LOCC)  \cite{Nielsen99} can be assisted by ancillary entangled states \cite{Jonathan99}. More precisely, even if it would be impossible to transform a given bipartite pure state $|\psi\rangle$ into $|\chi\rangle$ by LOCC operations, there may exist a bipartite state $|\xi\rangle$, such that $|\psi\rangle|\xi\rangle$ can be transformed into $|\chi\rangle|\xi\rangle$.  In other words, the state $|\xi\rangle$ enables the  transformation, but is not consumed in the process. Hence, the entanglement in $|\xi\rangle$ acts like a catalyst, similar to how coherence in the doubly-infinite ladder-model enables coherent operations. 

Although analogous, one can notice some crucial differences in terms of the underlying resources and allowed operations. For entanglement catalysis the relevant resource is entanglement, and the allowed set of operations are LOCC. This should be compared with catalytic coherence, where the resources are `off-diagonality' and coherence, and the allowed operations forms a special class of energy conserving unitary transformations (the $V(U)$ in Eq.~(\ref{bijection})). Moreover, in the case of entanglement catalysis the ancillary state remains intact, while catalytic coherence only preserves certain aspects of the state of the energy reservoir; the total state can change. 

 \begin{figure}
 \centering
 \includegraphics[width= 8cm]{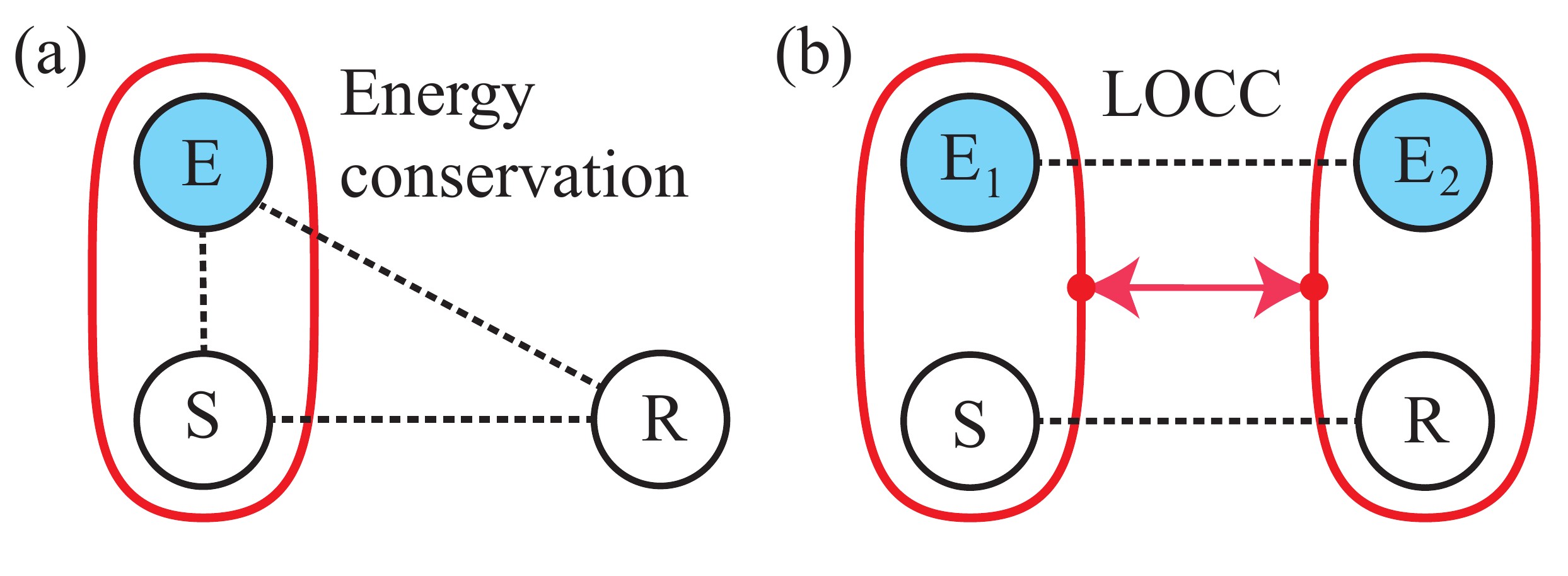} 
   \caption{\label{ComparisonCohEntang} {\bf Translated catalytic coherence versus entanglement catalysis.} (a) In the reformulation of catalytic coherence with respect to correlations with a reference $R$, the allowed operations are energy conserving, and act only on system $S$ and the energy reservoir $E$, but are not allowed to touch $R$. (b) The analogous setup for entanglement catalysis would be an entangled pair of systems $E_1$ and $E_2$, which initially is uncorrelated to $S$ and $R$, and would be restored to the same state at the end of the process. In this case we allow all LOCC operations with respect to the two subsystems $SE_1$ and $RE_2$. } 
\end{figure}

Due to the differences in settings it may not be entirely clear how to proceed with a more precise comparison. We will not consider any detailed analysis here, but merely point out that one possible method would be to use the reformulation of coherence in terms of correlations with a reference, as developed in section \ref{TranslationPhaseRef} (although one should keep in mind that this translation may not necessarily be unique). The resulting bipartition allows us to compare the resulting set of allowed operations (see Fig.~\ref{ComparisonCohEntang}). A potential complication with this method is that we may need to take into account the equivalence relation introduced in Sec.~\ref{Sec:equivalenceRelation}, i.e., not all states are distinguishable from the point of view of the allowed measurements, which may influence how we should compare operations in the two settings. (An additional technical complication is that the main results on entanglement catalysis is obtained via the majorization condition \cite{Nielsen99,Jonathan99}, which restricts the analysis to pure states, while we generally allow mixed states in catalytic coherence. Characterizations of mixed state  state entanglement transformations and catalysis appear challenging \cite{Eisert00,Gour05,Li11a,Li11b}.)

Even though the settings and assumptions are different for these two types of catalysis, one could nevertheless speculate that they can be viewed as specific instances of a more general family of catalytic phenomena, where a given dynamical restriction is accompanied with a corresponding catalytic resource. General resource theories (see, e.g., \cite{Gour08}) may be useful to investigate under which conditions this is possible.

\subsection{Any relation to the Wigner-Arkai-Yanase theorem?}
If the dynamics in a physical system is constrained by a conservation law, it follows by the  Wigner-Arkai-Yanase theorem \cite{Wigner52, Arkai60, Yanase61} that observables which do not commute with the conserved quantity can only be measured approximately. 
The quality of the approximation depends on the measurement device, e.g. to what extent it contains suitable superpositions of eigenstates of the conserved quantity. These results have been extended \cite{Ghirardi81a,Ghirardi81b,Ozawa02a,Navascues14}, generalized to the continuous variable case \cite{Busch11,Loveridge11}, as well as been analyzed in the context of the resource theory of asymmetry \cite{Marvian12,Ahmadi13}. Another version \cite{Ozawa02b,Karasawa07,Karasawa09} considers how accurately a given channel can be implemented under conservation laws. See also \cite{Barnes99,Enk01,GeaBanacloche02,GeaBanacloche05} for implementations of operations via interactions with fields. In view of the WAY theorems it is no surprise that access to a high degree of coherence yields good approximations to energy mixing unitary operations. 

The purpose of the present investigation is not primarily to quantify the accuracy of implemented measurements or channels, but to establish how coherence evolves under repeated use. Results on frame resources \cite{Bartlett06,Bartlett07b,Poulin07,Bartlett07} would intuitively suggest that an  iteration leads to a gradual degradation. However, as we have seen, one can design models where the set of channels that can be implemented remains invariant. Consequently, no matter of how we choose to quantify the quality of the implemented operations, this quality remains intact. Another way of phrasing this would be to say that the degree of accuracy in some sense is  `decoupled' from the degree of degradation; irrespective of whether the coherence resource in the reservoir allows us to implement a `high quality' coherent operation or not, that quality remains intact. 
Although the coherence thus is conserved in these models, the implementation of the channels nevertheless induces a back action on the reservoir. As speculated in Sec.~\ref{SingleShotWorkExtr} this may correspond to a loss of ordered energy. One could imagine to quantify this back action, where WAY theorems again may play a role. However, we will not consider this issue here, but leave it as an open question.

\end{appendix}

\end{document}